\documentclass{article}
\usepackage{tikz}
\usepackage{blkarray}
\usetikzlibrary{calc}
\usetikzlibrary{backgrounds}
\usepackage{mathtools}
\usepackage{hyperref}
\usepackage{float}

\input{gridsynth.sty}

\title{Optimal ancilla-free Clifford+$T$ approximation of $z$-rotations}

\author{\begin{tabular}{c}
    Neil J. Ross and Peter Selinger \\[.5ex]
    \normalsize Department of Mathematics and Statistics \\
    \normalsize Dalhousie University
  \end{tabular}
}

\date{}

\begin{document}
\maketitle

\begin{abstract}
  We consider the problem of approximating arbitrary single-qubit
  $z$-rotations by ancilla-free Clifford+$T$ circuits, up to given
  epsilon. We present a fast new probabilistic algorithm for
  solving this problem optimally, i.e., for finding the shortest
  possible circuit whatsoever for the given problem instance. The
  algorithm requires a factoring oracle (such as a quantum
  computer). Even in the absence of a factoring oracle, the algorithm
  is still near-optimal under a mild number-theoretic hypothesis. In
  this case, the algorithm finds a solution of $T$-count
  $m + O(\log(\log(1/\epsilon)))$, where $m$ is the $T$-count of the
  second-to-optimal solution. In the typical case, this yields circuit
  approximations of $T$-count $3\log_2(1/\epsilon) +
  O(\log(\log(1/\epsilon)))$. Our algorithm is efficient in practice,
  and provably efficient under the above-mentioned number-theoretic
  hypothesis, in the sense that its expected runtime is
  $O(\polylog(1/\epsilon))$.
\end{abstract}

\section{Introduction}

Practical quantum computing requires the fault-tolerant implementation
of a universal gate set. The decomposition of arbitrary unitary
operators into gates from this fixed set is then an important
problem. Most of the common error correction schemes, including most
stabilizer codes and surface codes, permit a relatively inexpensive
fault-tolerant implementation of gates from the Clifford
group. However, since Clifford gates are not universal for quantum
computation, at least one non-Clifford gate must be added to the basic
gate set to achieve universality. A common choice for this additional
gate is the $T$-gate or $\pi/8$-gate. The $T$-gate is not the only
possible extension of the Clifford group but it is considered to be
the most practical one. This is due to the availability of
fault-tolerant implementations of the $T$-gate. For this reason, the
Clifford+$T$ gate set is often considered as the most promising
candidate for practical quantum computing.

In this paper, we consider the problem of approximating arbitrary
single-qubit $z$-rotations by Clifford+$T$ circuits up to given
$\epsilon$.  Until about two years ago, the state-of-the-art algorithm
for this problem was the Solovay-Kitaev algorithm, which yields
circuits of size $O(\log^c(1/\epsilon))$, where $c>3$. At the other
end of the spectrum are algorithms based on exhaustive search. While
such algorithms achieve optimal circuit sizes, they have exponential
runtimes. For example, the algorithm of {\cite{Fowler04}} is feasible
up to $\epsilon\approx 10^{-4}$. Even the improved algorithm of
{\cite{KMM-practical}}, which combines search-based and other methods
and achieves optimal $T$-counts, still has exponential runtime which
makes it feasible only for precisions up to
$\epsilon\approx 10^{-17}$.

Within the last two years, a new generation of efficient
number-theoretic algorithms have been proposed for the approximate
synthesis problem, achieving circuit sizes of $O(\log(1/\epsilon))$
with polynomial runtime. Unlike the Solovay-Kitaev algorithm, which is
based on a geometric method of successive approximations, these new
algorithms are based on solving Diophantine equations.  The first
such algorithm was due to Kliuchnikov, Maslov, and Mosca
{\cite{KMM-approx}}. It uses a small number of ancilla qubits to
approximate a given single-qubit operator. An improved algorithm was
given in {\cite{Selinger-newsynth}}, which uses no ancillas and achieves
$T$-counts of $K+4\log_2(1/\epsilon)$ for approximating arbitrary
$z$-rotations. This compares to the information-theoretic lower bound
of $K+3\log_2(1/\epsilon)$.

In this paper, we present a fast new probabilistic algorithm for
solving the single-qubit approximate synthesis problem. Our algorithm
is optimal in an absolute sense, i.e., it finds the shortest possible
circuit whatsoever for any given problem instance. To achieve this
optimality, the algorithm requires an oracle for integer factoring. Of
course, a quantum computer can serve as such an oracle by using Shor's
algorithm {\cite{Shor}}.  But even in the absence of a factoring
oracle, our algorithm is still nearly optimal: In this case, under a
mild number-theoretic assumption, we can prove that the algorithm
finds a solution of $T$-count $m + O(\log(\log(1/\epsilon)))$, where
$m$ is the $T$-count of the second-to-optimal solution. In the typical
case, $m$ is given by the information-theoretic lower bound
$3\log_2(1/\epsilon)$. Therefore our algorithm, in the absence of a
factoring oracle, yields circuit approximations of $T$-count
$3\log_2(1/\epsilon) + O(\log(\log(1/\epsilon)))$ in the typical case.

We note that our algorithm is optimal only for the {\em specific}
problem of approximating a given $z$-rotation by a linear sequence of
single-qubit Clifford+$T$ gates. It is already known that even smaller
gate counts and/or circuit depths are achievable using additional
techniques, such as ancillas, measurements, or state distillation
{\cite{DCS12, WK13, BRS2015-1, BRS2015-2}}. In particular, the methods
of \cite{BRS2015-1} produce so-called repeat-until-success circuits
whose expected $T$-count is below the information theoretic lower
bound for deterministic circuits.

It is likely that in the future, the existence of an efficient
approximate synthesis algorithm will be considered an essential
requirement for any universal gate set proposed for practical quantum
computing, of similar importance, say, as the existence of a
fault-tolerant implementation. Our algorithm is specialized to the
Clifford+$T$ gate set which is arguably the most relevant gate set to
practical quantum computing. However, similar number-theoretic methods
are also applicable to certain other universal gate sets. For example,
Bocharov et al.\ {\cite{BGS2013}} gave an efficient synthesis
algorithm for the Clifford+$V$ gate set that achieves gate counts
linear in $\log(1/\epsilon)$, and Kliuchnikov et al.\ {\cite{KBS2013}}
did the same for the $\pair{{\cal F}, {\cal T}}$ gate set. One may
reasonably expect these gate sets to be amenable to the same kind of
optimal synthesis that we provide here for the Clifford+$T$ gate set.

\section{Overview}

Recall that the single-qubit Clifford group is generated by the
Hadamard gate $H$, the phase gate $S$, and the scalar
$\omega=e^{i\pi/4}$. By adding the non-Clifford operator $T$, one
obtains a universal gate set for quantum computing.
\[ 
  \omega = e^{i\pi/4}, \quad
  H = \frac{1}{\sqrt2}\zmatrix{cc}{1&1\\1&-1}, \quad
  S = \zmatrix{cc}{1&0\\0&i}, \quad
  T = \zmatrix{cc}{1&0\\0&e^{i\pi/4}}.
\]
Our goal is to approximate an arbitrary $z$-rotation 
\[ \Rz(\theta) = e^{-i\theta Z/2} = \zmatrix{cc}{e^{-i\theta/2} & 0 \\
  0 & e^{i\theta/2}}
\]
by a Clifford+$T$ operator up to given $\epsilon>0$. By a result of
Kliuchnikov, Maslov, and Mosca {\cite{Kliuchnikov-etal}}, a unitary
$2\times 2$-operator can be exactly written as a product of
Clifford+$T$ operators if and only if all of its matrix entries belong to
the ring $\D[\omega] = \Z[\frac{1}{\sqrt2},i]$. Our strategy is therefore
to approximate $\Rz(\theta)$ by a unitary operator of the form
\[ U = \zmatrix{cc}{u & -t\da \\ t & u\da},
\]
where $u,t\in\D[\omega]$. This problem can be solved in two stages:
\begin{enumerate}\roundlabels
\item\label{item-stage-1} find a suitable candidate $u\in\D[\omega]$ that is a good
  approximation of $e^{-i\theta/2}$;
\item\label{item-stage-2} solve the Diophantine
  equation $u\da u+t\da t=1$, to ensure that $U$ is unitary.
\end{enumerate}
Problem (\ref{item-stage-2}) can be solved by standard
number-theoretic methods. In the interest of self-containedness, we
summarize these methods in Section~\ref{sec-diophantine} and
Appendix~\ref{app-rings}. In general, solving the Diophantine equation
in (\ref{item-stage-2}) requires the ability to factor large
integers. However, if no efficient factoring method is available, then
(modulo a mild number-theoretic assumption), the Diophantine equation
can still be solved with large enough probability to ensure that at
most $O(\log(1/\epsilon))$ candidates need to be tried.

The main new technical innovation of this paper is a new and optimal
solution to problem (\ref{item-stage-1}). It turns out that the
``suitability'' of a candidate $u$ can be expressed as a problem of
the form $u\in A$ and $u\bul\in B$, where $A$ and $B$ are fixed convex
subsets of the complex plane depending only on $\theta$ and
$\epsilon$, and $(-)\bul$ is the automorphism of the ring $\D[\omega]$
obtained by mapping $\sqrt2$ to $-\sqrt2$. We call such a problem a
{\em two-dimensional grid problem}. In Sections~\ref{sec-grid-1d} and
{\ref{sec-grid-2d}}, we formulate a general algorithm for solving one-
and two-dimensional grid problems efficiently. Our algorithm proceeds
by transforming the sets $A$ and $B$ using invertible linear operators
that preserve solutions to the given problem, which we call {\em grid
  operators}. We show that for any two convex sets $A$ and $B$ one can
construct a grid operator $G$ such that the action of $G$ on $A$ and
$B$ normalizes the grid problem in an appropriate sense. The main
technical ingredient that makes our solution efficient is an iterative
process for constructing $G$, which is detailed in
Appendix~\ref{appendix-skew-red}.

The rest of this paper is organized as follows. In
Section~\ref{sec-algebra}, we review some basic notions from
algebra. We discuss one- and two-dimensional grid problems in
Sections~\ref{sec-grid-1d} and {\ref{sec-grid-2d}}, respectively. In
Section~\ref{sec-diophantine}, we show how to solve the relevant
Diophantine equation. A detailed description of the main synthesis
algorithm is given in Section~\ref{sec-algorithm}. In
Section~\ref{sec-analysis}, we analyze the algorithm's correctness,
optimality, and complexity. Some experimental results are given in
Section~\ref{sec-experimental}. For better readability of the main
body of the paper, certain technical results are proved in the
appendices.

\section{Some algebra}
\label{sec-algebra}

We introduce some notation and algebraic prerequisites. The set of
natural numbers, including 0, is denoted by $\N$, the ring of integers
is denoted by $\Z$, and we let $\omega=e^{i\pi /4}=(1+i)/\sqrt 2$.

\begin{definition}
\emph{(Extensions of $\Z$)}
We are interested in the following rings of algebraic integers: 
\begin{itemize}
  \item $\Z[\sqrt 2]=\s{a+b\sqrt 2\such a,b\in\Z}$, the ring of 
  \emph{quadratic integers with radicand 2}; 
  \item $\Z[\omega]=\s{a\omega^3 +b\omega^2 +c\omega +d 
  \such a,b,c,d \in \Z}$, the ring of \emph{cyclotomic integers 
  of degree 8};  
  \item $\D=\Z[\frac{1}{2}]=\s{\frac{a}{2^k}\such a\in \Z, k\in\N}$,
  the ring of \emph{dyadic fractions};
  \item $\D[\sqrt 2]= \Z[\rtt{}]=\s{a+b\sqrt{2}\such a,b\in\D}$; and 
  \item $\D[\omega]= \Z[\rtt{}, i]=\s{a\omega^3 +b\omega^2 +c\omega +d
  \such a,b,c,d \in \D}$.
\end{itemize}
\end{definition}

We have the inclusions $\Z\subseteq \Z[\sqrt 2]\subseteq \Z[\omega]$ 
and $\D\subseteq \D[\sqrt 2]\subseteq \D[\omega]$. Moreover, 
$\Z\subseteq\D$, $\Z[\sqrt 2]\subseteq \D[\sqrt 2]$, and 
$\Z[\omega]\subseteq \D[\omega]$. Finally, $\Z[\sqrt 2]$ and 
$\Z[\omega]$ are dense in $\R$ and $\C$, respectively. 

\begin{definition}
\emph{(Automorphisms)}
The following maps are automorphisms of $\D[\omega]$:
\begin{itemize} 
  \item \emph{Complex conjugation}, which we denote $(-)^\dagger$, 
  acts on an arbitrary element of $\D[\omega]$ or $\Z[\omega]$ as 
  follows:
  \[
  (a\omega^3 +b\omega^2 +c\omega +d)^\dagger = 
  -c\omega^3 -b\omega^2 -a\omega +d .
  \]
  \item  \emph{$\sqrt{2}$-conjugation}, which we denote 
  $(-)^\bullet$, acts on an arbitrary element of $\D[\omega]$ or $\Z[\omega]$
  as follows:
  \[
  (a\omega^3 +b\omega^2 +c\omega +d)^\bullet = 
  -a\omega^3 +b\omega^2 -c\omega +d
  \]
\end{itemize}
\end{definition}

The action of $(-)^\bullet$ on an element of $\D[\sqrt 2]$ or
$\Z[\sqrt 2]$ is given by $(a+b\sqrt{2})^\bullet = a-b\sqrt{2}$. In
particular, this implies that if $t=a+b\sqrt{2}$ is an element of $\D[\sqrt 2]$
(resp.\ $\Z[\sqrt 2]$), then $t^\bullet t=a^2-2b^2$ is an element of $\D$
(resp.\ $\Z$).

\begin{remark}
If $\alpha$ and $\beta$ are two distinct elements of $\Z[\sqrt 2]$, 
then the following inequality holds:
  \begin{equation}\label{eqn-discrete}
    |\alpha-\beta|\cdot |\alpha\bul-\beta\bul| \geq 1,
  \end{equation}
This follows from the fact that for $t\in\Z[\sqrt 2]$, 
$t^\bullet t$ is an integer and $t\bul t=0$ if and only if $t=0$.
\end{remark}  

\begin{definition}
\label{def-denomexp}
Let $t\in \D[\omega]$ and $k\in\N$. If $\rt{k}t\in\Z[\omega]$, 
then we say that $k$ is a \emph{denominator exponent} of $t$. 
The smallest such $k\geq 0$ is the \emph{least denominator exponent} 
of $t$. For $k\in\N$, the elements of $\D[\sqrt 2]$ (resp.\ $\D[\omega]$) 
having $k$ as a denominator exponent form a ring, denoted 
$\rtt{k}\Z[\sqrt 2]$ (resp.\ $\rtt{k}\Z[\omega]$).
\end{definition}

\begin{definition}
We frequently refer to the following elements of $\Z[\sqrt 2]$ and 
$\Z[\omega]$:
\begin{itemize}
  \item $\lambda=1+\sqrt{2} \in \Z[\sqrt 2]$ and
  \item $\delta = 1+\omega\in\Z[\omega]$.
\end{itemize}
\end{definition}

\begin{remark}
  The element $\lambda$ is invertible in the ring $\Z[\sqrt 2]$, with
  inverse $\lambda\inv = -1+\sqrt{2}=-\lambda\bul$. The element
  $\delta$ of the ring $\Z[\omega]$ satisfies
  $\delta^2=\lambda\omega\sqrt{2}$ and
  $\delta^{\dagger}\delta=\lambda\sqrt{2}$.
\end{remark}

\section{One-dimensional grid problems}
\label{sec-grid-1d}

\begin{definition}
  Let $B$ be a set of real numbers. The {\em (real) grid} for $B$ is
  the set
  \begin{equation}
    \grid(B) = \s{ \alpha\in\Z[\sqrt 2] \mid \alpha\bul \in B}.
  \end{equation}
  When $B$ is clear from the context, we refer to the elements of this
  set as {\em grid points}.
\end{definition}

In the following, we will only be interested in the case where $B$ is
a closed interval $[y_0,y_1]$ with $y_0<y_1$. In this case, the grid
is discrete and infinite. It is discrete because the distance between
grid points is bounded below by {\eqref{eqn-discrete}}. And it is
infinite by the density of $\Z[\sqrt 2]$: there are infinitely many
points $\beta\in B\cap \Z[\sqrt 2]$, and for each of them, $\beta\bul$
is a grid point.

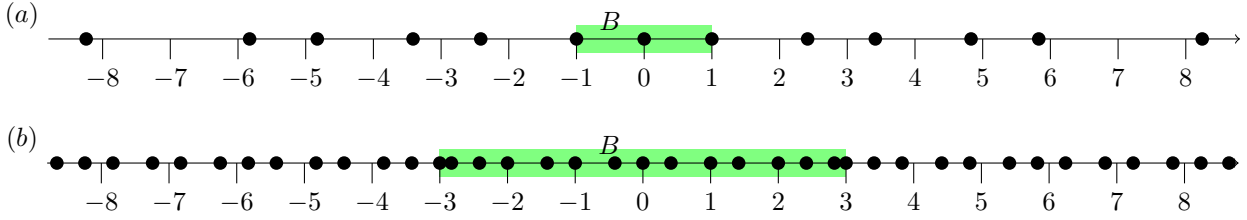
\begin{figure}
  \[ (a)~
  \mp{0.8}{\begin{tikzpicture}[scale=0.9]
\fill[thick,color=green!50] (-1.0,0.2) -- (1.0,0.2) -- (1.0,-0.2) -- (-1.0,-0.2);
\draw[->] (-8.8,0) -- (8.8, 0);
\draw (-8,0) -- (-8,-0.3) node[below] {$-8$};
\draw (-7,0) -- (-7,-0.3) node[below] {$-7$};
\draw (-6,0) -- (-6,-0.3) node[below] {$-6$};
\draw (-5,0) -- (-5,-0.3) node[below] {$-5$};
\draw (-4,0) -- (-4,-0.3) node[below] {$-4$};
\draw (-3,0) -- (-3,-0.3) node[below] {$-3$};
\draw (-2,0) -- (-2,-0.3) node[below] {$-2$};
\draw (-1,0) -- (-1,-0.3) node[below] {$-1$};
\draw (0,0) -- (0,-0.3) node[below] {$0$};
\draw (1,0) -- (1,-0.3) node[below] {$1$};
\draw (2,0) -- (2,-0.3) node[below] {$2$};
\draw (3,0) -- (3,-0.3) node[below] {$3$};
\draw (4,0) -- (4,-0.3) node[below] {$4$};
\draw (5,0) -- (5,-0.3) node[below] {$5$};
\draw (6,0) -- (6,-0.3) node[below] {$6$};
\draw (7,0) -- (7,-0.3) node[below] {$7$};
\draw (8,0) -- (8,-0.3) node[below] {$8$};
\fill (-8.242640687119286,0) circle (.1);
\fill (-5.82842712474619,0) circle (.1);
\fill (-4.82842712474619,0) circle (.1);
\fill (-3.414213562373095,0) circle (.1);
\fill (-2.414213562373095,0) circle (.1);
\fill (-1.0,0) circle (.1);
\fill (0.0,0) circle (.1);
\fill (1.0,0) circle (.1);
\fill (2.414213562373095,0) circle (.1);
\fill (3.414213562373095,0) circle (.1);
\fill (4.82842712474619,0) circle (.1);
\fill (5.82842712474619,0) circle (.1);
\fill (8.242640687119286,0) circle (.1);
\draw (-0.5,0) node[above] {$B$};\end{tikzpicture}}
  \]
  \[ (b)~
  \mp{0.8}{\begin{tikzpicture}[scale=0.9]
\fill[thick,color=green!50] (-3.0,0.2) -- (3.0,0.2) -- (3.0,-0.2) -- (-3.0,-0.2);
\draw[->] (-8.8,0) -- (8.8, 0);
\draw (-8,0) -- (-8,-0.3) node[below] {$-8$};
\draw (-7,0) -- (-7,-0.3) node[below] {$-7$};
\draw (-6,0) -- (-6,-0.3) node[below] {$-6$};
\draw (-5,0) -- (-5,-0.3) node[below] {$-5$};
\draw (-4,0) -- (-4,-0.3) node[below] {$-4$};
\draw (-3,0) -- (-3,-0.3) node[below] {$-3$};
\draw (-2,0) -- (-2,-0.3) node[below] {$-2$};
\draw (-1,0) -- (-1,-0.3) node[below] {$-1$};
\draw (0,0) -- (0,-0.3) node[below] {$0$};
\draw (1,0) -- (1,-0.3) node[below] {$1$};
\draw (2,0) -- (2,-0.3) node[below] {$2$};
\draw (3,0) -- (3,-0.3) node[below] {$3$};
\draw (4,0) -- (4,-0.3) node[below] {$4$};
\draw (5,0) -- (5,-0.3) node[below] {$5$};
\draw (6,0) -- (6,-0.3) node[below] {$6$};
\draw (7,0) -- (7,-0.3) node[below] {$7$};
\draw (8,0) -- (8,-0.3) node[below] {$8$};
\fill (-8.65685424949238,0) circle (.1);
\fill (-8.242640687119286,0) circle (.1);
\fill (-7.82842712474619,0) circle (.1);
\fill (-7.242640687119286,0) circle (.1);
\fill (-6.82842712474619,0) circle (.1);
\fill (-6.242640687119286,0) circle (.1);
\fill (-5.82842712474619,0) circle (.1);
\fill (-5.414213562373095,0) circle (.1);
\fill (-4.82842712474619,0) circle (.1);
\fill (-4.414213562373095,0) circle (.1);
\fill (-3.8284271247461903,0) circle (.1);
\fill (-3.414213562373095,0) circle (.1);
\fill (-3.0,0) circle (.1);
\fill (-2.8284271247461903,0) circle (.1);
\fill (-2.414213562373095,0) circle (.1);
\fill (-2.0,0) circle (.1);
\fill (-1.4142135623730951,0) circle (.1);
\fill (-1.0,0) circle (.1);
\fill (-0.41421356237309515,0) circle (.1);
\fill (0.0,0) circle (.1);
\fill (0.41421356237309515,0) circle (.1);
\fill (1.0,0) circle (.1);
\fill (1.4142135623730951,0) circle (.1);
\fill (2.0,0) circle (.1);
\fill (2.414213562373095,0) circle (.1);
\fill (2.8284271247461903,0) circle (.1);
\fill (3.0,0) circle (.1);
\fill (3.414213562373095,0) circle (.1);
\fill (3.8284271247461903,0) circle (.1);
\fill (4.414213562373095,0) circle (.1);
\fill (4.82842712474619,0) circle (.1);
\fill (5.414213562373095,0) circle (.1);
\fill (5.82842712474619,0) circle (.1);
\fill (6.242640687119286,0) circle (.1);
\fill (6.82842712474619,0) circle (.1);
\fill (7.242640687119286,0) circle (.1);
\fill (7.82842712474619,0) circle (.1);
\fill (8.242640687119286,0) circle (.1);
\fill (8.65685424949238,0) circle (.1);
\draw (-0.5,0) node[above] {$B$};\end{tikzpicture}}
  \]
  \caption{The real grid for two different intervals $B$. In both
    cases, the interval $B$ is shown in green, and grid points are
    shown as black dots.}
  \label{fig-grid-1d}
  \rule{\textwidth}{0.1mm}
\end{figure}

\begin{example}
  Figure~\ref{fig-grid-1d} illustrates the grids for the intervals
  $[-1,1]$ and $[-3,3]$, respectively. For example, the first few
  non-negative points in $\grid([-1,1])$ are $0$, $1$, $1+\sqrt2$,
  $2+\sqrt2$, $2+2\sqrt2$, $3+2\sqrt 2$, and $4+3\sqrt 2$. As one
  would expect, the grid for $[-3,3]$ is about three times denser than
  that for $[-1,1]$.  We also note that $B\seq B'$ implies
  $\grid(B)\seq\grid(B')$.
\end{example}

\begin{definition}
  Let $A$ and $B$ be sets of real numbers.  The {\em one-dimensional
    grid problem} for $A$ and $B$ is the following:
  \begin{equation}\label{eqn-grid-constraint}
    \mbox{{\bf One-dimensional grid problem:} Find $\alpha\in\Z[\sqrt
      2]$ satisfying $\alpha \in A$ and $\alpha\bul\in B$.}
  \end{equation}
\end{definition}

Note that {\eqref{eqn-grid-constraint}} can be equivalently written as
$\alpha \in A \cap \grid(B)$. In other words, the grid problem is to
find points in some given set $A$ that belong to the grid for $B$. We
also refer to the conditions $\alpha\in A$ and $\alpha\bul\in B$ as
{\em grid constraints}.

In the case where $A$ and $B$ are finite intervals, the grid problem
is guaranteed to have a finite number of solutions.  We recall the
following facts from {\cite{Selinger-newsynth}}:

\begin{lemma}\label{lem-grid-bounds}
  Let $A=[x_0,x_1]$ and $B=[y_0,y_1]$ be closed real intervals, such
  that $x_1-x_0=\delta$ and $y_1-y_0=\Delta$. If $\delta\Delta < 1$,
  then the grid problem {\eqref{eqn-grid-constraint}} has at most one
  solution. If $\delta\Delta \geq (1+\sqrt2)^2$, then the grid problem
  {\eqref{eqn-grid-constraint}} has at least one solution.
\end{lemma}

\begin{proof}
  Lemmas~16 and 17 of {\cite{Selinger-newsynth}}.
\end{proof}

\begin{proposition}\label{prop-algorithm-1d}
  There is an algorithm for enumerating all solutions of the
  one-dimensional grid problem for closed intervals $A=[x_0,x_1]$ and
  $B=[y_0,y_1]$. Moreover, the algorithm is efficient in the sense
  that it only requires a constant number of arithmetic operations per
  solution produced. 
\end{proposition}

\begin{proof}
  It was already noted in {\cite[Lemma~17]{Selinger-newsynth}} that there
  is an efficient algorithm for computing one solution. To see that we
  can efficiently enumerate all solutions, let $\delta=x_1-x_0$ and
  $\Delta=y_1-y_0$ as before.  Recall that $\lambda=1+\sqrt2$ and that 
  $\lambda\inv=-\lambda\bul$. The grid problem for the sets $A$ and
  $B$ is equivalent to the grid problem for $\lambda\inv A$ and
  $-\lambda B$, because $\alpha\in A$ and $\alpha\bul\in B$ hold if
  and only if $\lambda\inv\alpha\in\lambda\inv A$ and
  $(\lambda\inv\alpha)\bul\in-\lambda B$. Using such rescaling, we may
  without loss of generality assume that $\lambda\inv\leq \delta<1$.

  Now consider any solution $\alpha = a+b\sqrt 2\in\Z[\sqrt2]$. From
  $\alpha\in[x_0,x_1]$, we know that $x_0-b\sqrt 2 \leq a \leq
  x_1-b\sqrt 2$. But since $x_1-x_0<1$, it follows that for any
  $b\in\Z$, there is at most one $a\in\Z$ yielding a
  solution. Moreover, we note that $b=(\alpha-\alpha\bul)/\rt{3}$, so
  that any solution satisfies $(x_0-y_1)/\rt{3} \leq b \leq
  (x_1-y_0)/\rt{3}$. The algorithm then proceeds by enumerating all
  the integers $b$ in the interval $[(x_0-y_1)/\rt{3},
  (x_1-y_0)/\rt{3}]$. For each such $b$, find the unique integer $a$
  (if any) in the interval $[x_0-b\sqrt 2, x_1-b\sqrt 2]$. Finally,
  check if $a+b\sqrt2$ is a solution. The runtime is
  governed by the number of $b\in\Z$ that need to be checked, of which
  there are at most $O(y_1-y_0) = O(\delta\Delta)$. As a consequence
  of Lemma~\ref{lem-grid-bounds}, the total number of solutions is
  at least $\Omega(\delta\Delta)$, and so the algorithm is efficient.
\end{proof}

\begin{remark}
  For the purposes of this paper, by an {\em arithmetic operation} we
  mean addition, subtraction, multiplication, division,
  exponentiation, and logarithm.
\end{remark}

\begin{remark}\label{rem-interval-specify}
  Since the inputs to the algorithm are real intervals, if we were to
  give a rigorous complexity-theoretic account, we should also clarify
  how these intervals are specified (for example, with rational
  endpoints, endpoints as computable real numbers, etc.). For our
  purposes, the manner in which an interval $A=[x_0,x_1]$ is specified
  as an input to the algorithm does not matter very much; it would be
  sufficient, for example, to assume that we are given rational bounds
  $a,b$ with $a< x_0 < x_1 < b$, such that $b-a$ exceeds $x_1-x_0$ by
  at most a fixed constant factor, as well as a procedure for deciding
  whether any given point of $\D[\sqrt 2]$ is in $A$ or not.
\end{remark}

\section{Two-dimensional grid problems}
\label{sec-grid-2d}

Recall that $\Z[\omega]$ is a subset of the complex numbers. In what
follows, it is often convenient to identify the complex numbers with
the Euclidean plane $\R^2$, so we will often interchangeably regard
$\Z[\omega]$ as a subset of $\C$ and of $\R^2$.

\begin{definition}
  Let $B$ be a subset of $\R^2$. The {\em (complex) grid} for $B$ is
  the set
  \begin{equation}
    \Grid(B) = \s{ u\in\Z[\omega] \mid u\bul \in B}.
  \end{equation}
\end{definition}

We will only be interested in the case where $B$ is a bounded convex
set with non-empty interior. In this case, the grid is discrete and
infinite, just as in the one-dimensional case.

\begin{figure}
  \[ (a)~
  \mp{0.95}{\begin{tikzpicture}[scale=0.89]
\fill[color=green!50] (1.001,1.001) -- (-1.001,1.001) -- (-1.001,-1.001) -- (1.001,-1.001) -- cycle;
\draw[->] (-4.5,0) -- (4.5, 0);
\draw (-4,0) -- (-4,-0.1) node[below] {\small $-4$};
\draw (-3,0) -- (-3,-0.1) node[below] {\small $-3$};
\draw (-2,0) -- (-2,-0.1) node[below] {\small $-2$};
\draw (-1,0) -- (-1,-0.1) node[below] {\small $-1$};
\draw (0,0) -- (0,-0.1) node[below] {\small $0$};
\draw (1,0) -- (1,-0.1) node[below] {\small $1$};
\draw (2,0) -- (2,-0.1) node[below] {\small $2$};
\draw (3,0) -- (3,-0.1) node[below] {\small $3$};
\draw (4,0) -- (4,-0.1) node[below] {\small $4$};
\draw[->] (0,-4.5) -- (0,4.5);
\draw (0,-4) -- (-0.1,-4) node[left] {\small $-4$};
\draw (0,-3) -- (-0.1,-3) node[left] {\small $-3$};
\draw (0,-2) -- (-0.1,-2) node[left] {\small $-2$};
\draw (0,-1) -- (-0.1,-1) node[left] {\small $-1$};
\draw (0,0) -- (-0.1,0) node[left] {\small $0$};
\draw (0,1) -- (-0.1,1) node[left] {\small $1$};
\draw (0,2) -- (-0.1,2) node[left] {\small $2$};
\draw (0,3) -- (-0.1,3) node[left] {\small $3$};
\draw (0,4) -- (-0.1,4) node[left] {\small $4$};
\fill (-3.414213562373095,-3.414213562373096) circle (.1);
\fill (-3.414213562373095,-2.4142135623730954) circle (.1);
\fill (-3.4142135623730954,-1.0) circle (.1);
\fill (-3.4142135623730954,0.00000000000000011102230246251565) circle (.1);
\fill (-3.4142135623730954,1.0000000000000004) circle (.1);
\fill (-3.4142135623730954,2.414213562373096) circle (.1);
\fill (-3.4142135623730954,3.414213562373096) circle (.1);
\fill (-2.414213562373095,-3.414213562373096) circle (.1);
\fill (-2.414213562373095,-2.4142135623730954) circle (.1);
\fill (-2.4142135623730954,-1.0) circle (.1);
\fill (-2.4142135623730954,0.00000000000000011102230246251565) circle (.1);
\fill (-2.4142135623730954,1.0000000000000004) circle (.1);
\fill (-2.4142135623730954,2.414213562373096) circle (.1);
\fill (-2.4142135623730954,3.414213562373096) circle (.1);
\fill (-0.9999999999999999,-3.4142135623730954) circle (.1);
\fill (-0.9999999999999999,-2.4142135623730954) circle (.1);
\fill (-1.0,-1.0000000000000002) circle (.1);
\fill (-1.0,0.0) circle (.1);
\fill (-1.0,1.0000000000000002) circle (.1);
\fill (-1.0,2.4142135623730954) circle (.1);
\fill (-1.0,3.4142135623730954) circle (.1);
\fill (0.00000000000000011102230246251565,-3.4142135623730954) circle (.1);
\fill (0.00000000000000011102230246251565,-2.4142135623730954) circle (.1);
\fill (0.0,-1.0000000000000002) circle (.1);
\fill (0.0,0.0) circle (.1);
\fill (0.0,1.0000000000000002) circle (.1);
\fill (-0.00000000000000011102230246251565,2.4142135623730954) circle (.1);
\fill (-0.00000000000000011102230246251565,3.4142135623730954) circle (.1);
\fill (1.0,-3.4142135623730954) circle (.1);
\fill (1.0,-2.4142135623730954) circle (.1);
\fill (1.0,-1.0000000000000002) circle (.1);
\fill (1.0,0.0) circle (.1);
\fill (1.0,1.0000000000000002) circle (.1);
\fill (0.9999999999999999,2.4142135623730954) circle (.1);
\fill (0.9999999999999999,3.4142135623730954) circle (.1);
\fill (2.4142135623730954,-3.414213562373096) circle (.1);
\fill (2.4142135623730954,-2.414213562373096) circle (.1);
\fill (2.4142135623730954,-1.0000000000000004) circle (.1);
\fill (2.4142135623730954,-0.00000000000000011102230246251565) circle (.1);
\fill (2.4142135623730954,1.0) circle (.1);
\fill (2.414213562373095,2.4142135623730954) circle (.1);
\fill (2.414213562373095,3.414213562373096) circle (.1);
\fill (3.4142135623730954,-3.414213562373096) circle (.1);
\fill (3.4142135623730954,-2.414213562373096) circle (.1);
\fill (3.4142135623730954,-1.0000000000000004) circle (.1);
\fill (3.4142135623730954,-0.00000000000000011102230246251565) circle (.1);
\fill (3.4142135623730954,1.0) circle (.1);
\fill (3.414213562373095,2.4142135623730954) circle (.1);
\fill (3.414213562373095,3.414213562373096) circle (.1);
\fill (-4.121320343559643,-4.121320343559644) circle (.1);
\fill (-4.121320343559643,-1.7071067811865477) circle (.1);
\fill (-4.121320343559643,-0.7071067811865475) circle (.1);
\fill (-4.121320343559643,0.7071067811865478) circle (.1);
\fill (-4.121320343559643,1.7071067811865483) circle (.1);
\fill (-4.121320343559643,4.121320343559644) circle (.1);
\fill (-1.7071067811865475,-4.121320343559643) circle (.1);
\fill (-1.7071067811865475,-1.707106781186548) circle (.1);
\fill (-1.7071067811865475,-0.7071067811865476) circle (.1);
\fill (-1.7071067811865477,0.7071067811865477) circle (.1);
\fill (-1.7071067811865477,1.707106781186548) circle (.1);
\fill (-1.707106781186548,4.121320343559644) circle (.1);
\fill (-0.7071067811865475,-4.121320343559643) circle (.1);
\fill (-0.7071067811865476,-1.707106781186548) circle (.1);
\fill (-0.7071067811865476,-0.7071067811865476) circle (.1);
\fill (-0.7071067811865477,0.7071067811865477) circle (.1);
\fill (-0.7071067811865477,1.707106781186548) circle (.1);
\fill (-0.7071067811865478,4.121320343559644) circle (.1);
\fill (0.7071067811865478,-4.121320343559644) circle (.1);
\fill (0.7071067811865477,-1.707106781186548) circle (.1);
\fill (0.7071067811865477,-0.7071067811865477) circle (.1);
\fill (0.7071067811865476,0.7071067811865476) circle (.1);
\fill (0.7071067811865476,1.707106781186548) circle (.1);
\fill (0.7071067811865475,4.121320343559643) circle (.1);
\fill (1.707106781186548,-4.121320343559644) circle (.1);
\fill (1.7071067811865477,-1.707106781186548) circle (.1);
\fill (1.7071067811865477,-0.7071067811865477) circle (.1);
\fill (1.7071067811865475,0.7071067811865476) circle (.1);
\fill (1.7071067811865475,1.707106781186548) circle (.1);
\fill (1.7071067811865475,4.121320343559643) circle (.1);
\fill (4.121320343559643,-4.121320343559644) circle (.1);
\fill (4.121320343559643,-1.7071067811865483) circle (.1);
\fill (4.121320343559643,-0.7071067811865478) circle (.1);
\fill (4.121320343559643,0.7071067811865475) circle (.1);
\fill (4.121320343559643,1.7071067811865477) circle (.1);
\fill (4.121320343559643,4.121320343559644) circle (.1);
\draw (0.4,0.4) node {B};
\end{tikzpicture}}
  \quad (b)~
  \mp{0.95}{\begin{tikzpicture}[scale=0.89]
\fill[color=green!50] (0.0,0.0) circle (1.415213562373095);
\draw[->] (-4.5,0) -- (4.5, 0);
\draw (-4,0) -- (-4,-0.1) node[below] {\small $-4$};
\draw (-3,0) -- (-3,-0.1) node[below] {\small $-3$};
\draw (-2,0) -- (-2,-0.1) node[below] {\small $-2$};
\draw (-1,0) -- (-1,-0.1) node[below] {\small $-1$};
\draw (0,0) -- (0,-0.1) node[below] {\small $0$};
\draw (1,0) -- (1,-0.1) node[below] {\small $1$};
\draw (2,0) -- (2,-0.1) node[below] {\small $2$};
\draw (3,0) -- (3,-0.1) node[below] {\small $3$};
\draw (4,0) -- (4,-0.1) node[below] {\small $4$};
\draw[->] (0,-4.5) -- (0,4.5);
\draw (0,-4) -- (-0.1,-4) node[left] {\small $-4$};
\draw (0,-3) -- (-0.1,-3) node[left] {\small $-3$};
\draw (0,-2) -- (-0.1,-2) node[left] {\small $-2$};
\draw (0,-1) -- (-0.1,-1) node[left] {\small $-1$};
\draw (0,0) -- (-0.1,0) node[left] {\small $0$};
\draw (0,1) -- (-0.1,1) node[left] {\small $1$};
\draw (0,2) -- (-0.1,2) node[left] {\small $2$};
\draw (0,3) -- (-0.1,3) node[left] {\small $3$};
\draw (0,4) -- (-0.1,4) node[left] {\small $4$};
\fill (-3.414213562373095,-3.414213562373096) circle (.1);
\fill (-3.414213562373095,-2.4142135623730954) circle (.1);
\fill (-3.4142135623730954,-1.0) circle (.1);
\fill (-3.4142135623730954,0.00000000000000011102230246251565) circle (.1);
\fill (-3.4142135623730954,1.0000000000000004) circle (.1);
\fill (-3.4142135623730954,2.414213562373096) circle (.1);
\fill (-3.4142135623730954,3.414213562373096) circle (.1);
\fill (-2.414213562373095,-3.414213562373096) circle (.1);
\fill (-2.414213562373095,-2.4142135623730954) circle (.1);
\fill (-2.4142135623730954,-1.0) circle (.1);
\fill (-2.4142135623730954,0.00000000000000011102230246251565) circle (.1);
\fill (-2.4142135623730954,1.0000000000000004) circle (.1);
\fill (-2.4142135623730954,2.414213562373096) circle (.1);
\fill (-2.4142135623730954,3.414213562373096) circle (.1);
\fill (-1.4142135623730954,0.00000000000000011102230246251565) circle (.1);
\fill (-0.9999999999999999,-3.4142135623730954) circle (.1);
\fill (-0.9999999999999999,-2.4142135623730954) circle (.1);
\fill (-1.0,-1.0000000000000002) circle (.1);
\fill (-1.0,0.0) circle (.1);
\fill (-1.0,1.0000000000000002) circle (.1);
\fill (-1.0,2.4142135623730954) circle (.1);
\fill (-1.0,3.4142135623730954) circle (.1);
\fill (0.00000000000000011102230246251565,-3.4142135623730954) circle (.1);
\fill (0.00000000000000011102230246251565,-2.4142135623730954) circle (.1);
\fill (0.00000000000000011102230246251565,-1.4142135623730954) circle (.1);
\fill (0.0,-1.0000000000000002) circle (.1);
\fill (0.0,0.0) circle (.1);
\fill (0.0,1.0000000000000002) circle (.1);
\fill (-0.00000000000000011102230246251565,1.4142135623730954) circle (.1);
\fill (-0.00000000000000011102230246251565,2.4142135623730954) circle (.1);
\fill (-0.00000000000000011102230246251565,3.4142135623730954) circle (.1);
\fill (1.0,-3.4142135623730954) circle (.1);
\fill (1.0,-2.4142135623730954) circle (.1);
\fill (1.0,-1.0000000000000002) circle (.1);
\fill (1.0,0.0) circle (.1);
\fill (1.0,1.0000000000000002) circle (.1);
\fill (0.9999999999999999,2.4142135623730954) circle (.1);
\fill (0.9999999999999999,3.4142135623730954) circle (.1);
\fill (1.4142135623730954,-0.00000000000000011102230246251565) circle (.1);
\fill (2.4142135623730954,-3.414213562373096) circle (.1);
\fill (2.4142135623730954,-2.414213562373096) circle (.1);
\fill (2.4142135623730954,-1.0000000000000004) circle (.1);
\fill (2.4142135623730954,-0.00000000000000011102230246251565) circle (.1);
\fill (2.4142135623730954,1.0) circle (.1);
\fill (2.414213562373095,2.4142135623730954) circle (.1);
\fill (2.414213562373095,3.414213562373096) circle (.1);
\fill (3.4142135623730954,-3.414213562373096) circle (.1);
\fill (3.4142135623730954,-2.414213562373096) circle (.1);
\fill (3.4142135623730954,-1.0000000000000004) circle (.1);
\fill (3.4142135623730954,-0.00000000000000011102230246251565) circle (.1);
\fill (3.4142135623730954,1.0) circle (.1);
\fill (3.414213562373095,2.4142135623730954) circle (.1);
\fill (3.414213562373095,3.414213562373096) circle (.1);
\fill (-4.121320343559643,-4.121320343559644) circle (.1);
\fill (-4.121320343559643,-3.121320343559643) circle (.1);
\fill (-4.121320343559643,-2.707106781186548) circle (.1);
\fill (-4.121320343559643,-1.7071067811865477) circle (.1);
\fill (-4.121320343559643,-0.7071067811865475) circle (.1);
\fill (-4.121320343559643,0.7071067811865478) circle (.1);
\fill (-4.121320343559643,1.7071067811865483) circle (.1);
\fill (-4.121320343559643,2.7071067811865483) circle (.1);
\fill (-4.121320343559643,3.121320343559643) circle (.1);
\fill (-4.121320343559643,4.121320343559644) circle (.1);
\fill (-3.121320343559643,-4.121320343559644) circle (.1);
\fill (-3.121320343559643,-1.7071067811865477) circle (.1);
\fill (-3.121320343559643,-0.7071067811865475) circle (.1);
\fill (-3.121320343559643,0.7071067811865478) circle (.1);
\fill (-3.121320343559643,1.7071067811865483) circle (.1);
\fill (-3.121320343559643,4.121320343559644) circle (.1);
\fill (-2.7071067811865475,-4.121320343559643) circle (.1);
\fill (-2.7071067811865475,-1.707106781186548) circle (.1);
\fill (-2.707106781186548,1.707106781186548) circle (.1);
\fill (-2.707106781186548,4.121320343559644) circle (.1);
\fill (-1.7071067811865475,-4.121320343559643) circle (.1);
\fill (-1.7071067811865475,-3.121320343559643) circle (.1);
\fill (-1.7071067811865475,-2.707106781186548) circle (.1);
\fill (-1.7071067811865475,-1.707106781186548) circle (.1);
\fill (-1.7071067811865475,-0.7071067811865476) circle (.1);
\fill (-1.7071067811865477,0.7071067811865477) circle (.1);
\fill (-1.7071067811865477,1.707106781186548) circle (.1);
\fill (-1.7071067811865477,2.707106781186548) circle (.1);
\fill (-1.707106781186548,3.1213203435596433) circle (.1);
\fill (-1.707106781186548,4.121320343559644) circle (.1);
\fill (-0.7071067811865475,-4.121320343559643) circle (.1);
\fill (-0.7071067811865475,-3.121320343559643) circle (.1);
\fill (-0.7071067811865476,-1.707106781186548) circle (.1);
\fill (-0.7071067811865476,-0.7071067811865476) circle (.1);
\fill (-0.7071067811865477,0.7071067811865477) circle (.1);
\fill (-0.7071067811865477,1.707106781186548) circle (.1);
\fill (-0.7071067811865478,3.1213203435596433) circle (.1);
\fill (-0.7071067811865478,4.121320343559644) circle (.1);
\fill (0.7071067811865478,-4.121320343559644) circle (.1);
\fill (0.7071067811865478,-3.1213203435596433) circle (.1);
\fill (0.7071067811865477,-1.707106781186548) circle (.1);
\fill (0.7071067811865477,-0.7071067811865477) circle (.1);
\fill (0.7071067811865476,0.7071067811865476) circle (.1);
\fill (0.7071067811865476,1.707106781186548) circle (.1);
\fill (0.7071067811865475,3.121320343559643) circle (.1);
\fill (0.7071067811865475,4.121320343559643) circle (.1);
\fill (1.707106781186548,-4.121320343559644) circle (.1);
\fill (1.707106781186548,-3.1213203435596433) circle (.1);
\fill (1.7071067811865477,-2.707106781186548) circle (.1);
\fill (1.7071067811865477,-1.707106781186548) circle (.1);
\fill (1.7071067811865477,-0.7071067811865477) circle (.1);
\fill (1.7071067811865475,0.7071067811865476) circle (.1);
\fill (1.7071067811865475,1.707106781186548) circle (.1);
\fill (1.7071067811865475,2.707106781186548) circle (.1);
\fill (1.7071067811865475,3.121320343559643) circle (.1);
\fill (1.7071067811865475,4.121320343559643) circle (.1);
\fill (2.707106781186548,-4.121320343559644) circle (.1);
\fill (2.707106781186548,-1.707106781186548) circle (.1);
\fill (2.7071067811865475,1.707106781186548) circle (.1);
\fill (2.7071067811865475,4.121320343559643) circle (.1);
\fill (3.121320343559643,-4.121320343559644) circle (.1);
\fill (3.121320343559643,-1.7071067811865483) circle (.1);
\fill (3.121320343559643,-0.7071067811865478) circle (.1);
\fill (3.121320343559643,0.7071067811865475) circle (.1);
\fill (3.121320343559643,1.7071067811865477) circle (.1);
\fill (3.121320343559643,4.121320343559644) circle (.1);
\fill (4.121320343559643,-4.121320343559644) circle (.1);
\fill (4.121320343559643,-3.121320343559643) circle (.1);
\fill (4.121320343559643,-2.7071067811865483) circle (.1);
\fill (4.121320343559643,-1.7071067811865483) circle (.1);
\fill (4.121320343559643,-0.7071067811865478) circle (.1);
\fill (4.121320343559643,0.7071067811865475) circle (.1);
\fill (4.121320343559643,1.7071067811865477) circle (.1);
\fill (4.121320343559643,2.707106781186548) circle (.1);
\fill (4.121320343559643,3.121320343559643) circle (.1);
\fill (4.121320343559643,4.121320343559644) circle (.1);
\draw (0.4,0.4) node {B};
\end{tikzpicture}}
  \]
  \[ (c)~
  \mp{0.95}{\begin{tikzpicture}[scale=0.89]
\fill[color=green!50] (-0.20388260857871685,-0.2773192715870234)
.. controls (1.6259196305123527,-1.622573251067478) and (3.200548781081991,-2.588957006105265) .. (3.3131500365557374,-2.435797801573568)
.. controls (3.4257512920294837,-2.282638597041871) and (2.0336848476697864,-1.0679347078934311) .. (0.20388260857871685,0.2773192715870234)
.. controls (-1.6259196305123527,1.622573251067478) and (-3.200548781081991,2.588957006105265) .. (-3.3131500365557374,2.435797801573568)
.. controls (-3.4257512920294837,2.282638597041871) and (-2.0336848476697864,1.0679347078934311) .. (-0.20388260857871685,-0.2773192715870234)
-- cycle;
\draw[->] (-4.5,0) -- (4.5, 0);
\draw (-4,0) -- (-4,-0.1) node[below] {\small $-4$};
\draw (-3,0) -- (-3,-0.1) node[below] {\small $-3$};
\draw (-2,0) -- (-2,-0.1) node[below] {\small $-2$};
\draw (-1,0) -- (-1,-0.1) node[below] {\small $-1$};
\draw (0,0) -- (0,-0.1) node[below] {\small $0$};
\draw (1,0) -- (1,-0.1) node[below] {\small $1$};
\draw (2,0) -- (2,-0.1) node[below] {\small $2$};
\draw (3,0) -- (3,-0.1) node[below] {\small $3$};
\draw (4,0) -- (4,-0.1) node[below] {\small $4$};
\draw[->] (0,-4.5) -- (0,4.5);
\draw (0,-4) -- (-0.1,-4) node[left] {\small $-4$};
\draw (0,-3) -- (-0.1,-3) node[left] {\small $-3$};
\draw (0,-2) -- (-0.1,-2) node[left] {\small $-2$};
\draw (0,-1) -- (-0.1,-1) node[left] {\small $-1$};
\draw (0,0) -- (-0.1,0) node[left] {\small $0$};
\draw (0,1) -- (-0.1,1) node[left] {\small $1$};
\draw (0,2) -- (-0.1,2) node[left] {\small $2$};
\draw (0,3) -- (-0.1,3) node[left] {\small $3$};
\draw (0,4) -- (-0.1,4) node[left] {\small $4$};
\fill (-4.414213562373095,-1.4142135623730951) circle (.1);
\fill (-4.414213562373096,1.0000000000000004) circle (.1);
\fill (-3.8284271247461907,-4.414213562373096) circle (.1);
\fill (-3.8284271247461907,-1.0) circle (.1);
\fill (-3.8284271247461903,1.4142135623730954) circle (.1);
\fill (-3.414213562373095,-2.4142135623730954) circle (.1);
\fill (-3.4142135623730954,3.414213562373096) circle (.1);
\fill (-3.0,2.0000000000000004) circle (.1);
\fill (-2.8284271247461907,-2.0) circle (.1);
\fill (-2.414213562373095,-3.414213562373096) circle (.1);
\fill (-2.4142135623730954,0.00000000000000011102230246251565) circle (.1);
\fill (-2.4142135623730954,2.414213562373096) circle (.1);
\fill (-2.0,-1.4142135623730954) circle (.1);
\fill (-2.0,4.414213562373097) circle (.1);
\fill (-1.4142135623730954,-1.0) circle (.1);
\fill (-1.4142135623730954,1.4142135623730954) circle (.1);
\fill (-0.9999999999999999,-2.4142135623730954) circle (.1);
\fill (-1.0,1.0000000000000002) circle (.1);
\fill (-1.0,3.4142135623730954) circle (.1);
\fill (-0.41421356237309515,-4.414213562373096) circle (.1);
\fill (-0.41421356237309537,-2.0000000000000004) circle (.1);
\fill (-0.41421356237309537,3.8284271247461903) circle (.1);
\fill (0.00000000000000011102230246251565,-2.4142135623730954) circle (.1);
\fill (0.0,0.0) circle (.1);
\fill (-0.00000000000000011102230246251565,2.4142135623730954) circle (.1);
\fill (0.41421356237309537,-3.8284271247461903) circle (.1);
\fill (0.41421356237309537,2.0000000000000004) circle (.1);
\fill (0.41421356237309515,4.414213562373096) circle (.1);
\fill (1.0,-3.4142135623730954) circle (.1);
\fill (1.0,-1.0000000000000002) circle (.1);
\fill (0.9999999999999999,2.4142135623730954) circle (.1);
\fill (1.4142135623730954,-1.4142135623730954) circle (.1);
\fill (1.4142135623730954,1.0) circle (.1);
\fill (2.0,-4.414213562373097) circle (.1);
\fill (2.0,1.4142135623730954) circle (.1);
\fill (2.4142135623730954,-2.414213562373096) circle (.1);
\fill (2.4142135623730954,-0.00000000000000011102230246251565) circle (.1);
\fill (2.414213562373095,3.414213562373096) circle (.1);
\fill (2.8284271247461907,2.0) circle (.1);
\fill (3.0,-2.0000000000000004) circle (.1);
\fill (3.4142135623730954,-3.414213562373096) circle (.1);
\fill (3.414213562373095,2.4142135623730954) circle (.1);
\fill (3.8284271247461903,-1.4142135623730954) circle (.1);
\fill (3.8284271247461907,1.0) circle (.1);
\fill (3.8284271247461907,4.414213562373096) circle (.1);
\fill (4.414213562373096,-1.0000000000000004) circle (.1);
\fill (4.414213562373095,1.4142135623730951) circle (.1);
\fill (-4.121320343559643,-4.121320343559644) circle (.1);
\fill (-4.121320343559643,-1.7071067811865477) circle (.1);
\fill (-4.121320343559643,1.7071067811865483) circle (.1);
\fill (-4.121320343559643,4.121320343559644) circle (.1);
\fill (-3.7071067811865475,0.29289321881345265) circle (.1);
\fill (-3.121320343559643,0.7071067811865478) circle (.1);
\fill (-3.121320343559643,3.121320343559643) circle (.1);
\fill (-2.7071067811865475,-3.121320343559643) circle (.1);
\fill (-2.7071067811865475,-0.7071067811865476) circle (.1);
\fill (-2.707106781186548,2.707106781186548) circle (.1);
\fill (-2.121320343559643,-2.707106781186548) circle (.1);
\fill (-2.121320343559643,-0.2928932188134524) circle (.1);
\fill (-1.7071067811865475,-4.121320343559643) circle (.1);
\fill (-1.7071067811865477,1.707106781186548) circle (.1);
\fill (-1.707106781186548,4.121320343559644) circle (.1);
\fill (-1.292893218813452,-2.121320343559643) circle (.1);
\fill (-1.1213203435596428,-3.707106781186548) circle (.1);
\fill (-0.7071067811865476,-1.707106781186548) circle (.1);
\fill (-0.7071067811865477,0.7071067811865477) circle (.1);
\fill (-0.7071067811865478,4.121320343559644) circle (.1);
\fill (-0.2928932188134522,-3.1213203435596433) circle (.1);
\fill (-0.2928932188134524,2.707106781186548) circle (.1);
\fill (0.2928932188134524,-2.707106781186548) circle (.1);
\fill (0.2928932188134522,3.1213203435596433) circle (.1);
\fill (0.7071067811865478,-4.121320343559644) circle (.1);
\fill (0.7071067811865477,-0.7071067811865477) circle (.1);
\fill (0.7071067811865476,1.707106781186548) circle (.1);
\fill (1.1213203435596428,3.707106781186548) circle (.1);
\fill (1.292893218813452,2.121320343559643) circle (.1);
\fill (1.707106781186548,-4.121320343559644) circle (.1);
\fill (1.7071067811865477,-1.707106781186548) circle (.1);
\fill (1.7071067811865475,4.121320343559643) circle (.1);
\fill (2.121320343559643,0.2928932188134524) circle (.1);
\fill (2.121320343559643,2.707106781186548) circle (.1);
\fill (2.707106781186548,-2.707106781186548) circle (.1);
\fill (2.7071067811865475,0.7071067811865476) circle (.1);
\fill (2.7071067811865475,3.121320343559643) circle (.1);
\fill (3.121320343559643,-3.121320343559643) circle (.1);
\fill (3.121320343559643,-0.7071067811865478) circle (.1);
\fill (3.7071067811865475,-0.29289321881345265) circle (.1);
\fill (4.121320343559643,-4.121320343559644) circle (.1);
\fill (4.121320343559643,-1.7071067811865483) circle (.1);
\fill (4.121320343559643,1.7071067811865477) circle (.1);
\fill (4.121320343559643,4.121320343559644) circle (.1);
\draw (-2.3,1.7) node {B};
\end{tikzpicture}}
  \]
  \caption{The complex grid for three different convex sets $B$. In
    each case, the set $B$ is shown in green, and grid points are
    shown as black dots. (a) $B=[-1,1]^2$. (b) $B=\s{(x,y)\mid
      x^2+y^2\leq 2}$. (c) $B=\s{(x,y)\mid 6x^2 + 16xy + 11y^2 \leq
      2}$.}
  \label{fig-grid-2d}
\end{figure}
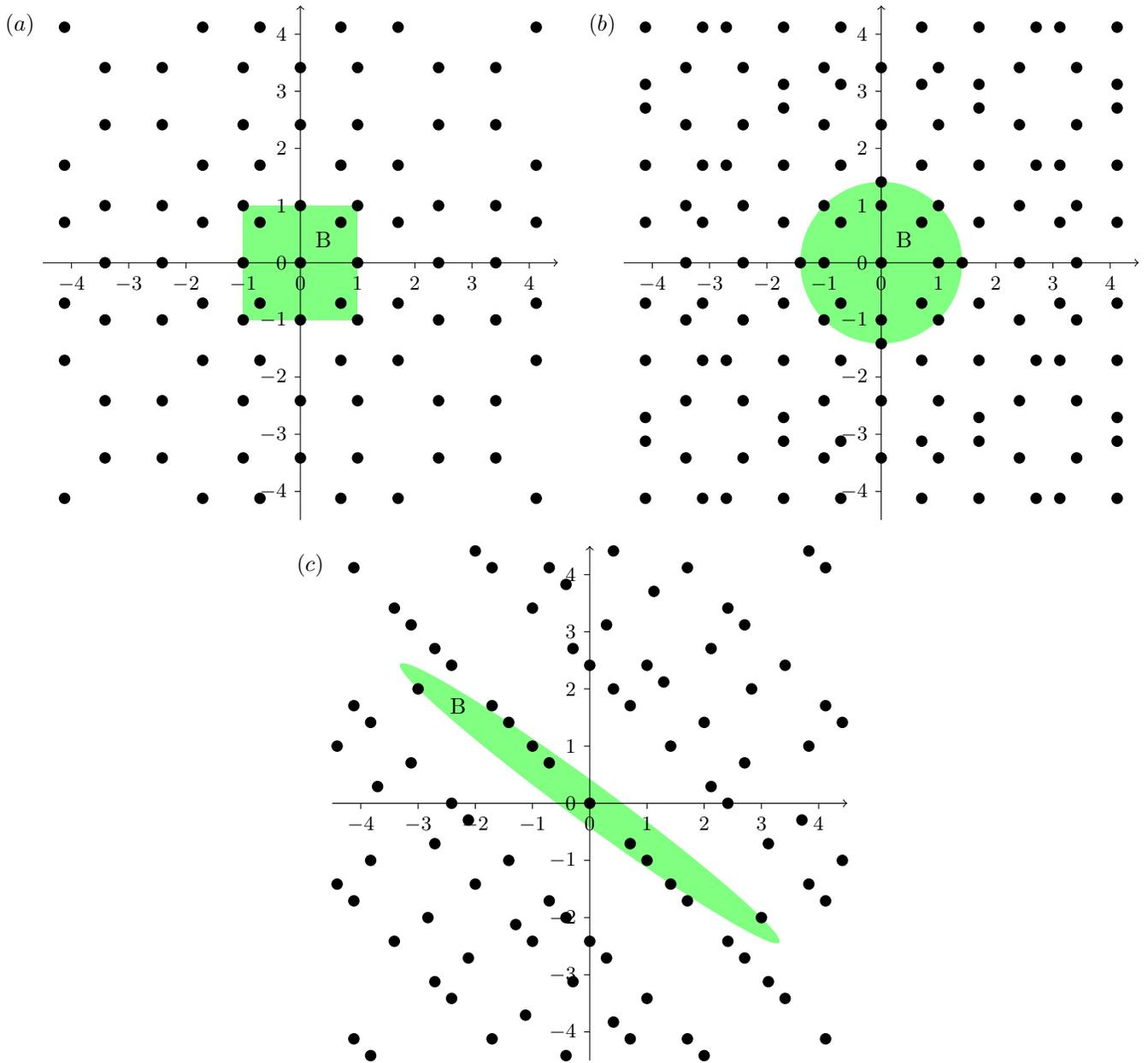

\begin{example}
  Figure~\ref{fig-grid-2d} illustrates the complex grids for several
  different convex sets $B$. Note that the grid has a $90$-degree
  symmetry in (a), a $45$-degree symmetry in (b), and a $180$-degree
  symmetry in (c).
\end{example}

\begin{definition}
  Let $A$ and $B$ be subsets of $\R^2$. The {\em two-dimensional grid
    problem} for $A$ and $B$ is the following:
  \begin{equation}\label{eqn-grid-constraint-2d}
    \mbox{{\bf Two-dimensional grid problem:} Find
      $u\in\Z[\omega]$ satisfying $u \in A$ and $u\bul\in B$.}
  \end{equation}
\end{definition}

As in the one-dimensional case, the grid problem can be understood as
looking for points in the intersection of the set $A$ and the grid
for $B$.  Our goal will be to prove a two-dimensional analogue of
Proposition~\ref{prop-algorithm-1d}, namely, that there is an
efficient algorithm which, given two bounded convex subsets $A$ and
$B$ of $\R^2$ with non-empty interior, enumerates all solutions of the
two-dimensional grid problem for $A$ and $B$.

However, the proof is more complicated than in the one-dimensional
case. We will consider several special cases before solving the
general problem in Section~\ref{ssec-proof-prop-algorithm-2d}.

\begin{remark}\label{rem-convex-specify}
  By analogy with Remark~\ref{rem-interval-specify}, we should
  indicate what it means for a bounded convex subset $A$ of $\R^2$ to
  be ``given'' as the input to an algorithm. Again, the details of
  this do not matter much. For our purposes, it will suffice to assume
  that a convex set is given along with the following information:
  \begin{itemize}
  \item a convex polygon enclosing $A$, say with rational vertices,
    and such that the area of the polygon exceeds that of $A$ by at
    most a fixed constant factor;
  \item a method to decide, for any given point of $\D[\omega]$,
    whether it is in $A$ or not; and
  \item a method to compute the intersection of $A$ with any straight
    line in $\D[\omega]$. More precisely, given any straight line
    parameterized as $L(t)=p+tq$, with $p,q\in\D[\omega]$, we can
    effectively determine the interval $\s{t\mid L(t)\in A}$ in the
    sense of Remark~\ref{rem-interval-specify}.
  \end{itemize}
\end{remark}

\subsection{Upright rectangles}\label{subsec-upright-rectangles}

A special case of the two-dimensional grid problem arises when both
$A$ and $B$ are {\em upright rectangles}, by which we mean sets of the
form $[x_0,x_1]\times [y_0,y_1]$.  If $A$ and $B$ are upright
rectangles, then the two-dimensional grid problem can easily be
reduced to the one-dimensional case. We start with a lemma
characterizing $\Z[\omega]$.

\begin{lemma}\label{lem-zomega}
  A complex number $u$ is in $\Z[\omega]$ if and only it can be
  written of the form $u=\alpha+\beta i$ or of the form
  $u=\alpha+\beta i+\omega$, where $\alpha,\beta\in\Z[\sqrt 2]$.
\end{lemma}

\begin{proof}
  The right-to-left implication is trivial. For the left-to-right
  implication, let $u=a\omega^3+b\omega^2+c\omega+d$, where
  $a,b,c,d\in\Z$. Noting that $\omega=\frac{1+i}{\sqrt 2}$, we have
  \[ u = (d + \frac{c-a}{2} \sqrt2) + (b + \frac{c+a}{2}{\sqrt2})i.
  \]
  If $c-a$ (and therefore $c+a$) is even, then $u$ is of the first
  form, with $\alpha=d + \frac{c-a}{2} \sqrt2$ and $\beta=b +
  \frac{c+a}{2}{\sqrt2}$. If $c-a$ (and therefore $c+a$) is odd, then
  $u$ is of the second form, with $\alpha=d + \frac{c-a-1}{2} \sqrt2$
  and $\beta = b + \frac{c+a-1}{2}{\sqrt2}$. 
\end{proof}

\begin{lemma}\label{lem-upright-rectangles}
  There is an algorithm which, given a pair of upright rectangles $A$
  and $B$, enumerates all solutions of the two-dimensional grid
  problem for $A$ and $B$. Moreover, the algorithm requires only a
  constant number of arithmetic operations per solution produced.
\end{lemma}

\begin{proof}
  By assumption, $A=A_x\times A_y$ and $B=B_x\times B_y$, where $A_x$,
  $A_y$, $B_x$, and $B_y$ are closed intervals. By
  Lemma~\ref{lem-zomega}, any potential solution is of the form
  $u=\alpha+\beta i$ or $u=\alpha+\beta i+\omega$, where
  $\alpha,\beta\in\Z[\sqrt 2]$. When $u=\alpha+\beta i$, then $u\bul =
  \alpha\bul + \beta\bul i$. Therefore, the two-dimensional grid
  constraints $u\in A$ and $u\bul\in B$ are equivalent to the
  one-dimensional constraints $\alpha\in A_x$, $\alpha\bul\in B_x$ and
  $\beta\in A_y$, $\beta\bul\in B_y$.  On the other hand, when
  $u=\alpha+\beta i+\omega$, let $v = u-\omega = \alpha+\beta i$. Then
  $v\bul = u\bul + \omega$, and the constraints $u\in A$ and $u\bul\in
  B$ are equivalent to $v\in A-\omega$ and $v\bul \in B+\omega$, which
  reduces to the one-dimensional constraints $\alpha\in
  A_x-\frac{1}{\sqrt2}$, $\alpha\bul\in B_x+\frac{1}{\sqrt2}$ and
  $\beta\in A_y-\frac{1}{\sqrt2}$, $\beta\bul\in
  B_y+\frac{1}{\sqrt2}$.  In both cases, the solutions to the
  one-dimensional constraints can be efficiently enumerated by
  Proposition~\ref{prop-algorithm-1d}.
\end{proof}

\subsection{Upright sets}\label{ssec-uprightness}

We can generalize the method of
Section~\ref{subsec-upright-rectangles} to convex sets that are {\em
  close} to upright rectangles in a suitable sense.

\begin{definition}\label{def-uprightness}
  Let $A$ be a bounded convex subset of $\R^2$. The {\em bounding box}
  of $A$, denoted $\BBox(A)$, is the smallest set of the form
  $[x_0,x_1]\times [y_0,y_1]$ that contains $A$. The {\em uprightness}
  of $A$, denoted $\up(A)$, is defined to be the ratio of the area of
  $A$ to the area of its bounding box:
  \begin{equation}
    \up(A) = \frac{\area(A)}{\area(\BBox(A))}.
  \end{equation}
  Therefore, the uprightness is a real number between 0 and 1. We say
  that $A$ is {\em $M$-upright} if $\up(A)\geq M$.
\end{definition}

\begin{lemma}\label{lem-upright-set}
  There exists an algorithm which, given a pair $A,B$ of convex
  $M$-upright sets, enumerates all solutions of the two-dimensional
  grid problem for $A$ and $B$. Moreover, the algorithm requires
  $O(1/M^2)$ arithmetic operations per solution produced. In
  particular, when $M>0$ is fixed, it requires only a constant number
  of operations per solution.
\end{lemma}

\begin{proof}
  By Lemma~\ref{lem-upright-rectangles}, we can efficiently enumerate
  the solutions of the grid problem for $\BBox(A)$ and
  $\BBox(B)$. Moreover, as shown in the proof of
  Lemma~\ref{lem-upright-rectangles}, the solutions are arranged in
  rows and columns. For each such candidate solution $u$, we only need
  to check whether $u$ is also a solution for $A$ and $B$.
  To establish the efficiency of the algorithm, we need to ensure that
  the total number of solutions is not too small in relation to the
  total number of candidates produced. To see this, note that, with the
  exception of trivial cases, when the number of rows or columns is
  very small, $M$-uprightness and convexity ensure that the proportion
  of candidates $u$ that are solutions for $A$ and $B$ is
  approximately $M^2 : 1$. Therefore, the runtime per solution differs
  from that of Lemma~\ref{lem-upright-rectangles} by at most a
  factor of $O(1/M^2)$.
\end{proof}

\begin{figure}
  \[ 
  \mp{0.95}{\begin{tikzpicture}[scale=0.9]
\fill[color=green!50] (1.001,1.001) -- (-1.001,1.001) -- (-1.001,-1.001) -- (1.001,-1.001) -- cycle;
\fill[color=red!50] (-4.414213562373095,2.3)
.. controls (-4.414213562373095,3.3083293854947433) and (-3.78104858350254,4.125741858350553) .. (-3.0,4.125741858350553)
.. controls (-2.21895141649746,4.125741858350553) and (-1.585786437626905,3.3083293854947433) .. (-1.585786437626905,2.3)
.. controls (-1.585786437626905,1.2916706145052563) and (-2.21895141649746,0.47425814164944624) .. (-3.0,0.47425814164944624)
.. controls (-3.78104858350254,0.47425814164944624) and (-4.414213562373095,1.2916706145052563) .. (-4.414213562373095,2.3)
-- cycle;
\fill[color=red!50] (3.3898024799571416,0.75)
.. controls (3.3898024799571416,2.702621458756349) and (3.55902293854554,4.285533905932736) .. (3.767766952966369,4.285533905932736)
.. controls (3.9765109673871977,4.285533905932736) and (4.145731425975596,2.702621458756349) .. (4.145731425975596,0.75)
.. controls (4.145731425975596,-1.202621458756349) and (3.9765109673871977,-2.785533905932736) .. (3.767766952966369,-2.785533905932736)
.. controls (3.55902293854554,-2.785533905932736) and (3.3898024799571416,-1.202621458756349) .. (3.3898024799571416,0.75)
-- cycle;
\fill[color=red!50] (0.5890570822881064,-2.2190159580867657)
.. controls (2.574398059992575,-3.031864993828518) and (4.225824174901229,-3.588252783348306) .. (4.27762054820651,-3.4617428683483897)
.. controls (4.3294169215117915,-3.3352329533484735) and (2.7619693203656004,-2.57373206497732) .. (0.7766283426611317,-1.7608830292355675)
.. controls (-1.2087126350433368,-0.9480339934938151) and (-2.8601387499519904,-0.39164620397402683) .. (-2.9119351232572717,-0.5181561189739432)
.. controls (-2.963731496562553,-0.6446660339738596) and (-1.3962838954163623,-1.4061669223450133) .. (0.5890570822881064,-2.2190159580867657)
-- cycle;
\draw[color=red] (-4.414213562373095,0.474258141649446) -- (-1.5857864376269049,0.474258141649446) -- (-1.5857864376269049,4.125741858350554) -- (-4.414213562373095,4.125741858350554) -- cycle;
\draw[color=red] (3.3898024799571416,-2.7855339059327373) -- (4.145731425975596,-2.7855339059327373) -- (4.145731425975596,4.285533905932738) -- (3.3898024799571416,4.285533905932738) -- cycle;
\draw[color=red] (-2.913158320992765,-3.4794618913596915) -- (4.278843745942003,-3.4794618913596915) -- (4.278843745942003,-0.5004370959626416) -- (-2.913158320992765,-0.5004370959626416) -- cycle;
\draw (0.4,0.4) node {$B$};
\draw (-3.0,3.0) node {$A_1$};
\draw (3.767766952966369,3.0) node {$A_2$};
\draw (1.7,-2.4) node {$A_3$};
\draw[->] (-4.8,0) -- (4.8, 0);
\draw (-4,0) -- (-4,-0.1) node[below] {\small $-4$};
\draw (-3,0) -- (-3,-0.1) node[below] {\small $-3$};
\draw (-2,0) -- (-2,-0.1) node[below] {\small $-2$};
\draw (-1,0) -- (-1,-0.1) node[below] {\small $-1$};
\draw (0,0) -- (0,-0.1) node[below] {\small $0$};
\draw (1,0) -- (1,-0.1) node[below] {\small $1$};
\draw (2,0) -- (2,-0.1) node[below] {\small $2$};
\draw (3,0) -- (3,-0.1) node[below] {\small $3$};
\draw (4,0) -- (4,-0.1) node[below] {\small $4$};
\draw[->] (0,-4.5) -- (0,4.5);
\draw (0,-4) -- (-0.1,-4) node[left] {\small $-4$};
\draw (0,-3) -- (-0.1,-3) node[left] {\small $-3$};
\draw (0,-2) -- (-0.1,-2) node[left] {\small $-2$};
\draw (0,-1) -- (-0.1,-1) node[left] {\small $-1$};
\draw (0,0) -- (-0.1,0) node[left] {\small $0$};
\draw (0,1) -- (-0.1,1) node[left] {\small $1$};
\draw (0,2) -- (-0.1,2) node[left] {\small $2$};
\draw (0,3) -- (-0.1,3) node[left] {\small $3$};
\draw (0,4) -- (-0.1,4) node[left] {\small $4$};
\fill (-3.414213562373095,-3.414213562373096) circle (.1);
\fill (-3.414213562373095,-2.4142135623730954) circle (.1);
\fill (-3.4142135623730954,-1.0) circle (.1);
\fill (-3.4142135623730954,0.00000000000000011102230246251565) circle (.1);
\fill (-3.4142135623730954,1.0000000000000004) circle (.1);
\fill (-3.4142135623730954,2.414213562373096) circle (.1);
\fill (-3.4142135623730954,3.414213562373096) circle (.1);
\fill (-2.414213562373095,-3.414213562373096) circle (.1);
\fill (-2.414213562373095,-2.4142135623730954) circle (.1);
\fill (-2.4142135623730954,-1.0) circle (.1);
\fill (-2.4142135623730954,0.00000000000000011102230246251565) circle (.1);
\fill (-2.4142135623730954,1.0000000000000004) circle (.1);
\fill (-2.4142135623730954,2.414213562373096) circle (.1);
\fill (-2.4142135623730954,3.414213562373096) circle (.1);
\fill (-0.9999999999999999,-3.4142135623730954) circle (.1);
\fill (-0.9999999999999999,-2.4142135623730954) circle (.1);
\fill (-1.0,-1.0000000000000002) circle (.1);
\fill (-1.0,0.0) circle (.1);
\fill (-1.0,1.0000000000000002) circle (.1);
\fill (-1.0,2.4142135623730954) circle (.1);
\fill (-1.0,3.4142135623730954) circle (.1);
\fill (0.00000000000000011102230246251565,-3.4142135623730954) circle (.1);
\fill (0.00000000000000011102230246251565,-2.4142135623730954) circle (.1);
\fill (0.0,-1.0000000000000002) circle (.1);
\fill (0.0,0.0) circle (.1);
\fill (0.0,1.0000000000000002) circle (.1);
\fill (-0.00000000000000011102230246251565,2.4142135623730954) circle (.1);
\fill (-0.00000000000000011102230246251565,3.4142135623730954) circle (.1);
\fill (1.0,-3.4142135623730954) circle (.1);
\fill (1.0,-2.4142135623730954) circle (.1);
\fill (1.0,-1.0000000000000002) circle (.1);
\fill (1.0,0.0) circle (.1);
\fill (1.0,1.0000000000000002) circle (.1);
\fill (0.9999999999999999,2.4142135623730954) circle (.1);
\fill (0.9999999999999999,3.4142135623730954) circle (.1);
\fill (2.4142135623730954,-3.414213562373096) circle (.1);
\fill (2.4142135623730954,-2.414213562373096) circle (.1);
\fill (2.4142135623730954,-1.0000000000000004) circle (.1);
\fill (2.4142135623730954,-0.00000000000000011102230246251565) circle (.1);
\fill (2.4142135623730954,1.0) circle (.1);
\fill (2.414213562373095,2.4142135623730954) circle (.1);
\fill (2.414213562373095,3.414213562373096) circle (.1);
\fill (3.4142135623730954,-3.414213562373096) circle (.1);
\fill (3.4142135623730954,-2.414213562373096) circle (.1);
\fill (3.4142135623730954,-1.0000000000000004) circle (.1);
\fill (3.4142135623730954,-0.00000000000000011102230246251565) circle (.1);
\fill (3.4142135623730954,1.0) circle (.1);
\fill (3.414213562373095,2.4142135623730954) circle (.1);
\fill (3.414213562373095,3.414213562373096) circle (.1);
\fill (-4.121320343559643,-4.121320343559644) circle (.1);
\fill (-4.121320343559643,-1.7071067811865477) circle (.1);
\fill (-4.121320343559643,-0.7071067811865475) circle (.1);
\fill (-4.121320343559643,0.7071067811865478) circle (.1);
\fill (-4.121320343559643,1.7071067811865483) circle (.1);
\fill (-4.121320343559643,4.121320343559644) circle (.1);
\fill (-1.7071067811865475,-4.121320343559643) circle (.1);
\fill (-1.7071067811865475,-1.707106781186548) circle (.1);
\fill (-1.7071067811865475,-0.7071067811865476) circle (.1);
\fill (-1.7071067811865477,0.7071067811865477) circle (.1);
\fill (-1.7071067811865477,1.707106781186548) circle (.1);
\fill (-1.707106781186548,4.121320343559644) circle (.1);
\fill (-0.7071067811865475,-4.121320343559643) circle (.1);
\fill (-0.7071067811865476,-1.707106781186548) circle (.1);
\fill (-0.7071067811865476,-0.7071067811865476) circle (.1);
\fill (-0.7071067811865477,0.7071067811865477) circle (.1);
\fill (-0.7071067811865477,1.707106781186548) circle (.1);
\fill (-0.7071067811865478,4.121320343559644) circle (.1);
\fill (0.7071067811865478,-4.121320343559644) circle (.1);
\fill (0.7071067811865477,-1.707106781186548) circle (.1);
\fill (0.7071067811865477,-0.7071067811865477) circle (.1);
\fill (0.7071067811865476,0.7071067811865476) circle (.1);
\fill (0.7071067811865476,1.707106781186548) circle (.1);
\fill (0.7071067811865475,4.121320343559643) circle (.1);
\fill (1.707106781186548,-4.121320343559644) circle (.1);
\fill (1.7071067811865477,-1.707106781186548) circle (.1);
\fill (1.7071067811865477,-0.7071067811865477) circle (.1);
\fill (1.7071067811865475,0.7071067811865476) circle (.1);
\fill (1.7071067811865475,1.707106781186548) circle (.1);
\fill (1.7071067811865475,4.121320343559643) circle (.1);
\fill (4.121320343559643,-4.121320343559644) circle (.1);
\fill (4.121320343559643,-1.7071067811865483) circle (.1);
\fill (4.121320343559643,-0.7071067811865478) circle (.1);
\fill (4.121320343559643,0.7071067811865475) circle (.1);
\fill (4.121320343559643,1.7071067811865477) circle (.1);
\fill (4.121320343559643,4.121320343559644) circle (.1);
\end{tikzpicture}}
  \]
  \caption{Grid problems for upright and non-upright sets}\label{fig-upright}
  \rule{\textwidth}{0.1mm}
\end{figure}
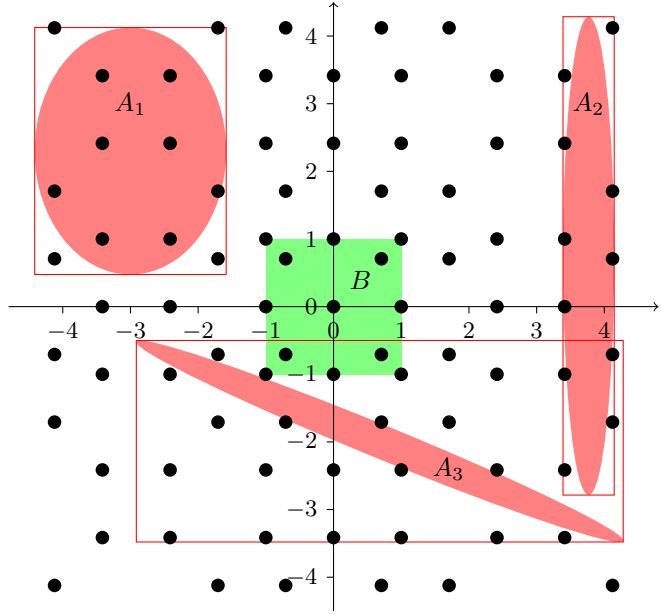

\begin{example}
  Figure~\ref{fig-upright} shows three different examples of grid
  problems. Each example uses the same set $B=[-1,1]\times[-1,1]$, shown in
  green. The grid for $B$ is shown as black dots. The sets $A_i$ are
  shown in red, for $i=1,2,3$, and their bounding boxes are shown in
  outline.  The typical case of an upright set is $A_1$. Here, a fixed
  proportion of grid points from the bounding box of $A_1$ are
  elements of $A_1$. The exceptional case of an upright set is
  $A_2$. Its bounding box spans only two columns of the
  grid. Therefore, although the bounding box contains many grid
  points, $A_2$ does not. However, this case is easily dealt with, by
  solving a one-dimensional grid problem for each of the grid columns
  separately. Finally, the set $A_3$ is not upright. In this case,
  Lemma~\ref{lem-upright-set} is not helpful, and a priori, it could
  be a difficult problem to find grid points in $A_3$.
\end{example}

\subsection{Grid operators}\label{ssec-gridops}

The method of Section~\ref{ssec-uprightness} can be further
generalized by using certain linear transformations to turn
non-upright sets into upright sets. The linear transformations that
are useful for this purpose are {\em special grid operators}:

\begin{definition}
  As before, we regard $\Z[\omega]$ as a subset of $\R^2$. A real linear
  operator $\G:\R^2 \to \R^2$ is called a {\em grid operator} if
  $\G(\Z[\omega]) \seq \Z[\omega]$. Moreover, a grid operator $\G$ is
  called {\em special} if it has determinant $\pm 1$.
\end{definition}

Grid operators are characterized by the following lemma.

\begin{lemma}\label{lem-gridoperators}
  Let $\G:\R^2\to\R^2$ be a linear operator, which we can identify
  with a real $2\times 2$-matrix with real entries. Then $\G$ is a
  grid operator if and only if it is of the form
  \begin{equation}\label{eqn-gridoperator}
    \G =
    \left[
      \begin{array}{cc}
        a+\frac{a'}{\sqrt{2}} & b+\frac{b'}{\sqrt{2}}\\
        c+\frac{c'}{\sqrt{2}} & d+\frac{d'}{\sqrt{2}}
      \end{array}
    \right],
  \end{equation}
  where $a,b,c,d,a',b',c',d'$ are integers satisfying $a+b+c+d \equiv
  0\mmod{2}$ and $a'\equiv b'\equiv c' \equiv d'\mmod{2}$.
\end{lemma}

\begin{proof}
  By Lemma~\ref{lem-zomega}, we know that a vector $u\in\R^2$ is in
  $\Z[\omega]$ if and only if it can be written of the form
  \begin{equation}\label{eqn-zomega-point}
    u = \left[\begin{matrix} x_1+\frac{x_2}{\sqrt{2}} \\ y_1+\frac{y_2}{\sqrt{2}}
      \end{matrix}\right],
  \end{equation}
  where $x_1,x_2,y_1,y_2$ are integers and $x_2\equiv y_2\mmod{2}$.
  A simple computation then shows that every operator of the
  form {\eqref{eqn-gridoperator}} is a grid operator. For the
  converse, consider an arbitrary grid operator $\G$. We prove the
  claim by applying $\G$ to the three particular points
  $\begin{bsmallmatrix}1\\0\end{bsmallmatrix}$,
  $\begin{bsmallmatrix}0\\1\end{bsmallmatrix}$, and
  $\frac{1}{\sqrt{2}}\begin{bsmallmatrix}1\\1\end{bsmallmatrix} \in
  \Z[\omega]$.  From
  $\G\begin{bsmallmatrix}1\\0\end{bsmallmatrix}\in\Z[\omega]$ and
  $\G\begin{bsmallmatrix}0\\1\end{bsmallmatrix}\in\Z[\omega]$, it
  follows that the columns of $\G$ are of the form
  {\eqref{eqn-zomega-point}}, so that $\G$ is of the form
  {\eqref{eqn-gridoperator}}, with integers
  $a,b,c,d,a',b',c',d'$ satisfying $a'\equiv c'\mmod{2}$ and
  $b'\equiv d'\mmod{2}$. Moreover, we have
  \[ \G\begin{bsmallmatrix}1/\sqrt 2\\1/\sqrt 2\end{bsmallmatrix} =
  \begin{bmatrix}
    \frac{a'+b'}{2}+\frac{a+b}{\sqrt2}\\
    \frac{c'+d'}{2}+\frac{c+d}{\sqrt2}
  \end{bmatrix}\in\Z[\omega],
  \]
  which implies $a+b\equiv c+d\mmod{2}$ and $a'+b'\equiv
  c'+d'\equiv 0\mmod{2}$. Together, these conditions imply
  $a+b+c+d \equiv 0\mmod{2}$ and $a'\equiv b'\equiv c' \equiv
  d'\mmod{2}$, as claimed.
\end{proof}

\begin{remark}\label{gridoperatorcomposition}
  The composition of two (special) grid operators is again a (special)
  grid operator. If $\G$ is a special grid operator, then $\G$ is
  invertible and $\G^{-1}$ is a special grid operator. If $\G$ is a
  (special) grid operator, then $\G\bul$ is a (special) grid operator,
  defined by applying $(-)\bul$ separately to each matrix entry, and
  satisfying $\G\bul u\bul = (\G u)\bul$.
\end{remark}

The interest of special grid operators lies in the following fact:

\begin{proposition}\label{prop-operator-on-constraint}
  Let $\G$ be a special grid operator, and let $A$ and $B$ be subsets
  of $\R^2$. Define
  \[ \begin{array}{r@{~}c@{~}l}
    \G(A) &=& \s{\G u \mid u\in A}, \\
    \G\bul(B) &=& \s{\G\bul u \mid u\in B}.
  \end{array}
  \]
  Then $u\in\Z[\omega]$ is a solution to the two-dimensional grid
  problem for $A$ and $B$ if and only if $\G u$ is a solution to the
  two-dimensional grid problem for $\G(A)$ and $\G\bul(B)$. In
  particular, the two-dimensional grid problem for $A$ and $B$ is
  computationally equivalent to that for $\G(A)$ and $\G\bul(B)$.
\end{proposition}

\begin{proof}
  Let $u\in\Z[\omega]$. Then $u$ is a solution to the grid
  problem for $A$ and $B$ if and only if $u\in A$ and $u\bul\in B$,
  if and only if $\G u\in \G(A)$ and $\G\bul u\bul = (\G u)\bul \in \G\bul(B)$,
  if and only if $\G u$ is a solution to the grid problem for $\G(A)$ and
  $\G\bul(B)$.
\end{proof}

\begin{figure}
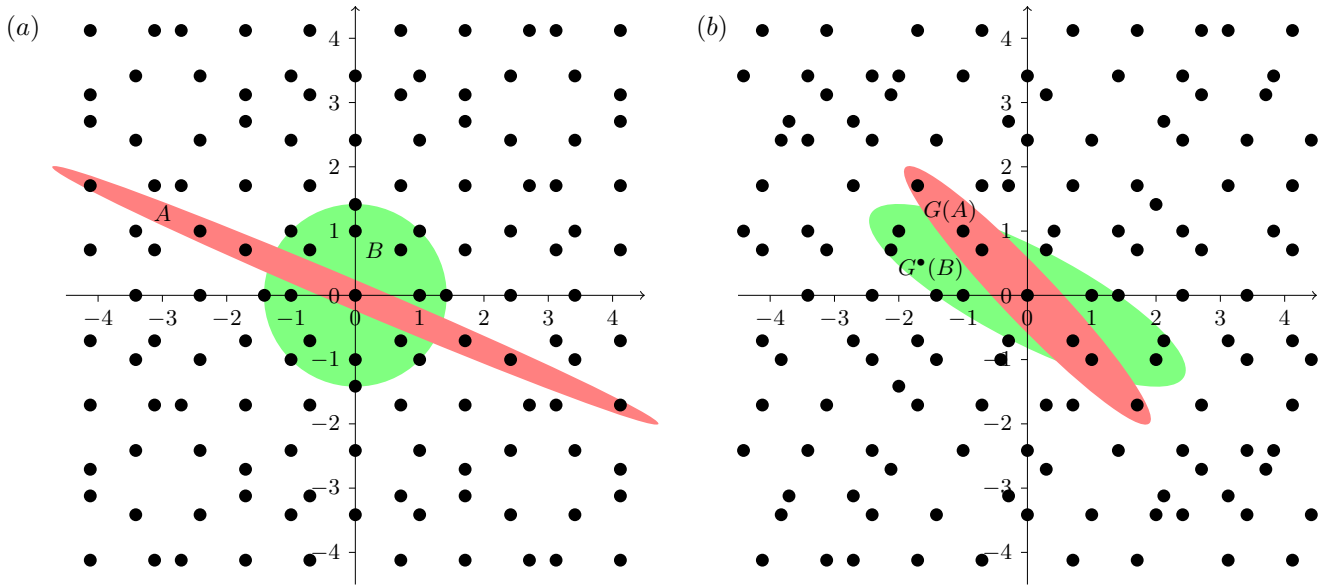

  \[ (a)~
  \mp{0.95}{\scalebox{0.95}{\input{grid-4-a.tex}}}
  \quad  
  (b)~
  \mp{0.95}{\scalebox{0.95}{\input{grid-4-b.tex}}}
  \]
  \caption{(a) The grid problem for two sets $A$ and $B$. (b) The
    grid problem with $\G(A)$ and $\G\bul(B)$. Note that the
    solutions of (a), which are the grid points in the set $A$, are
    in one-to-one correspondence with the solutions of (b), which
    are the grid points in the set $\G(A)$.}
  \label{fig-grid-4}
  \rule{\textwidth}{0.1mm}
\end{figure}

\begin{example}
  Figure~\ref{fig-grid-4}(a) illustrates the grid problem for a pair
  of sets $A$ and $B$. As before, the set $B$ is shown in green, and
  $\Grid(B)$ is shown as black dots. The set $A$ is shown in red, and
  the solutions to the grid problem are the seven grid points that lie
  in $A$.  Figure~\ref{fig-grid-4}(b) shows the grid problem for the
  sets $\G(A)$ and $\G\bul(B)$, where $\G$ is the special grid operator
  \[ \G = \left[\begin{matrix} 1 & \sqrt2 \\ 0 & 1\end{matrix}\right].
  \]
  Note that, as predicted by
  Proposition~\ref{prop-operator-on-constraint}, the solutions of the
  transformed grid problem are in one-to-one correspondence with those
  of the original problem; namely, in each case, there are seven
  solutions.
\end{example}

\subsection{Ellipses}\label{ssec-ellipses}

Combining the results of Sections~\ref{ssec-uprightness} and
{\ref{ssec-gridops}}, we know that the grid problem for convex sets
$A$ and $B$ can be solved efficiently, provided that we can find a
grid operator $\G$ such that $\G(A)$ and $\G\bul(B)$ are sufficiently
upright. Our key technical result is that in case $A$ and $B$ are
ellipses, this is always the case.

\begin{definition}
\label{def-ellipses}
  Let $D$ be a positive definite real $2\times 2$-matrix with non-zero
  determinant, and let $p\in \R^2$ be a point. The {\em ellipse}
  defined by $D$ and centered at $p$ is the set
  \[ E = \s{u\in\R^2 \mid (u-p)\da D(u-p) \leq 1}.
  \]
\end{definition}

\begin{theorem}\label{thm-ellipse}
  Suppose $A,B\subseteq \R^2$ are ellipses. Then there exists a grid
  operator $\G$ such that $\G(A)$ and $\G^\bullet (B)$ are
  $1/\formula{\oneoverM}$-upright. Moreover, if $A$ and $B$ are
  $M$-upright, then $\G$ can be efficiently computed in $O(\log(1/M))$
  arithmetic operations.
\end{theorem}

Since the proof is long and technical, we give it in
Appendix~\ref{appendix-skew-red}.

\subsection{The enclosing ellipse of a bounded convex set}

Our final step in the solution of the two-dimensional grid problem is
to generalize Theorem~\ref{thm-ellipse} from ellipses to arbitrary
bounded convex sets with non-empty interior.  This can be done because
every such set $A$ can be inscribed in an ellipse whose area is not
much greater than that of $A$, as stated in the following proposition.

\begin{proposition}\label{prop-enclosing-ellipse}
  Let $A$ be a bounded convex subset of $\R^2$ with non-empty
  interior. Then there exists an ellipse $E$ such that $A\seq E$, and
  such that 
  \[ \area(E) \leq \frac{4\pi}{3\sqrt 3}\area(A). 
  \]
\end{proposition}

The proof is in Appendix~\ref{app-enclosing-ellipse}.  Note that
$\frac{4\pi}{3\sqrt 3} \approx 2.4184$.  We remark that the bound in
Proposition~\ref{prop-enclosing-ellipse} is sharp; the bound is
attained in case $A$ is an equilateral triangle. In this case, the
enclosing ellipse is a circle, and the ratio of the areas is exactly
$\frac{4\pi}{3\sqrt 3}$.
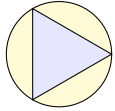
\[ 
\m{\begin{tikzpicture}[scale=0.7]
    \draw[fill=yellow!20] (0,0) circle (1);
    \draw[fill=blue!10] (1,0) -- (-0.5,.8660254037) --
    (-0.5,-.8660254037) -- cycle;
  \end{tikzpicture}}
\]

\subsection{General solution of the two-dimensional grid problem}
\label{ssec-proof-prop-algorithm-2d}

We are finally in a position to solve the two-dimensional grid problem
for arbitrary bounded convex sets of non-empty interior.  

\begin{theorem}\label{thm-main}
  There is an algorithm which, given two bounded convex subset $A$ and
  $B$ of $\R^2$ with non-empty interior, enumerates all solutions of
  the two-dimensional grid problem for $A$ and $B$. Moreover, if $A$
  and $B$ are $M$-upright, then the algorithm requires $O(\log(1/M))$
  arithmetic operations overall, plus a constant number of arithmetic
  operations per solution produced.
\end{theorem}

\begin{proof}
  Given two such sets $A$ and $B$, we can first find ellipses $A'$ and
  $B'$ containing $A$ and $B$, respectively, and whose areas exceed
  those of $A$ and $B$ by at most a fixed constant factor $N$. Such
  ellipses exist by Proposition~\ref{prop-enclosing-ellipse}; moreover,
  it is not hard to see that they can be found efficiently if $A$ and
  $B$ are polygons with rational vertices. Since we assume that each
  given convex set is equipped with such an enclosing rational polygon
  (Remark~\ref{rem-convex-specify}), $A'$ and $B'$ can be found
  efficiently for arbitrary given $A$ and $B$.
  
  Next, by Theorem~\ref{thm-ellipse}, we can use $O(\log(1/M))$
  arithmetic operations to find a grid operator $\G$ such that
  $\G(A')$ and $\G\bul(B')$ are $1/\formula{\oneoverM}$-upright. It
  follows that $\G(A)$ and $\G\bul(B)$ are
  $N/\formula{\oneoverM}$-upright. By Lemma~\ref{lem-upright-set}, we
  can efficiently enumerate all solutions $u$ of the grid problem for
  $\G(A)$ and $\G\bul(B)$. By
  Proposition~\ref{prop-operator-on-constraint}, $\G\inv u$ then
  enumerates the solutions to the grid problem for $A$ and $B$.
\end{proof}

\begin{remark}
  Note that the complexity of $O(\log(1/M))$ overall operations in
  Theorem~\ref{thm-main} is exponentially better than the
  complexity of $O(1/M^2)$ per candidate we obtained in
  Lemma~\ref{lem-upright-set}. This improvement is entirely due to the
  use of grid operators in Theorem~\ref{thm-ellipse}.
\end{remark}

\subsection{Scaled grid problems}

Sometimes we want to find solutions to a grid problem where the points
are taken in $\D[\omega]$ instead of $\Z[\omega]$. There are two
variants of this problem: we may either enumerate the solutions for a
{\em fixed} denominator exponent, or enumerate all solutions in order
of {\em increasing} least denominator exponent.

\begin{definition}
  Let $A$ and $B$ be subsets of $\R^2$. The {\em two-dimensional
    scaled grid problem for fixed $k\geq 0$} is to find
  $u\in\rtt{k}\Z[\omega]$ satisfying $u\in A$ and $u\bul\in B$.  The
  {\em two-dimensional scaled grid problem for arbitrary $k\geq 0$} is
  to find $u\in\D[\omega]$ satisfying $u\in A$ and $u\bul\in
  B$.
\end{definition}

\begin{proposition}\label{prop-scaled1}
  There is an algorithm which, given two bounded convex subsets $A$
  and $B$ of $\R^2$ with non-empty interior and an integer $k\geq 0$,
  enumerates all solutions of the two-dimensional scaled grid problem
  for $A$, $B$, and $k$. Moreover, if $A$ and $B$ are $M$-upright,
  then the algorithm requires $O(\log(1/M))$ arithmetic operations
  overall, plus a constant number of arithmetic operations per
  solution produced.
\end{proposition}

\begin{proof}
  Note that $u=\rtt{k}v$ is a solution to the scaled grid problem for
  $A$, $B$, and $k$ if and only if $v$ is a solution to the (unscaled)
  grid problem for $\rt{k}A$ and $(-\sqrt2)^kB$. The claim then immediately
  follows from Proposition~\ref{thm-main}.
\end{proof}

\begin{proposition}\label{prop-scaled-increasing}
  There is an algorithm which, given two bounded convex subsets $A$
  and $B$ of $\R^2$ with non-empty interior, enumerates (the infinite
  sequence of) all solutions $u$ of the two-dimensional scaled grid
  problem for $A$, $B$, and arbitrary $k\geq 0$, in order of
  increasing $k$. Moreover, if $A$ and $B$ are $M$-upright, then the
  algorithm requires $O(\log(1/M))$ arithmetic operations overall,
  plus a constant number of arithmetic operations per solution
  produced.
\end{proposition}

\begin{proof}
  This can be done by applying Lemma~\ref{prop-scaled1} to each
  $k=0,1,2,\ldots$, in increasing order. In principle, this method
  enumerates each solution multiple times, since each solution for $k$
  is also a solution for $k+1$. As a slight optimization, such
  duplicate enumeration can be avoided by noting that for $k>0$,
  $u=\rtt{k}(a\omega^3+b\omega^2+c\omega+d)$ is an element of
  $\Z[\omega]/\rt{k} - \Z[\omega]/\rt{\,k-1}$ if and only if $a-c$ or
  $b-d$ (or both) are odd.
  Finally, we note that, because uprightness is invariant under
  scaling, the grid operator $\G$ in the proof of
  Proposition~\ref{thm-main} only needs to be computed once,
  rather than once for every $k$.
\end{proof}

We end this section with some lower bounds on the number of
solutions to two-dimensional scaled grid problems.

\begin{lemma}\label{lem-2-solutions}
  Let $A$ and $B$ be convex subsets of $\R^2$, and let $k\geq 0$.
  Assume $A$ contains a circle of radius $r$ and $B$ contains a circle
  of radius $R$, such that $rR\geq \frac{1}{2^k}(1+\sqrt2)^2$. Then
  the scaled grid problem for $k$ has at least $2$ solutions.
\end{lemma}

\begin{proof}
  By scaling the problem by a factor of $\rt{k}$, we can assume
  without loss of generality that $k=0$. Let $\delta={r}/{\sqrt2}$
  and $\Delta=R\sqrt2$, and inscribe two squares of size
  $\delta\times\delta$ in the first circle, and one square of size
  $\Delta\times\Delta$ in the second circle, as shown here:
  
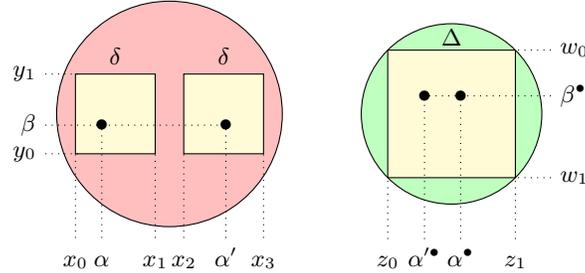
\[ \begin{tikzpicture}[scale=1.5] 
    \draw[fill=red!25] (0,0) circle (1);
    \draw[shift={(0.48,0)}, scale=0.5] (0,.707) node[above] {\small $\delta$};
    \draw[shift={(0.48,0)}, scale=0.5, fill=yellow!20] (-.707,-.707) -- (.707,-.707) -- (.707,.707) -- (-.707,.707) -- cycle;
    \draw[shift={(-0.48,0)}, scale=0.5] (0,.707) node[above] {\small $\delta$};
    \draw[shift={(-0.48,0)}, scale=0.5, fill=yellow!20] (-.707,-.707) -- (.707,-.707) -- (.707,.707) -- (-.707,.707) -- cycle;
    \draw[shift={(-0.48,0)}, scale=0.5, dotted] (-.707,-2.7) node[above, anchor=base] {\small $x_0$} +(0, 0.5) -- (-.707,-.707);
    \draw[shift={(-0.48,0)}, scale=0.5, dotted] (.707,-2.7) node[above, anchor=base] {\small $x_1$} +(0, 0.5) -- (.707,-.707);
    \draw[shift={(0.48,0)}, scale=0.5, dotted] (-.707,-2.7) node[above, anchor=base] {\small $x_2$} +(0, 0.5) -- (-.707,-.707);
    \draw[shift={(0.48,0)}, scale=0.5, dotted] (.707,-2.7) node[above, anchor=base] {\small $x_3$} +(0, 0.5) -- (.707,-.707);
    \draw[dotted] (-0.6, -1.35) node[above, anchor=base] {\small $\alpha$} +(0, 0.25) -- (-0.6, -.1);
    \draw[dotted] (0.5, -1.35) node[above, anchor=base] {\small $\alpha'$} +(0, 0.25) -- (0.5, -.1);
    \draw[shift={(-0.45,0)}, scale=0.5, dotted] (-1.3,-.707) node[left] {\small $y_0$} -- (-.707,-.707);
    \draw[shift={(-0.45,0)}, scale=0.5, dotted] (-1.3,.707) node[left] {\small $y_1$} -- (-.707,.707);
    \draw[dotted] (-1.1, -.1) node[left] {\small $\beta$} -- (0.5, -.1);
    \def\dot#1{\draw (#1) node {$\bullet$};}
    \dot{-0.6, -.1} 
    \dot{0.5, -.1} 

    \begin{scope}[shift={(2.5,0)}, scale=0.8]
    \draw[fill=green!25] (0,0) circle (1);
    \draw[fill=yellow!20, line join=round] (-.707,-.707) -- (.707,-.707) -- (.707,.707) -- (-.707,.707) -- cycle;
    \draw (0,.707) +(0,-0.05) node[above] {\small $\Delta$};
    \draw[dotted] (-.707,-1.6875) node[above, anchor=base] {\small $z_0$} +(0, 0.3) -- (-.707,-.707);
    \draw[dotted] (.707,-1.6875) node[above, anchor=base] {\small $z_1$} +(0, 0.3) -- (.707,-.707);
    \draw[dotted] (1.1,.707) node[right] {\small $w_0$} -- (.707,.707);
    \draw[dotted] (1.1,-.707) node[right] {\small $w_1$} -- (.707,-.707);
    \def\dot#1{\draw (#1) node {$\bullet$};}
    \dot{0.1, 0.2}
    \draw[dotted] (0.1, -1.6875) node[above, anchor=base] {\small $\alpha\bul$} +(0, 0.3) -- (0.1, 0.2);
    \dot{-0.3, 0.2} 
    \draw[dotted] (-0.3, -1.6875) node[above, anchor=base] {\small $\alpha'^{\bullet}$} +(0, 0.3) -- (-0.3, 0.2);
    \draw[dotted] (1.1, 0.2) node[right] {\small $\beta\bul$} -- (-0.3, 0.2);
    \end{scope}
  \end{tikzpicture}
  \]
  Since $\delta\Delta=rR\geq (1+\sqrt2)^2$, by
  Lemma~\ref{lem-grid-bounds}, we can find
  $\alpha,\alpha',\beta\in\Z[\sqrt2]$ such that
  $\alpha\in[x_0,x_1]$, $\alpha\bul\in[z_0,z_1]$,
  $\alpha'\in[x_2,x_3]$, $\alpha'^{\bullet}\in[z_0,z_1]$,
  $\beta\in[y_0,y_1]$, and $\beta\bul\in[w_0,w_1]$.
  Then $u=\alpha+i\beta$ and $v=\alpha'+i\beta$ are two different
  solutions to the two-dimensional grid problem as claimed.
\end{proof}

\begin{lemma}\label{lem-exponential-grid}
  Let $A$ and $B$ be convex subsets of $\R^2$, and assume that the
  two-dimensional scaled grid problem for $k$ has at least two
  distinct solutions. Then for all $\ell\geq 0$, the scaled grid
  problem for $k+2\ell$ has at least $2^\ell+1$ solutions.
\end{lemma}

\begin{proof}
  Let $u\neq v$ be solutions of the scaled grid problem 
  for $k$. For each $j=0,1,\ldots,2^\ell$, let $\phi=\frac{j}{2^\ell}$, 
  and consider $u_j = \phi u + (1-\phi) v$. Then $u_j$ has denominator
  exponent $k+2\ell$. Also, $u_j$ is a convex combination of $u$ and
  $v$; moreover, since $\phi\bul=\phi$, we also know that $u_j\bul =
  \phi u\bul + (1-\phi)v\bul$ is a convex combination of $u\bul$ and
  $v\bul$. Since $A$ and $B$ are convex, it follows that $u_j$ is a
  solution of the scaled grid problem for $k+2\ell$, yielding
  $2^\ell+1$ distinct solutions.
\end{proof}

\begin{remark}
  The bound in Lemma~\ref{lem-exponential-grid} is sufficient for our
  purposes, but it is not tight. In fact, the number of solutions
  grows as $O(4^k)$.
\end{remark}

\section{Solving a Diophantine equation}
\label{sec-diophantine}

We will be interested in solving equations of the following form:
given $\xi\in\D[\sqrt2]$, find $t\in\D[\omega]$ such that
\begin{equation}\label{eqn-tdatxi}
  t\da t = \xi.
\end{equation}
The following necessary condition is immediate:

\begin{lemma}[Necessary condition]\label{lem-nec-cond}
  The equation {\eqref{eqn-tdatxi}} has a solution only if $\xi\geq 0$
  and $\xi\bul\geq 0$.
\end{lemma}

\begin{proof}
  Assume $t\da t = \xi$. Since $t$ is a complex number, we have
  $\xi=t\da t\geq 0$. Similarly, since $t\bul$ is a complex number, we
  have $\xi\bul = (t\bul)\da(t\bul)\geq 0$.
\end{proof}

The following theorem states that the problem of solving the equation
{\eqref{eqn-tdatxi}} can be reduced to the prime factorization problem
for integers. 

\begin{theorem}\label{thm-diophantine}
  Let $\xi\in\D[\sqrt 2]$. Note that $\xi\bul\xi\in\D$, so we can
  write $\xi\bul\xi=\frac{n}{2^\ell}$ for some $n\in\Z$ and $\ell\in\N$. 
  There exists a probabilistic algorithm which, given $\xi$ and, in 
  case $n\neq 0$, a prime factorization of $n$, determines whether or 
  not the equation {\eqref{eqn-tdatxi}} has a solution, and finds a 
  solution if there is one. Moreover, the expected runtime of this 
  algorithm is polynomial in the size of $n$.
\end{theorem}

Theorem~\ref{thm-diophantine} is a well-known result in computational
algebraic number theory. For the benefit readers who are not experts
in number theory, we give an elementary and more or less
self-contained proof in Appendix~\ref{app-rings}.

\section{The approximate synthesis algorithm}
\label{sec-algorithm}

\subsection{The approximate synthesis problem}
\label{ssec-app-synth-problem}

Recall that the $z$-rotation by angle $\theta$ is the unitary operator
\[ \Rz(\theta) = e^{-i\theta Z/2} = \zmatrix{cc}{e^{-i\theta/2} & 0 \\
  0 & e^{i\theta/2}}.
\]
\begin{definition}\label{def-synthesis-problem}
  Given $\theta$ and a precision $\epsilon>0$, the {\em approximate
    synthesis problem for $z$-rotations} is to find an operator $U$
  expressible in the single-qubit Clifford+$T$ gate set, such that
  \begin{equation}\label{eqn-norm}
    \norm{\Rz(\theta) - U} \leq \epsilon.
  \end{equation}
  Moreover, we want the $T$-count of the operator $U$ to be as small
  as possible; here, the $T$-count of a Clifford+$T$ circuit is the
  number of $T$-gates appearing in it. The norm in {\eqref{eqn-norm}}
  is the operator norm.
\end{definition}

It is known from {\cite{Kliuchnikov-etal}} that a single-qubit
operator can be exactly represented over the Clifford+$T$ gate set if
and only if it can be written of the form
\begin{equation}\label{eqn-u-ell}
 U = \zmatrix{cc}{u & -t\da\omega^\ell \\ t & u\da\omega^\ell},
\end{equation}
where $u,t\in\D[\omega]$ and $\ell$ is an integer. The following lemma
shows that, for the purposes of approximate synthesis, we may assume
without loss of generality that $\ell=0$.

\begin{lemma}\label{lem-ell}
  If $\epsilon<|1-e^{i\pi/8}|$, then all solutions of the approximate
  synthesis problem have the form
  \begin{equation}\label{eqn-u}
    U = \zmatrix{cc}{u & -t\da \\ t & u\da}.
  \end{equation}
  If $\epsilon\geq|1-e^{i\pi/8}|$, then there exists a solution of
  $T$-count 0 (i.e., a Clifford operator), and it is also of the form
  {\eqref{eqn-u}}.
\end{lemma}

\begin{proof}
  To prove the first claim, assume $\epsilon<|1-e^{i\pi/8}|$.  Let $U$
  be of the form {\eqref{eqn-u-ell}}, satisfying
  {\eqref{eqn-norm}}. Let $e^{i\phi_1}$ and $e^{i\phi_2}$ be the
  eigenvalues of $U\Rz(\theta)\inv$, with
  $\phi_1,\phi_2\in[-\pi,\pi]$.  Using {\eqref{eqn-norm}}, we have
  $\norm{I-U\Rz(\theta)\inv}\leq\epsilon< |1-e^{i\pi/8}|$.  On the
  other hand, $\norm{I-U\Rz(\theta)\inv} = \max\s{|1-e^{i\phi_1}|,
    |1-e^{i\phi_2}|}$. It follows that $|1-e^{i\phi_j}| <
  |1-e^{i\pi/8}|$ for $j=1,2$, hence $-\pi/8<\phi_j<\pi/8$, hence
  $-\pi/4<\phi_1+\phi_2<\pi/4$, so $|1-e^{i(\phi_1+\phi_2)}| <
  |1-e^{i\pi/4}| = |1-\omega|$.  On other hand, we have
  $e^{i(\phi_1+\phi_2)} = \det(U\Rz(\theta)\inv) = \omega^\ell$, hence
  $|1-\omega^\ell|<|1-\omega|$, which implies
  $\omega^\ell=1$. Therefore, $U$ is of the form {\eqref{eqn-u}}.  To
  prove the second claim, assume $\epsilon\geq|1-e^{i\pi/8}|$. Let $j$
  be the integer closest to $\frac{2\theta}{\pi}$, so that
  $|j-\frac{2\theta}{\pi}|\leq\frac{1}{2}$, or equivalently,
  $|j\frac{\pi}{4}-\frac{\theta}{2}| \leq \frac{\pi}{8}$. Let $U$ be
  as in {\eqref{eqn-u}}, with $u=\omega^{-j}$ and $t=0$.  Then
  $\norm{\Rz(\theta) - U} = |e^{i\theta/2} - u\da| = |1 - u\da
  e^{-i\theta/2}| = |1 - e^{i(j\frac{\pi}{4}-\frac{\theta}{2})}| \leq
  |1 - e^{i\frac{\pi}{8}}| \leq \epsilon$. So {\eqref{eqn-norm}}
  holds. But $U=S^j\omega^{-j}$ is a Clifford operator, so has
  $T$-count 0.
\end{proof}

Our strategy is therefore to approximate $\Rz(\theta)$ by an operator 
$U$ of the form {\eqref{eqn-u}}, with $u,t\in\D[\omega]$, and then use 
the exact synthesis algorithm of {\cite{Kliuchnikov-etal}} to synthesize 
$U$ into a sequence of Clifford+$T$ gates with minimal $T$-count. The 
following lemma relates the $T$-count of the resulting circuit to the 
least denominator exponent $k$ of $u$.

\begin{lemma}\label{lem-2k-2}
  Let $U$ be a unitary operator as in {\eqref{eqn-u}}, where
  $u,t\in\D[\omega]$, and let $k$ be the least denominator exponent of
  $u$. Then the $T$-count of $U$ is either $2k-2$ or $2k$. Moreover,
  if $k>0$ and $U$ has $T$-count $2k$, then $U'=TUT\da$ has $T$-count
  $2k-2$.  We further note that $\norm{\Rz(\theta) - U'} =
  \norm{\Rz(\theta) - U}$, so for the purpose of solving
  {\eqref{eqn-norm}}, it does not matter whether $U$ or $U'$ is
  used. Hence, without loss of generality, we may assume that $U$ as
  in {\eqref{eqn-u}} always has $T$-count exactly $2k-2$ when $k>0$,
  and $0$ when $k=0$.
\end{lemma}

\begin{proof}
  Because $U$ is unitary, we have $t\da t+u\da u=1$. We first claim
  that $t$ and $u$ have the same least denominator exponent. Indeed,
  in the ring $\Z[\omega]$, an element $s$ is divisible by $\sqrt2$ if and
  only if ${s}\da s$ is divisible by 2. The left-to-right implication
  is obvious, and the right-to-left implication follows, e.g., from
  Lemma~2 of {\cite{Giles-Selinger}}. Then for any $k\geq 0$, we have
  $\rt{k}u\in\Z[\omega]$ iff $2^k{u}\da u\in\Z[\omega]$ iff
  $2^k(1-{t}\da t)\in\Z[\omega]$ iff $2^k{t}\da {t}\in\Z[\omega]$ iff
  $\rt{k}{t}\in\Z[\omega]$. This proves that $u$ and $t$ have the same
  denominator exponents, and in particular, the same least denominator
  exponent.

  The claims about the $T$-counts of $U$ and $U'$ follow by inspection
  of Figure~2 of {\cite{ma-remarks}}. Using the terminology of
  Definitions~7.4 and 7.6 of {\cite{ma-remarks}}, this figure shows
  every possible $k$-residue of a Clifford+$T$ operator, modulo a
  right action of the group $\groupspan{S,X,\omega}$. Because $U$ is
  of the form {\eqref{eqn-u}}, only a subset of the $k$-residues is
  actually possible, and the figure shows that for this subset, the
  $T$-count is $2k$ or $2k-2$. Moreover, in each of the possible cases
  where $k>0$ and $U$ has $T$-count $2k$, the figure also shows that
  $U'=TUT\da$ has $T$-count $2k-2$.

  For the final claim, we have $\norm{\Rz(\theta) -
    U}=\norm{T\Rz(\theta)T\da - TUT\da} = \norm{\Rz(\theta) - U'}$
  because $\Rz(\theta)$ and $T$ commute.
\end{proof}

So our task is to find an operator $U$ of the form {\eqref{eqn-u}},
satisfying {\eqref{eqn-norm}}, and such that the denominator exponent 
of $u$ is as small as possible. It is useful to first re-express 
{\eqref{eqn-norm}} as a property of $u$. Let $z=e^{-i\theta/2}$. Using
$u\da u+t\da t=1$ and $z\da z=1$, we have
\[ \norm{\Rz(\theta) - U}^2
  = \norm{u-z}^2 + \norm{t}^2
  = (u-z)\da(u-z) + {t}\da t
  = {u}\da u+{t}\da t - z\da u - {u}\da z + z\da z
  = 2 - 2\Realpart(z\da u).
\]
So {\eqref{eqn-norm}} is equivalent to $2 - 2\Realpart(z\da u)\leq
\epsilon^2$, or equivalently, $\Realpart(z\da u) \geq 1 -
\frac{\epsilon^2}{2}$. If we identify the complex numbers $z=x+yi$ and
$u=a+bi$ with 2-dimensional real vectors $\vec z = (x,y)^T$ and $\vec u
= (a,b)^T$, then $\Realpart(z\da u)$ is just their inner product $\vec
z\cdot \vec u$, and therefore {\eqref{eqn-norm}} is equivalent to
\begin{equation}\label{eqn-zu}
  \vec z\cdot \vec u \geq 1 - \frac{\epsilon^2}{2}.
\end{equation}
In summary, the approximate synthesis problem reduces to the following:

\begin{problem}\label{prob-main}
  Given an angle $\theta$ and a precision $\epsilon>0$, find
  $u,t\in\D[\omega]$ such that
  \begin{enumerate}\alphalabels
  \item ${t}\da t+{u}\da u=1$,
  \item $\vec z\cdot \vec u \geq 1 - \epsilon^2/2$, where
    $z=e^{-i\theta/2}$, with notation as above,
  \item and such that $u$ has the smallest denominator exponent we can find.
  \end{enumerate}
\end{problem}

\subsection{Reduction to a grid problem and a Diophantine equation}
\label{ssec-reduction}

As we will now show, Problem~\ref{prob-main} can be reduced to a
scaled grid problem and a Diophantine equation, and therefore it can
be efficiently solved by the methods of Sections~\ref{sec-grid-2d} and
{\ref{sec-diophantine}}. Let $\Disk$ be the closed unit disk, regarded
either as a subset of $\C$ or of $\R^2$.

\begin{lemma}\label{lem-reduction-a}
  If $u,t\in\D[\omega]$ and ${t}\da {t}+{u}\da u=1$, then $u\in\Disk$
  and ${u}\bul\in\Disk$.
\end{lemma}

\begin{proof}
  Note that ${u}\da u = 1-{t}\da t \leq 1$, so $u\in\Disk$.
  Similarly, $({u}\bul)\da({u}\bul) = 1-({t}\bul)\da({t}\bul) \leq 1$,
  so ${u}\bul\in\Disk$.
\end{proof}

The condition $\vec z\cdot \vec u \geq 1 - \epsilon^2/2$ from
Problem~\ref{prob-main}(b) defines a certain subset of the unit disk,
which we call the {\em $\epsilon$-region} for $\theta$:
\begin{equation}\label{eqn-eps-region}
  \Repsilon = \s{\vec u\in\Disk \mid \vec u\cdot \vec z\geq 1 -
    \frac{\epsilon^2}{2}}.
\hspace{2cm}
\begin{tikzpicture}[scale=1.8, baseline=0]
  \fill[fill=blue!10, yscale=-1, rotate=35] (cos 50, sin 50) -- (cos 50, -sin 50) arc (-50:50:1) -- cycle;
  \path[color=gray] (-.3,.4) node {$\Disk$};
  \draw[color=gray,->] (-1.2,0) -- (1.4,0);
  \draw[color=gray,->] (0,-1.1) -- (0,1.2);
  \path[color=gray] (0,1) node[above left] {$i$};
  \path[color=gray] (1,0) node[above=6pt,right=-2pt] {$1$};
  \draw[color=gray] (0,0) circle (1);
  \draw[yscale=-1, rotate=35] (cos 50, -sin 50) arc (-50:50:1);
  \draw[->, yscale=-1, rotate=35] (0,0) -- (1,0) node[right] {$\vec z$};
  \draw[yscale=-1, rotate=35] (cos 50, sin 50) -- (cos 50,-1);
  \draw[yscale=-1, rotate=35] (1,sin 50) -- (1,-1);
  \draw[yscale=-1, rotate=35] (cos 50, sin 50) -- (1, sin 50);
  \draw[yscale=-1, rotate=35] (cos 50, -sin 50) -- (1, -sin 50);
  \draw[<->, yscale=-1, rotate=35] (cos 50,-0.9) -- (1,-0.9);
  \draw[yscale=-1, rotate=35] (0.9, -0.9) node[above] {$\frac{\epsilon^2}{2}$};
  \path[yscale=-1, rotate=35] (0.8,-0.2) node {$\Repsilon$};
  \draw (0.5,0) arc (0:-35:0.5);
  \path (-17.5:0.4) node {\small $\frac{\theta}{2}$};
\end{tikzpicture}
\end{equation}
By Lemma~\ref{lem-reduction-a} and Problem~\ref{prob-main}(b), a {\em
  necessary} condition for a solution to Problem~\ref{prob-main} is
that $u\in\Repsilon$ and ${u}\bul\in\Disk$. This is an instance of a
scaled two-dimensional grid problem; note that both the
$\epsilon$-region and the unit disk are convex, and are effectively
given in the sense of Remark~\ref{rem-convex-specify}. Therefore, by
Proposition~\ref{prop-scaled-increasing}, there exists an efficient
algorithm that enumerates all such $u$ in increasing order of least
denominator exponent.

For each $u\in\D[\omega]$ satisfying the grid problem, it remains to
check whether the equation ${t}\da t+{u}\da u=1$ has a solution
$t\in\D[\omega]$. This is equivalent to solving the Diophantine equation
\[ t\da t = 1 - u\da u,
\]
which is of the form {\eqref{eqn-tdatxi}}. Let $\xi=1 - u\da
u\in\D[\sqrt2]$, and write $\xi\bul\xi=\frac{n}{2^\ell}$, where
$n\in\Z$ and $\ell\in\N$. By Theorem~\ref{thm-diophantine}, there is
an efficient algorithm that can solve this equation (or determine that
no solution exists), given a prime factorization of $n$.

\subsection{The main algorithm}

Putting together the results of Sections~\ref{ssec-app-synth-problem}
and {\ref{ssec-reduction}}, we obtain the following algorithm for
solving the approximate synthesis problem:

\begin{algorithm}\label{alg-main}
  Given $\theta$ and $\epsilon$, let $A=\Repsilon$ be the
  $\epsilon$-region, and let $B=\Disk$ be the unit
  disk. 
  \begin{enumerate}
  \item[1.] Use Proposition~\ref{prop-scaled-increasing} to enumerate
    the infinite sequence of solutions to the scaled grid problem
    $u\in A$ and ${u}\bul\in B$, where $u\in\D[\omega]$, in the order
    of increasing least denominator exponent $k$. 
  \item[2.] For each such solution $u$:
    \begin{enumerate}
    \item[(a)] Let $\xi=1 - u\da u\in\D[\sqrt2]$, and write
      $\xi\bul\xi=\frac{n}{2^\ell}$, where $n\in\Z$ and $\ell\geq 0$
      is minimal.
    \item[(b)] Attempt to find a prime factorization of
      $n$. If $n\neq 0$ but no prime factorization is found, skip step
      2(c) and continue with the next $u$.
    \item[(c)] Use the algorithm of Theorem~\ref{thm-diophantine} to
      solve the equation $t\da t = \xi$. If a solution $t$ exists,
      go to step 3; otherwise, continue with the next $u$.
    \end{enumerate}
  \item[3.] Define $U$ as in equation {\eqref{eqn-u}}, let
      $U'=TUT\da$, and use the exact synthesis algorithm of
      {\cite{Kliuchnikov-etal}} to find a Clifford+$T$ circuit
      implementing either $U$ or $U'$, whichever has smaller
      $T$-count. Output this circuit and stop.
  \end{enumerate}
\end{algorithm}

\section{Analysis of the algorithm}
\label{sec-analysis}

\subsection{Correctness}

\begin{proposition}[Correctness]
  If Algorithm~\ref{alg-main} terminates, then it yields a valid
  solution to the approximate synthesis problem, i.e., it yields a
  Clifford+$T$ circuit approximating $\Rz(\theta)$ up to $\epsilon$.
\end{proposition}

\begin{proof}
  By construction. By steps 2(a) and 2(c) of the algorithm, we have
  ${t}\da t+{u}\da u=1$, so $U$ is unitary. By step 1 of the
  algorithm, $u$ belongs to the $\epsilon$-region, so
  {\eqref{eqn-zu}} holds. This implies that $U$ satisfies
  {\eqref{eqn-norm}}. Moreover, as noted in Lemma~\ref{lem-2k-2}, $U'$
  also satisfies {\eqref{eqn-norm}}, so whichever of these operators
  is returned approximates $\Rz(\theta)$ up to $\epsilon$.
\end{proof}

\subsection{Optimality}

The optimality of the algorithm hinges on step 2(b), ``attempt to find
a prime factorization of $n$''. In the presence of a factoring oracle
(such as a quantum computer), this can always be done. In this case,
Algorithm~\ref{alg-main} is guaranteed to find an optimal solution to
the approximate synthesis problem. In the absence of a factoring
oracle, we must potentially discard some candidate solutions $u$,
until we find one for which $n$ can be factored. We analyze these two
situations in Propositions~\ref{prop-optimality} and
{\ref{prop-near-optimality}}.

\begin{proposition}[Optimality in the presence of a factoring oracle]
  \label{prop-optimality}
  In the presence of an oracle for integer factoring, the circuit
  returned by Algorithm~\ref{alg-main} has the smallest $T$-count of
  any single-qubit Clifford+$T$ circuit approximating $\Rz(\theta)$ up
  to $\epsilon$.
\end{proposition}

\begin{proof}
  By construction, step 1 of the algorithm enumerates all solutions
  $u$ of the scaled grid problem in order of increasing denominator
  exponent $k$. Step 2(a) always succeeds, and step 2(b) always
  succeeds by using the factoring oracle. By
  Theorem~\ref{thm-diophantine}, step 2(c) succeeds if and only if the
  equation $t\da t+u\da u=1$ has a solution. Therefore, when step 2(c)
  succeeds, the algorithm has found a solution of
  Problem~\ref{prob-main} for which $u$ has the lowest possible
  denominator exponent $k$.  Let $m$ be the $T$-count of the final
  solution. By Lemma~\ref{lem-2k-2}, we have $m=2k-2$, except when
  $k=0$, in which case $m=0$.
  
  To show that this $T$-count is minimal, let $\bar U$ be any solution
  of the approximate synthesis problem with $T$-count $\bar m$. By
  Lemma~\ref{lem-ell}, we may assume without loss of generality that
  \[ \bar U = \zmatrix{cc}{\bar u & -\bar t\da \\ \bar t & \bar u\da}.
  \]
  Let $\bar k$ be the denominator exponent of $\bar u$. By minimality
  of $k$, we have $k\leq \bar k$, hence $m\leq \bar m$. 
\end{proof}

We emphasize that the optimality in Proposition~\ref{prop-optimality}
is {\em absolute}, i.e., not merely asymptotic or up to an additive
constant.  Of all the Clifford+$T$ operators approximating
$\Rz(\theta)$ to within $\epsilon$, the algorithm finds one with the
lowest $T$-count.

To analyze the algorithm in the absence of a factoring oracle, we must
address the question of how many candidates must be generated before
steps 2(b) and 2(c) of the algorithm succeed. In this case, the
algorithm may still use any classical algorithm to try to factor the
number $n$ in step 2(b), but the amount of effort extended on any
particular $n$ must be capped. In our complexity analysis for this
case, in Proposition~\ref{prop-near-optimality} below, we
conservatively assume that the only $n$ that the algorithm can
successfully factor are those $n$ that are already prime. In reality
the algorithm might do a little better. In order to complete the
analysis, we must rely on a number-theoretic assumption about the
distribution of prime numbers.

\begin{hypothesis}\label{hyp-numbers}
  Each number $n$ produced in step 2(a) of Algorithm~\ref{alg-main} is
  asymptotically as likely to be prime as a randomly chosen odd number
  of comparable size. Moreover, the primality of each $n$ can be
  modelled as an independent random event.
\end{hypothesis}

\begin{lemma}\label{lem-n}
  Each of the numbers $n$ produced in step 2(a) of
  Algorithm~\ref{alg-main} satisfies $n\geq 0$, and either $n=0$ or
  $n\equiv1\mmod{8}$.
\end{lemma}

\begin{proof}
  See Appendix~\ref{appendix-n}.
\end{proof}

\begin{lemma}\label{lem-n2}
  Let $u$ be a candidate produced in step 1 of
  Algorithm~\ref{alg-main}, let $k$ be its least denominator exponent,
  and let $n$ be the integer computed in step 2(a).
  Then $n\leq 4^k$.
\end{lemma}

\begin{proof}
  By assumption, we can write $u={v}/{\rt{k}}$, where
  $v\in\Z[\omega]$.  From step 2(a) of the algorithm, we have $\xi =
  1-u\da u = \frac{1}{2^k}(2^k-v\da v) = \frac{\alpha}{2^k}$, where
  $\alpha\in\Z[\sqrt2]$. Therefore $\xi\bul\xi =
  \frac{\alpha\bul\alpha}{2^{2k}} = \frac{n}{2^\ell}$. Since
  $\alpha\bul\alpha$ is an integer and $\ell$ is minimal, we have
  $\ell\leq 2k$. Also, by assumption, both $u$ and $u\bul$ are in the
  unit disk, so $u\da u\leq 1$ and $(u\bul)\da(u\bul)\leq 1$. It
  follows that $0\leq \xi,\xi\bul\leq 1$, hence $\xi\bul\xi\leq 1$.
  Therefore, $\frac{n}{2^\ell}\leq 1$, which implies $n\leq 2^\ell
  \leq 4^k$ as claimed.
\end{proof}

\begin{lemma}\label{lem-summation}
  Let $b>0$ be an arbitrary fixed constant. Then for $a\geq 1$,
  \[ \sum_{x=1}^{\infty} \bigparen{1-\frac{1}{a+b\ln x}}^x = O(a).
  \]
\end{lemma}

\begin{proof}
  See Appendix~\ref{app-d}.
\end{proof}

\begin{definition}
  Let $U',U''$ be two solutions of the approximate synthesis problem 
  of the form
  \begin{equation}\label{eqn-u1-u2}
    U' = \zmatrix{cc}{u' & -t'^{\dagger} \\ t' & u'^{\dagger}},
    \quad
    U'' = \zmatrix{cc}{u'' & -t''^{\dagger} \\ t'' & u''^{\dagger}}.
  \end{equation}
  We say that $U'$ and $U''$ are {\em equivalent solutions} if
  $u'=u''$.
\end{definition}

\begin{proposition}[Near-optimality in the absence of a factoring oracle]
  \label{prop-near-optimality}
  Let $m$ be the $T$-count of the solution of the approximate
  synthesis problem found by Algorithm~\ref{alg-main} in the absence
  of an oracle for integer factoring. Then
  \begin{enumerate}\alphalabels
  \item The approximate synthesis problem has an expected number of at
    most $O(\log(1/\epsilon))$ non-equivalent solutions with $T$-count
    less than $m$.
  \item The expected value of $m$ is $m'' +
    O(\log(\log(1/\epsilon)))$, where $m'$ and $m''$ are the
    $T$-counts of the optimal and second-to-optimal solutions of the
    approximate synthesis problem (up to equivalence), and $m'\leq m''$.
  \end{enumerate}
\end{proposition}

\begin{proof}
  If $\epsilon\geq|1-e^{i\pi/8}|$, then by Lemma~\ref{lem-ell}, there
  is a solution with $T$-count 0, and the algorithm easily finds
  it. In this case, there is nothing to show. So assume without loss
  of generality that $\epsilon<|1-e^{i\pi/8}|$. Then by
  Lemma~\ref{lem-ell}, all solutions are of the form {\eqref{eqn-u}}.

  (a) Consider the list $u_1,u_2,\ldots$ of candidates generated in
  step 1 of the algorithm. Let $k_1,k_2,\ldots$ be their respective
  least denominator exponents, and let $n_1,n_2,\ldots$ be the
  corresponding integers calculated in step 2(a). By
  Lemma~\ref{lem-n2}, we have $n_j\leq 4^{k_j}$ for all $j$. By
  Hypothesis~\ref{hyp-numbers}, the probability that $n_j$ is prime is
  asymptotically no smaller than that of a randomly chosen odd integer
  less than $4^{k_j}$, which, by the well-known prime number theorem,
  is greater than
  \begin{equation}\label{eqn-pj}
    p_j:=\frac{2}{\ln (4^{k_j})} = \frac{1}{k_j\ln 2}.
  \end{equation}
  Note that $u_1$ and $u_2$ are two distinct solutions to the scaled
  grid problem of step 1 of the algorithm. Since the candidates are
  enumerated in order of increasing denominator exponent, $k_2$ is a
  denominator exponent for both $u_1$ and $u_2$. It follows by
  Lemma~\ref{lem-exponential-grid} that there are at least $2^\ell+1$
  distinct candidates of denominator exponent $k_2+2\ell$, for all
  $\ell\geq 0$. In other words, for all $j$, if $j\leq 2^\ell+1$, we
  have $k_j\leq k_2+2\ell$. In particular, this holds for
  $\ell=\floor{1+\log_2 j}$, and therefore, 
  \begin{equation}\label{eqn-kj}
    k_j\leq k_2+2(1+\log_2 j).
  \end{equation}
  Combining {\eqref{eqn-kj}} with {\eqref{eqn-pj}}, we have
  \begin{equation}\label{eqn-pj2}
    p_j \geq \frac{1}{(k_2+2(1+\log_2 j))\ln 2}
    = \frac{1}{(k_2+2)\ln 2 + 2\ln j}
  \end{equation}
  Let $j_0$ be the smallest index such that $n_{j_0}$ is prime.
  By Hypothesis~\ref{hyp-numbers}, we can treat the primality of each
  $n_j$ as an independent random event. Therefore, 
  \begin{eqnarray}
    P(j_0 > j) &=& P(\mbox{$n_1,\ldots,n_j$ are not prime})  \nonumber\\
    &\leq& (1-p_1)(1-p_2)\cdots(1-p_j) \nonumber\\
    &\leq& (1-p_j)^j \nonumber\\
    &\leq& \bigparen{1-\frac{1}{(k_2+2)\ln 2 + 2\ln j}}^j. \nonumber
  \end{eqnarray}
  The expected value of $j_0$ is
  \begin{equation}\label{eqn-ej0}
    E(j_0) ~=~ \sum_{j=0}^{\infty} P(j_0 > j)
    ~\leq~ 1+\sum_{j=1}^{\infty} \bigparen{1-\frac{1}{(k_2+2)\ln 2 + 2\ln
        j}}^j
    ~=~ O(k_2),
  \end{equation}
  where we have used Lemma~\ref{lem-summation} to estimate the sum.

  Next, we will estimate $k_2$.  The width of the $\epsilon$-region
  $\Repsilon$, as shown in {\eqref{eqn-eps-region}}, is
  $\epsilon^2/2$ at the widest point, and we can inscribe a disk of
  radius $r={\epsilon^2}/{4}$ in it. Also, the closed unit disk
  $\Disk$ has radius $R=1$. It follows by Lemma~\ref{lem-2-solutions}
  that the scaled grid problem for $\Repsilon$ and $\Disk$, as in step
  1 of the algorithm, has at least two solutions, provided that
  \begin{equation}\label{eqn-rR}
    rR = \frac{\epsilon^2}{4} \geq \frac{(1+\sqrt2)^2}{2^k},
  \end{equation}
  or equivalently, provided that
  \begin{equation}\label{eqn-kgeq}
    k\geq 2 + 2 \log_2(1+\sqrt2) + 2\log_2 ({1}/{\epsilon}).
  \end{equation}
  We therefore have $k_2\leq k$ for {\em all} $k$ satisfying
  {\eqref{eqn-kgeq}}. It follows that
  \begin{equation}\label{eqn-k2}
    k_2=O(\log({1}/{\epsilon})),
  \end{equation}
  and therefore, using {\eqref{eqn-ej0}}, also
  \begin{equation}\label{eqn-ej0-eps}
    E(j_0)=O(\log({1}/{\epsilon})).
  \end{equation}
  
  To finish the proof of part (a), recall that $j_0$ was defined to be
  the smallest index such that $n_{j_0}$ is prime. This ensures that
  step 2(b) of the algorithm succeeds for the candidate $u_{j_0}$.
  Furthermore, we have $n\equiv 1\mmod{8}$ by Lemma~\ref{lem-n}, and
  therefore the equation $t\da t=\xi$ has a solution by
  Proposition~\ref{prop-prime-1mod8}. Hence the remaining steps of the
  algorithm also succeed for $u_{j_0}$.

  Now let $r$ be the number of non-equivalent solutions of the
  approximate synthesis problem of $T$-count strictly less than
  $m$. As noted above, any such solution $U$ is of the form
  {\eqref{eqn-u}}. Then the least denominator exponent of $u$ is
  strictly smaller than $k_{j_0}$, so that $u=u_j$ for some
  $j<j_0$. In this way, each of the $r$ non-equivalent solutions is
  mapped to a different index $j<j_0$. It follows that $r<j_0$, and
  hence $E(r)\leq E(j_0)=O(\log({1}/{\epsilon}))$, as was to be shown.

  (b) Let $U'$ be an optimal solution of the approximate synthesis
  problem, and let $U''$ be optimal among the solutions that are not
  equivalent to $U'$. Let $u'$ and $u''$ be as in {\eqref{eqn-u1-u2}},
  and let $k'$, $k''$ be the least denominator exponents of $u'$ and
  $u''$, respectively, with $k'\leq k''$. Let $k_2$ and $j_0$ be as in
  the proof of part (a). Note that, by definition, $k_2\leq k''$. Let
  $k$ be the least denominator exponent of the solution of the
  approximate synthesis problem found by the algorithm. Then $k\leq
  k_{j_0}$.  Using {\eqref{eqn-kj}}, we have
  \[ k \leq k_{j_0} \leq k_2 + 2(1+\log_2 j_0) \leq k'' +  2(1+\log_2 j_0).
  \]
  Recall from Lemma~\ref{lem-2k-2} that $2k-2 \leq m \leq 2k$, and
  similarly $2k''-2\leq m''\leq 2k''$. Hence $m \leq 2k\leq
  2k''+4(1+\log_2 j_0)\leq m'' + 6 + 4\log_2
  j_0$. These calculations apply to any one run of the
  algorithm. Taking expected values over many randomized runs, we
  therefore have
  \begin{equation}\label{eqn-em}
    E(m) \leq m'' + 6 + 4 E(\log_2 j_0) 
    \leq m'' + 6 + 4\log_2 E(j_0).
  \end{equation}
  Note that we have used the law $E(\log j_0)\leq \log(E(j_0))$, which
  holds because $\log$ is a concave function. Combining {\eqref{eqn-em}}
  with {\eqref{eqn-ej0-eps}}, we therefore have the desired result:
  \begin{equation}
    E(m) = m'' + O(\log(\log(1/\epsilon))).
  \end{equation}
\end{proof}

\begin{remark}\label{rem-lower-bound}
  In the near-optimal case of Proposition~\ref{prop-near-optimality},
  our algorithm can additionally be used to compute a firm lower bound
  for the $T$-count of any solution of the approximate synthesis
  problem for the given $\theta$ and $\epsilon$. Namely, the algorithm
  can consider the first candidate $u_j$ for which the Diophantine
  equation solver does not fail --- i.e., either it solves the
  equation or it times out. If $k_j$ is the least denominator exponent
  of $u_j$, then a lower bound for the $T$-count is $2k_j-2$ (or $0$
  when $k_j=0$).  Note that this is not the usual
  information-theoretic lower bound that applies in the average case,
  but an actual lower bound for each particular problem instance.
\end{remark}

\subsection{Worst-case behavior}

In {\cite{Selinger-newsynth}}, an approximate synthesis algorithms for
$z$-rotations was given that always returns a solution of $T$-count at
most $K+4\log_2(1/\epsilon)$, where $K$ is a constant approximately
equal to $10$.  We note that Algorithm~\ref{alg-main} enumerates all
the solutions of the grid problem for the $\epsilon$-region, whereas
the algorithm of {\cite{Selinger-newsynth}} only enumerates a subset
of the solutions.  Also, Algorithm~\ref{alg-main} can solve the
Diophantine equation in all the cases in which the algorithm of
{\cite{Selinger-newsynth}} can solve it. It follows that in all cases,
the solution returned by Algorithm~\ref{alg-main} is at least as good
as that returned by the algorithm of {\cite{Selinger-newsynth}}. In
particular, Algorithm~\ref{alg-main} always returns a solution of
$T$-count at most $K+4\log_2(1/\epsilon)$.  Moreover, it is known from
{\cite[Section 9]{Selinger-newsynth}}, that there are certain
combinations of $\theta$ and $\epsilon$ for which this $T$-count is
actually optimal to within a constant number of gates. Thus our
algorithm's performance is $K+4\log_2(1/\epsilon)$ in the worst case,
but this worst case behavior is only achieved in those cases where it
is actually optimal.

One may ask whether a more precise statement can be made about the
angles $\theta$ for which this worst-case behavior occurs. We believe
that this is indeed possible. Let $\Q(\sqrt 2)$ be the field of
rational numbers extended with $\sqrt 2$.

\begin{conjecture}\label{con-theta}
  Fix an angle $\theta$, and consider the $T$-count as
  $\epsilon\to 0$. Then the $T$-count behaves as
  $K+4\log_2(1/\epsilon)$ when $\tan\frac{\theta}{2}\in \Q(\sqrt 2)$,
  and as $K+3\log_2(1/\epsilon)$ otherwise.
\end{conjecture}

Giving a rigorous proof of this conjecture is beyond the scope of this
paper and probably difficult, but we will sketch a plausibility
argument. We believe that the first part, showing that the $T$-count
behaves no better than $K+4\log_2(1/\epsilon)$ when
$\tan\frac{\theta}{2}\in \Q(\sqrt 2)$, could be made rigorous without
too much difficulty. In addition, we also offer experimental evidence
supporting Conjecture~\ref{con-theta} in
Section~\ref{sec-experimental}.

\begin{figure}
  \[ (a)~
  \mp{0.95}{\input{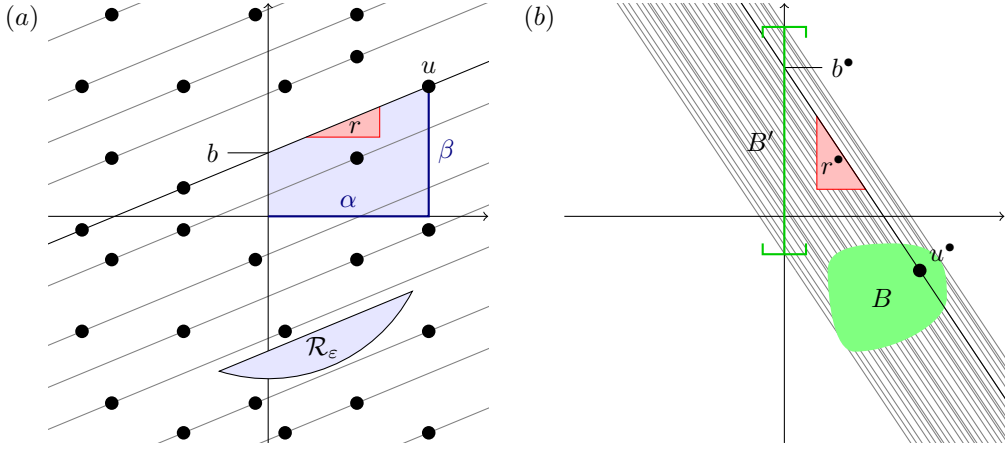}}
  \quad (b)~
  \mp{0.95}{\begin{tikzpicture}[scale=0.9]
\def\gridline#1{\begin{scope}[on background layer]\draw[thin,gray] #1 -- +(5*2,-5*3);%
\draw[thin,gray] #1 -- +(-5*2,5*3);
\end{scope}
}
\def\blackgridline#1{\draw[thin] #1 -- +(5*2,-5*3);%
\draw[thin] #1 -- +(-5*2,5*3);}
\def\mycircle{circle (0.1/0.8);}
\path[clip] (-2.0,-2.1) -- (-2.0,4.4) -- (4.5,4.4) -- (4.5,-2.1) -- cycle;
\begin{scope}[on background layer]
\path[clip] (-2.0,-2.1) -- (-2.0,4.4) -- (4.5,4.4) -- (4.5,-2.1) -- cycle;
\end{scope}
\draw[->] (-2.0,1.25) -- (4.5, 1.25);
\draw[->] (1.25,-2.1) -- (1.25,4.4);
\begin{scope}[shift={(1.25,1.25)}]
\begin{scope}[scale=0.8]
\fill[color=green!50] (2,-0.5) .. controls (3,-0.5) and (3,-1) .. (3,-1.5)
               .. controls (3,-2) and (2,-2.5) .. (1.3,-2.5)
               .. controls (0.9,-2.5) and (0.8,-1.5) .. (0.8,-1)
               .. controls (0.8,-0.6) and (1.3,-0.5) .. (2,-0.5) -- cycle;
\fill (2.5,-1) \mycircle;
\path (2.5,-1) node[above right] {$u\bul$};
\foreach\x in {0.20958187289913877,0.3141695875350399,0.418757302170941,0.5666666666666647,0.6712543813025724,0.7758420959384611,0.8191637457982817,0.923751460434191,1.0283391750700797,1.176248539565803,1.2808362542017107,1.3854239688376246,1.4287456186974248,1.5333333333333292,1.6379210479692432,1.7858304124649476,1.8904181271008618,1.995005841736763,2.0383274915965663,2.1429152062324675,2.247502920868382,2.2908245707281876,2.3954122853640865,2.5,2.6479093644957112,2.752497079131619,2.85708479376752,2.90040644362733,3.0049941582632376} {
  \gridline{(\x,-1)};
}
\draw[red,fill=pink] 
   (2.5-0.5*2,-1+0.5*3) --
   (2.5-0.95*2,-1+0.95*3) --
   (2.5-0.95*2,-1+0.5*3) --
   cycle;
\path[black] (2.5-0.8*2,-1+0.65*3) node {$r\bul$};
\blackgridline{(2.5,-1)};
\path (1.8,-1.5) node {$B$};
\draw[black] (0,2.75) -- (0.7,2.75) node [right] {$b\bul$};
\draw[thick,color=green!80!black] (0,3.5) +(-0.4,-0.2) -- +(-0.4,0) -- +(0.4,0) -- +(0.4,-0.2);
\draw[thick,color=green!80!black] (0,-0.7) +(-0.4,0.2) -- +(-0.4,0) -- +(0.4,0) -- +(0.4,0.2);
\draw[thick,color=green!80!black] (0,3.5) -- (0,-0.7);
\path (0,1.4) node [left] {$B'$};
\end{scope}
\end{scope}

\end{tikzpicture}}
  \]
  \caption{Part (a) shows the grid points and grid lines (of slope
    $r\in\Q(\sqrt 2)$) for a convex set $B$. Part (b) shows the set
    $B$ and the dual grid lines (of slope $r\bul$), all of whose
    $y$-intercepts lie in an interval $B'$.  }
  \label{fig-x}
\end{figure}

The essential point is this. Consider the grid associated to a convex
bounded set $B$ as in Figure~\ref{fig-x}. Note that the slope $r$ of
any line passing through two grid points satisfies $r\in\Q(\sqrt 2)$
(unless the line is vertical, in which case the slope is infinite).
Conversely, fix some real number $r$.  Let us call a line of slope $r$
a {\em grid line} if it passes through at least one grid point. We
first wish to show that if $r\in\Q(\sqrt2)$, then the set of grid
lines is discrete, i.e., there is a minimum distance between each pair
of grid lines. To see why this is so, note that
$r\in\frac{1}{n}\Z[\sqrt2]$ for some $n\in\Z$. Each grid line is
uniquely determined by its $y$-intercept, i.e., its point of
intersection with the $y$-axis.  Consider a grid line with
$y$-intercept $b$, and passing through the grid point
$u\in\Z[\omega]$. Using Lemma~\ref{lem-zomega}, we can assume without
loss of generality that $u=\alpha+\beta i$, where
$\alpha,\beta\in\Z[\sqrt2]$; the case where $u=\alpha+\beta i+\omega$
is analogous. Then we have $b=\beta-r\alpha$, and therefore
$b\in\frac{1}{n}\Z[\sqrt2]$. Moreover, the line with $y$-intercept
$b\bul$ and slope $r\bul$ contains $u\bul$, and therefore intersects
$B$. Since $B$ is bounded, it follows that $b\bul$ lies in a certain
finite interval $B'$ of real numbers determined only by $B$ and
$r\bul$.  The two relations $b\in\frac{1}{n}\Z[\sqrt2]$ and
$b\bul\in B'$ determine a one-dimensional grid problem, and therefore
the set of solutions is discrete. This shows that the set of grid
lines is discrete when $r\in\Q(\sqrt2)$.

Recall from {\eqref{eqn-eps-region}} that for given $\theta$ and
$\epsilon$, the $\epsilon$-region is bounded by a circular arc and a
straight line. The slope of the straight line is
$r=-1/\tan\frac{\theta}{2}$, and is independent of $\epsilon$. If
$r\in\Q(\sqrt2)$, then the $\epsilon$-region is ``parallel'' to the
grid lines of slope $r$, as shown in Figure~\ref{fig-x}(a). In this
case, the number of solutions to the grid problem is not dominated by
the {\em area}, but by the {\em width} of the $\epsilon$-region. More
specifically, for the scaled grid problem of Algorithm~\ref{alg-main}
to have a solution for a given denominator exponent $k$, it is a
necessary condition that the $\epsilon$-region intersects some grid
line. The distance between the grid lines scales as $\frac{1}{2^k}$,
and the width of the $\epsilon$-region is $\frac{\epsilon^2}{2}$, and
therefore, we expect the $\epsilon$-region to intersect at least one
grid line when $\frac{1}{2^k}\leq c\epsilon^2$, for some constant
$c$. Equivalently $k\geq K+2\log\frac{1}{\epsilon}$. By
Lemmas~\ref{lem-2k-2} and {\ref{lem-2k-1}}, the $T$-count scales as
$2k$, i.e., $K+4\log\frac{1}{\epsilon}$.

On the other hand, if $r\not\in\Q(\sqrt2)$, then the set of grid lines
is dense, and in this case, it appears that the solution to the grid
problem is dominated by the {\em area} of the $\epsilon$-region. In
this case, the number of grid points per area scales as $4^k$, and the
area of the $\epsilon$-region is proportional to $\epsilon^3$; hence
we expect a grid point to fall in the $\epsilon$-region when
$\frac{1}{4^k}\leq c\epsilon^3$, or equivalently,
$2k\geq K+3\log\frac{1}{\epsilon}$, yielding the typical $T$-count of
$K+3\log\frac{1}{\epsilon}$ in that case.

\subsection{Time complexity of the algorithm}

\begin{proposition}[Complexity]\label{prop-time-complexity}
  Algorithm~\ref{alg-main} runs in expected time
  $O(\polylog(1/\epsilon))$. This is true whether or not a
  factorization oracle is used.
\end{proposition}

\begin{proof}
  Let $M$ be the uprightness of the $\epsilon$-region. Let $j_0$ be
  the average number of candidates tried in steps 2(a)--(c) of the
  algorithm, and let $k_{j_0}$ be the least denominator exponent of
  the final candidate. Let $n$ be the largest integer that appears in
  step 2(a) of the algorithm.

  By Proposition~\ref{prop-scaled-increasing}, step 1 of the algorithm
  requires $O(\log(1/M))$ arithmetic operations, plus a constant
  number per candidate. For each of the $j_0$ candidates, step 2(a)
  requires $O(1)$ arithmetic operations. Step 2(b) also requires
  $O(1)$ arithmetic operations, either due to the use of a factoring
  oracle, or else, because we can put an arbitrary fixed bound on the
  amount of effort invested in factoring any given integer. At
  minimum, this will succeed when the integer in question is prime,
  which is sufficient for the estimates of
  Proposition~\ref{prop-near-optimality}. Step 2(c) requires
  $O(\polylog(n))$ operations by
  Theorem~\ref{thm-diophantine}. Finally, step 3 requires $O(k_{j_0})$
  arithmetic operations; see, e.g., Theorem~5.1 of
  {\cite{ma-remarks}}. So the total expected number of arithmetic
  operations is
  \begin{equation} \label{eqn-expected}
    O(\log(1/M)) + j_0\cdot O(\polylog(n)) + O(k_{j_0}).
  \end{equation}
  Recall that the $\epsilon$-region $\Repsilon$, shown in
  {\eqref{eqn-eps-region}}, contains a disk of radius $\epsilon^2/4$.
  Hence, $\area(\Repsilon)\geq \frac{\pi}{16}\epsilon^4$. On the
  other hand, the square $[-1,1]\times[-1,1]$ is a (not very tight)
  bounding box for $\Repsilon$. It follows that 
  \[ M = \up(\Repsilon) = \frac{\area(\Repsilon)}{\area(\BBox(\Repsilon))} =
  \Omega(\epsilon^4),
  \]
  hence $\log(1/M) = O(\log(1/\epsilon))$. From {\eqref{eqn-ej0-eps}},
  the expected value of $j_0$ is $O(\log({1}/{\epsilon}))$.
  Combining {\eqref{eqn-kj}} with {\eqref{eqn-k2}}, we therefore have
  \[ k_{j_0} \leq k_2+2(1+\log_2 j_0) = O(\log({1}/{\epsilon})) +
  O(\log(\log(1/\epsilon))) = O(\log(1/\epsilon)).
  \]
  From Lemma~\ref{lem-n2}, and the fact that candidates are enumerated
  in order of increasing denominator exponent, we have $n\leq
  4^{k_{j_0}}$, hence
  \[ \polylog(n) = O(\poly(k_{j_0})) = O(\polylog(1/{\epsilon})).
  \]
  Combining all of these estimates with {\eqref{eqn-expected}}, the
  expected number of arithmetic operations for the algorithm is
  $O(\polylog(1/{\epsilon}))$. Moreover, each individual arithmetic
  operation can be performed with precision $O(\log(1/\epsilon))$,
  taking time $O(\polylog(1/{\epsilon}))$. Therefore the total
  expected time complexity of the algorithm is
  $O(\polylog(1/{\epsilon}))$, as desired.
\end{proof}

\section{Approximation up to a phase}

So far, we have considered the problem of approximate synthesis ``on
the nose'', i.e., the operator $U$ in
Definition~\ref{def-synthesis-problem} was literally required to
approximate $\Rz(\theta)$ in the operator norm.  However, it is
well-known that global phases have no observable effect in quantum
mechanics, so in quantum computing, it is also common to consider the
problem of approximate synthesis ``up to a phase''.  This is made
precise in the following definition.

\begin{definition}\label{def-approx-up-to-phase}
  Given $\theta$ and some $\epsilon>0$, the {\em approximate synthesis
    problem for $z$-rotations up to a phase} is to find an operator
  $U$ expressible in the single-qubit Clifford+$T$ gate set, and a
  unit scalar $\lambda$, such that
  \begin{equation}\label{eqn-approx-up-to-phase}
    \norm{\Rz(\theta)-\lambda U} \leq \epsilon.
  \end{equation}
  Moreover, it is desirable to find $U$ of smallest possible
  $T$-count. As before, the norm in {\eqref{eqn-approx-up-to-phase}}
  is the operator norm.
\end{definition}

In this section, we will give a version of Algorithm~\ref{alg-main}
that optimally solves the approximate synthesis problem up to a phase.
The central insight is that it is in fact sufficient to restrict
$\lambda$ to only two possible phases, namely $\lambda=1$ and
$\lambda=\sqrt{\omega}=e^{i\pi/8}$. A similar technique was also used
in {\cite{KMM-practical}}.

First, note that if $W$ is a unitary $2\times 2$-matrix and $\det
W=1$, then $\tr W$ is real. This is obvious, because $\det W=1$
ensures that the two eigenvalues of $W$ are each other's complex
conjugates.

\begin{lemma}\label{lem-IW}
  Let $W$ be a unitary $2\times 2$-matrix, and assume that $\det W =
  1$ and $\tr W \geq 0$.  Then for all unit scalars $\lambda$, we have
  \[ \norm{I-W} \leq \norm{I-\lambda W}.
  \]
\end{lemma}

\begin{proof}
  We may assume without loss of generality that $W$ is diagonal. Since
  $\det W = 1$, we can write
  \[ W = \begin{pmatrix} e^{i\phi} & 0 \\ 0 & e^{-i\phi}
  \end{pmatrix}
  \]
  for some $\phi$. By symmetry, we can assume without loss of
  generality that $0\leq\phi\leq\pi$. Since $\tr W\geq 0$, we have
  $\phi\leq\pi/2$. Now consider a unit scalar $\lambda=e^{i\psi}$,
  where $-\pi\leq\psi\leq\pi$.  Then $\norm{I-\lambda W} =
  \max\s{|1-e^{i(\psi+\phi)}|, |1-e^{i(\psi-\phi)}|}$ and $\norm{I-W}
  = |1-e^{i\phi}|$. If $\psi\geq 0$, then
  $|1-e^{i\phi}|\leq|1-e^{i(\psi+\phi)}|$. Similarly, if $\psi\leq 0$,
  then $|1-e^{i\phi}|\leq|1-e^{i(\psi-\phi)}|$. In either case, we
  have $\norm{I-W} \leq \norm{I-\lambda W}$, as claimed.
\end{proof}

\begin{lemma}\label{lem-discrete-phase}
  Fix $\epsilon$, a unitary operator $R$ with $\det R=1$, and a
  Clifford+$T$ operator $U$. The following are equivalent:
  \begin{enumerate}\roundlabels
  \item There exists a unit scalar $\lambda$ such that
    \[ \norm{R-\lambda U}\leq\epsilon;
    \]
  \item There exists $n\in\Z$ such that
    \[ \norm{R-e^{in\pi/8} U}\leq\epsilon.
    \]
  \end{enumerate}
\end{lemma}

\begin{proof}
  It is obvious that (2) implies (1). For the opposite implication,
  first note that, because $U$ is a Clifford+$T$ operator, we have
  $\det U=\omega^k$ for some $k\in\Z$, and therefore $\det (R\inv U) =
  \omega^{k}$. Let $V=e^{-ik\pi/8}R\inv U$, so that $\det V=1$. If
  $\tr V\geq 0$, let $W=V$; otherwise, let $W=-V$. Either way, we have
  $W=e^{in\pi/8}R\inv U$, where $n\in\Z$, and $\det W=1$, $\tr W\geq
  0$. Let $\lambda' = e^{-in\pi/8}\lambda$. By Lemma~\ref{lem-IW}, we
  have
  \[ \begin{array}{crcl}
    & \norm{I-W} &\leq& \norm{I-\lambda' W} \\
    \imp &
    \norm{I-e^{in\pi/8}R\inv U}  &\leq& 
    \norm{I-\lambda' e^{in\pi/8}R\inv U} \\
    \imp &
    \norm{R-e^{in\pi/8} U}  &\leq& 
    \norm{R-\lambda' e^{in\pi/8} U},\\
    \imp &
    \norm{R-e^{in\pi/8} U}  &\leq& 
    \norm{R-\lambda U}, \\
  \end{array}
  \]
  which implies the desired claim.
\end{proof}

\begin{remark}
  A version of Lemma~\ref{lem-discrete-phase} also applies to gate
  sets other than Clifford+$T$, as long as the gate set has discrete
  determinants.
\end{remark}

\begin{corollary}\label{cor-only-two-phases}
  In Definition~\ref{def-approx-up-to-phase}, it suffices without loss
  of generality to consider only the two scalars $\lambda=1$ and
  $\lambda=e^{i\pi/8}$.
\end{corollary}

\begin{proof}
  Let $U$ be a Clifford+$T$ operator satisfying
  {\eqref{eqn-approx-up-to-phase}} for some unit scalar $\lambda$. By
  Lemma~\ref{lem-discrete-phase}, there exists a $\lambda$ of the form
  $e^{in\pi/8}$ also satisfying {\eqref{eqn-approx-up-to-phase}}.
  Then we can write $\lambda=\omega^k\lambda'$, where $k\in\Z$ and
  $\lambda'\in \s{1,e^{i\pi/8}}$. Letting $U'=\omega^k U$, we have
  $\lambda'U' = \lambda U$, and therefore
  \[ \norm{\Rz(\theta)-\lambda' U'} \leq \epsilon,
  \]
  as claimed. Moreover, since $\omega=e^{i\pi/4}$ is a Clifford
  operator, $U$ and $U'$ have the same $T$-count.
\end{proof}

To solve the approximate synthesis problem up to a phase, we therefore
need an algorithm for finding optimal solutions of
{\eqref{eqn-approx-up-to-phase}} in case $\lambda=1$ and
$\lambda=e^{i\pi/8}$ For $\lambda=1$, this is of course just
Algorithm~\ref{alg-main}. So all that remains to do is to find an
algorithm for solving
\begin{equation}\label{eqn-norm8}
  \norm{\Rz(\theta)-e^{i\pi/8}U}\leq \epsilon.
\end{equation}
We use a sequence of steps very similar to those of
Section~\ref{ssec-app-synth-problem} to reduce this to a grid problem
and a Diophantine equation.  We first consider the form of $U$.

\begin{lemma}\label{lem-ell8}
  If $\epsilon<|1-e^{i\pi/8}|$, then all solutions of
  {\eqref{eqn-norm8}} have the form
  \begin{equation}\label{eqn-u8}
    U = \zmatrix{cc}{u & -t\da\omega^{-1} \\ t & u\da\omega^{-1}}.
  \end{equation}
\end{lemma}

\begin{proof}
  This is completely analogous to the proof of Lemma~\ref{lem-ell}, using 
  $e^{i\pi/8}U$ in place of $U$.
\end{proof}

Recall that $\delta = 1+\omega$, and note that
$\frac{\delta}{|\delta|} = e^{i\pi/8}$. Also note that
$\delta\omega\inv = \delta\da$, and that the element $\delta$ is
invertible in $\D[\omega]$ with inverse $\delta\inv = (\omega-i)/\sqrt
2$. Suppose that $U$ is of the form {\eqref{eqn-u8}}. Let $u'=\delta
u$ and $t'=\delta t$. We have:
\begin{eqnarray}
  \norm{\Rz(\theta) - e^{i\pi/8}U}
  &=& 
  \bignorm{\Rz(\theta) - \frac{\delta}{|\delta|}\zmatrix{cc}{u &
      -t\da\omega^{-1} \\ t & u\da\omega^{-1}}}
  \nonumber\\ &=&
  \bignorm{\Rz(\theta) - \frac{1}{|\delta|}\zmatrix{cc}{\delta u &
      -\delta\da t\da \\ \delta t & \delta\da u\da}}
  \nonumber\\ &=&
  \bignorm{\Rz(\theta) - \frac{1}{|\delta|}\zmatrix{cc}{u' &
      -{t'}\da \\ t' & {u'}\da}}.
  \nonumber
\end{eqnarray}
Recall the definition of the $\epsilon$-region $\Repsilon$ from
{\eqref{eqn-eps-region}}.  Using exactly the same argument as in
Section~\ref{sec-algorithm}, it follows that {\eqref{eqn-norm8}} holds
if and only if $\frac{u'}{|\delta|}\in\Repsilon$, i.e., $u'\in
|\delta|\Repsilon$.

As before, in order for $U$ to be unitary, of course it must satisfy
$u\da u+t\da t=1$, and a necessary condition for this is
$u,u\bul\in\Disk$. The latter condition can be equivalently
re-expressed in terms of $u'$ by requiring $u'\in|\delta|\,\Disk$ and
${u'}\bul\in|\delta\bul\!|\,\Disk$. Therefore, finding solutions to
{\eqref{eqn-norm8}} of the form {\eqref{eqn-u8}} reduces to the
two-dimensional grid problem $u'\in |\delta|\Repsilon$ and
${u'}\bul\in|\delta\bul\!|\,\Disk$, together with the usual
Diophantine equation $u\da u+t\da t=1$. The last remaining piece of
the puzzle is to compute the $T$-count of $U$, and in particular, to
ensure that potential solutions are found in order of increasing
$T$-count.

\begin{lemma}\label{lem-2k-1}
  Let $U$ be a Clifford+$T$ operator of the form {\eqref{eqn-u8}}, and
  let $k$ be the least denominator exponent of $u'= \delta u$. Then
  the $T$-count of $U$ is either $2k-1$ or $2k+1$. Moreover, if $k>0$
  and $U$ has $T$-count $2k+1$, then $U'=TUT\da$ has $T$-count
  $2k-1$. 
\end{lemma}

\begin{proof}
  This can be proved by a tedious but easy induction, analogous to
  Lemma~\ref{lem-2k-2}.
\end{proof}

We therefore arrive at the following algorithm for solving
{\eqref{eqn-norm8}}. Here we assume $\epsilon < |1-e^{i\pi/8}|$, so
that Lemma~\ref{lem-ell8} applies.

\begin{algorithm}\label{alg-phase8}
  Given $\theta$ and $\epsilon$, let $A=|\delta|\Repsilon$, and let
  $B=|\delta\bul\!|\,\Disk$.
  \begin{enumerate}
  \item[1.] Use Proposition~\ref{prop-scaled-increasing} to enumerate
    the infinite sequence of solutions to the scaled grid problem
    $u'\in A$ and ${u'}\bul\in B$, where $u'\in\D[\omega]$, in the order
    of increasing least denominator exponent $k$. 
  \item[2.] For each such solution $u'$:
    \begin{enumerate}
    \item[(a)] Let $u=u'/\delta$, let $\xi=1 - u\da u\in\D[\sqrt2]$, and write
      $\xi\bul\xi=\frac{n}{2^\ell}$, where $n\in\Z$ and $\ell\geq 0$
      is minimal.
    \item[(b)] Attempt to find a prime factorization of
      $n$. If $n\neq 0$ but no prime factorization is found, skip step
      2(c) and continue with the next $u'$.
    \item[(c)] Use the algorithm of Theorem~\ref{thm-diophantine} to
      solve the equation $t\da t = \xi$. If a solution $t$ exists,
      go to step 3; otherwise, continue with the next $u'$.
    \end{enumerate}
  \item[3.] Define $U$ as in equation {\eqref{eqn-u8}}, let
      $U'=TUT\da$, and use the exact synthesis algorithm of
      {\cite{Kliuchnikov-etal}} to find a Clifford+$T$ circuit
      implementing either $U$ or $U'$, whichever has smaller
      $T$-count. Output this circuit and stop.
  \end{enumerate}
\end{algorithm}

Algorithm~\ref{alg-phase8} is optimal in the presence of a factoring
oracle, and near-optimal in the absence of a factoring oracle, in the
same sense as Algorithm~\ref{alg-main}. Its expected time complexity
is $O(\polylog(1/\epsilon))$. The proofs are completely analogous to
those of Section~\ref{sec-analysis}. We then arrive at the following
composite algorithm for the approximate synthesis problem for
$z$-rotations up to a phase:

\begin{algorithm}[Approximate synthesis up to a phase]
  \label{alg-phase}
  Given $\theta$ and $\epsilon$, run both Algorithms~\ref{alg-main}
  and {\ref{alg-phase8}}, and return whichever circuit has the smaller
  $T$-count. 
\end{algorithm}

\begin{proposition}[Correctness, time complexity, and optimality]
  Algorithm~\ref{alg-phase} yields a valid solution to the approximate
  synthesis problem up to a phase. Its expected running time is
  $O(\polylog(1/\epsilon))$. In the presence of a factoring oracle, the
  algorithm is optimal, i.e., the returned circuit has the smallest
  $T$-count of any single-qubit Clifford+$T$ circuit approximating
  $\Rz(\theta)$ up to $\epsilon$ and up to a phase.  Moreover, in the
  absence of a factoring oracle, the algorithm is near-optimal in the
  following sense. Let $m$ be the $T$-count of the solution
  found. Then:
  \begin{enumerate}\alphalabels
  \item The approximate synthesis problem has an expected number of at
    most $O(\log(1/\epsilon))$ non-equivalent solutions with $T$-count
    less than $m$.
  \item The expected value of $m$ is $m''' +
    O(\log(\log(1/\epsilon)))$, where $m'''$ is the $T$-count of the
    third-to-optimal solution (up to equivalence) of the approximate
    synthesis problem up to a phase.
  \end{enumerate}
\end{proposition}

\begin{proof}
  The correctness and time complexity of Algorithm~\ref{alg-phase}
  follows from that of Algorithms~\ref{alg-main} and
  {\ref{alg-phase8}}. The optimality results follow from those of
  Algorithms~\ref{alg-main} and {\ref{alg-phase8}}, keeping in mind
  that Algorithm~\ref{alg-main} finds an optimal (or near-optimal)
  solution for the phase $\lambda=1$, Algorithm~\ref{alg-phase8} finds
  an optimal (or near-optimal) solution for the phase
  $\lambda=e^{i\pi/8}$, and by Corollary~\ref{cor-only-two-phases},
  these are the only two phases that need to be considered.

  The only subtlety that must be pointed out is that in part (b) of
  the near-optimality, we use the $T$-count of the {\em
    third}-to-optimal solution, rather than the second-to-optimal one
  as in Proposition~\ref{prop-near-optimality}. This is because the
  optimal and second-to-optimal solutions may belong to
  Algorithms~\ref{alg-main} and {\ref{alg-phase8}}, respectively, so
  that it may not be until the third-to-optimal solution that the
  near-optimality result of either Algorithm~\ref{alg-main} or
  Algorithm~\ref{alg-phase8} can be invoked.
\end{proof}

\begin{remark}
  Algorithms~\ref{alg-main} and {\ref{alg-phase8}} share the same
  $\epsilon$-region up to scaling, and therefore the uprightness
  computation of Theorem~\ref{thm-ellipse} only needs to be done once. 
\end{remark}

\begin{remark}
  By Lemmas~\ref{lem-2k-2} and {\ref{lem-2k-1}},
  Algorithm~\ref{alg-main} always produces circuits with even
  $T$-count, and Algorithm~\ref{alg-phase8} always produces circuits
  with odd $T$-count. Instead of running both algorithms to
  completion, it is possible to interleave the two algorithms, so that
  all potential solutions are considered in order of increasing
  $T$-count. This is a slight optimization which does not, however,
  affect the asymptotic time complexity.
\end{remark}

\section{Experimental results}
\label{sec-experimental}

\begin{table}
\[
\begin{array}{|l|r|r|l||r|r|r|}
\hline
\multicolumn{1}{|c|}{\mbox{$\epsilon$}} &
\multicolumn{1}{|c|}{\mbox{$T$-count}} &
\multicolumn{1}{|c|}{\mbox{$T$-bound}} &
\multicolumn{1}{|c||}{\mbox{Actual error}} &
\multicolumn{1}{|c|}{\mbox{Runtime}} &
\multicolumn{1}{|c|}{\mbox{Candidates}} &
\multicolumn{1}{|c|}{\mbox{Time/Candidate}} \\\hline

10^{-10} &                    
102 &                         
\geq 102 &                    
0.91180\cdot 10^{-10} &       
0.0190s &                     
3.0 &                         
0.0064s \\                    

10^{-20} &                    
200 &                         
\geq 198 &                    
0.87670\cdot 10^{-20} &       
0.0433s &                     
7.0 &                         
0.0061s \\                    

10^{-30} &                    
298 &                         
\geq 298 &                    
0.99836\cdot 10^{-30} &       
0.0600s &                     
7.0 &                         
0.0085s \\                    

10^{-40} &                    
402 &                         
\geq 400 &                    
0.77378\cdot 10^{-40} &       
0.0976s &                     
11.7 &                        
0.0084s \\                    

10^{-50} &                    
500 &                         
\geq 500 &                    
0.82008\cdot 10^{-50} &       
0.1353s &                     
20.3 &                        
0.0067s \\                    

10^{-60} &                    
602 &                         
\geq 596 &                    
0.61151\cdot 10^{-60} &       
0.1548s &                     
16.0 &                        
0.0097s \\                    

10^{-70} &                    
702 &                         
\geq 698 &                    
0.40936\cdot 10^{-70} &       
0.1931s &                     
20.9 &                        
0.0093s \\                    

10^{-80} &                    
804 &                         
\geq 794 &                    
0.92372\cdot 10^{-80} &       
0.2402s &                     
27.2 &                        
0.0088s \\                    

10^{-90} &                    
898 &                         
\geq 898 &                    
0.96607\cdot 10^{-90} &       
0.2696s &                     
22.2 &                        
0.0121s \\                    

10^{-100} &                   
1000 &                        
\geq 998 &                    
0.78879\cdot 10^{-100} &      
0.3443s &                     
31.2 &                        
0.0110s \\                    

10^{-200} &                   
1998 &                        
\geq 1994 &                   
0.73266\cdot 10^{-200} &      
1.1423s &                     
62.3 &                        
0.0183s \\                    

10^{-500} &                   
4990 &                        
\geq 4986 &                   
0.67156\cdot 10^{-500} &      
8.6509s &                     
170.4 &                       
0.0508s \\                    

10^{-1000} &                  
9974 &                        
\geq 9966 &                   
0.80457\cdot 10^{-1000} &     
47.9300s &                    
270.4 &                       
0.1773s \\                    


10^{-2000} &                  
19942 &                       
\geq 19934 &                  
0.88272\cdot 10^{-2000} &     
383.1024s &                   
556.7 &                       
0.6881s \\                    


\hline
\end{array}
\]
\caption{Experimental results. The first four columns report the
  $T$-count, computed lower bound on the $T$-count, and the actual error 
  for approximating the operator $\Rz(\pi/128)$ up to various $\epsilon$. 
  The remaining columns report the runtime for each $\epsilon$,
  averaged over 50 independent runs of Algorithm~\ref{alg-main} with
  random angles $\theta$. The runtime is further broken down into
  average number of candidates tried per run of the algorithm, and
  time spent per candidate.}
\label{tab-results}
\rule{\textwidth}{0.1mm}
\end{table}

We implemented Algorithm~\ref{alg-main} in the programming language
Haskell; the implementation is available from {\cite{implementation}}.
Running on a single core of a 3.4GHz Intel i5-3570 CPU, we
approximated the operator $\Rz(\pi/128)$ up to various $\epsilon$. The
results are summarized in the leftmost four columns of
Table~\ref{tab-results}.  For example, here is the approximation of
$\Rz(\pi/128)$ up to $\epsilon=10^{-10}$:
\[ \begin{split}
  \makebox[1em][c]{$u$} =&~ \rtt{52}(-26687414\omega^3+10541729\omega^2+10614512\omega+40727366) \\
  \makebox[1em][c]{$t$} =&~ \rtt{52}(30805761\omega^3-23432014\omega^2+2332111\omega+20133911) \\
  \makebox[1em][c]{$U$} =
  &~ \tt HTSHTSHTSHTHTHTHTSHTHTSHTSHTSHTHTHTSHTSHTHTHTSHTHTSHTHTHTHTHTHTHTSHTSHT\\[-1ex]
  &~ \tt SHTHTSHTHTSHTHTHTHTSHTHTHTSHTHTSHTHTHTHTSHTSHTSHTHTHTSHTSHTSHTSHTHTSHTS\\[-1ex]
  &~ \tt HTSHTSHTHTSHTHTSHTSHTHTHTHTHTSHTHTHTHTSHTSHTSHTHTSHTSHTHTHTSHTHTHTHTHTS\\[-1ex]
  &~ \tt HTSHTHTHTHTHTSHTHTHTHTSHTHTHTHTHTHTH\omega^7
\end{split}
\]
In addition to $\epsilon$, the $T$-count of the computed operator, and
the actual error, we have also reported the lower bound on the
$T$-count, which was computed according to Remark~\ref{rem-lower-bound}.
It can be seen that the actual $T$-count achieved by the algorithm
exceeds this lower bound by at most a very small amount, which is
consistent with $O(\log(\log(1/\epsilon)))$ as predicted by
Proposition~\ref{prop-near-optimality}(b). This excess could be
further reduced by increasing the amount of effort spent on factoring,
and will become zero in the presence of a factoring oracle.

Table~\ref{tab-results} also shows the runtime as a function of
$\epsilon$. Since the enumeration of solutions to the grid problem in
the algorithm is deterministic, the algorithm tends to find the same
one or two solutions each time it is run with the same parameters. For
this reason, we have averaged the runtimes in Table~\ref{tab-results}
over 50 runs of the algorithm with random angles $\theta$, for each
$\epsilon$. As shown in Table~\ref{tab-results}, we were able to
achieve approximations up to $\epsilon=10^{-1000}$ with a $T$-count of
under 10000 in less than 50 seconds on average. This compares to a
$T$-count of 13300 and an average runtime of 504.8 seconds reported in
{\cite{Selinger-newsynth}} on the same hardware. We also note that the
experimental runtimes in Table~\ref{tab-results} are consistent with
the polynomial runtime predicted by
Proposition~\ref{prop-time-complexity}, and appear to be 
$O(\log^3(1/\epsilon))$.

\begin{table}
\[
\begin{array}{|l|ccccc|}
\hline
& \epsilon=10^{-100}
& \epsilon=10^{-200}
& \epsilon=10^{-300}
& \epsilon=10^{-400}
& \epsilon=10^{-500}
\\\hline
\theta=2\tan^{-1}(2+3\sqrt{2})
& 1320
& 2646
& 3976
& 5306
& 6636
\\
\theta=2\tan^{-1}({5}/{3})
& 1314
& 2646
& 3966
& 5306
& 6636
\\
\theta=2\tan^{-1}({(2+7\sqrt{2})}/{5})
& 1308
& 2642
& 3968
& 5294
& 6630
\\
\theta=2\tan^{-1}(2+3\sqrt{3})
& 1002
& 1996
& 2994
& 3992
& 4990
\\
\theta=2\tan^{-1}(\sqrt{5})
& 998
& 1998
& 2996
& 3990
& 4990
\\
\theta=2\tan^{-1}(1+11\sqrt{7})
& 1000
& 1998
& 2996
& 3992
& 4990
\\
\hline\hline
\multicolumn{1}{|c|}{4\log_2(1/\epsilon)}
& 1329
& 2658
& 3986
& 5315
& 6644
\\
\multicolumn{1}{|c|}{3\log_2(1/\epsilon)}
& 997
& 1993
& 2990
& 3986
& 4983
\\\hline
\end{array}

\]
\caption{Worst-case and typical $T$-counts. The top part of the
  table shows the $T$-counts for various combinations of $\theta$ and
  $\epsilon$. The bottom part of the table shows the actual values of
  $4\log_2(1/\epsilon)$ and $3\log_2(1/\epsilon)$ for each
  $\epsilon$, rounded to the closest integer. As predicted by
  Conjecture~\ref{con-theta}, the
  $T$-counts are close to $4\log_2(1/\epsilon)$ when
  $\theta=2\tan^{-1}(r)$, for $r\in\Q(\sqrt2)$, and close to
  $3\log_2(1/\epsilon)$ otherwise.}
\label{tab-2}
\rule{\textwidth}{0.1mm}
\end{table}

Table~\ref{tab-2} shows $T$-counts as a function of $\theta$ and
$\epsilon$. As predicted by Conjecture~\ref{con-theta}, the $T$-counts
are close to $4\log_2(1/\epsilon)$ when $\theta=2\tan^{-1}(r)$, for
$r\in\Q(\sqrt2)$, and close to $3\log_2(1/\epsilon)$ otherwise.

\section{Conclusion}

We have presented a fast new probabilistic algorithm for approximating
arbitrary single-qubit $z$-rotations by Clif\-ford+$T$ circuits. Our
algorithm is optimal in the presence of a factoring oracle, i.e., it
finds the shortest possible circuit whatsoever for the given problem
instance. In the absence of a factoring oracle, our algorithm is still
nearly optimal. In particular, no quantum computer is required to run
our algorithm.  The main technical innovation of this paper is a new
efficient algorithm for solving two-dimensional grid problems, such as
the ones that arise in candidate selection for approximate synthesis.
We solved this problem by an iterative method that successively
increases the uprightness of a pair of convex sets until the problem
is in a form where it can be solved directly.

It is an interesting question whether a similar algorithm can be found
for giving optimal or near-optimal approximations of arbitrary
single-qubit operators. In its current form, our method only applies
to $z$-rotations. Since any arbitrary single-qubit operator can be
decomposed into three $z$-rotations using Euler angles, our algorithm
can currently achieve a $T$-count of $9\log_2(1/\epsilon) +
O(\log(\log(1/\epsilon)))$. (Interestingly, the average case gate
complexity can always be achieved in this situation, because if the
operator to be approximated happens to exhibit worst-case behavior, we
can always change it by multiplying it by a small number of random
Clifford+$T$ gates). However, the information-theoretic lower bound in
this situation remains $K+3\log_2(1/\epsilon)$, so there is still
potential for improvement.

\section{Acknowledgements}

This research was supported by the Natural Sciences and Engineering
Research Council of Canada (NSERC). This research was supported by the
Intelligence Advanced Research Projects Activity (IARPA) via
Department of Interior National Business Center contract number
D12PC00527. The U.S.\ Government is authorized to reproduce and
distribute reprints for Governmental purposes notwithstanding any
copyright annotation thereon.  Disclaimer: The views and conclusions
contained herein are those of the authors and should not be
interpreted as necessarily representing the official policies or
endorsements, either expressed or implied, of IARPA, DoI/NBC, or the
U.S.\ Government.

\appendix

\section[Proof of Theorem 5.16]{Proof of Theorem~\ref{thm-ellipse}}
\label{appendix-skew-red}

This appendix contains a proof of Theorem~\ref{thm-ellipse}. We start
by reformulating the problem in more convenient terms. Recall that the
notion of uprightness introduced in Section~\ref{ssec-uprightness} was
defined for an arbitrary bounded convex subset of $\R^2$. If the set
in question is an ellipse, we can expand the definition of uprightness
into an explicit expression. Recall from Definition~\ref{def-ellipses}
that an ellipse centered at $p$ can be written as $E=\s{u\in\R^2\such
  (u-p)^\dagger D(u-p)\leq 1}$, where $D$ is a positive definite matrix
whose entries are, e.g., as follows:
\[
  D = 
  \left[ 
    \begin{array}{cc} 
      a & b  \\
      b & d
    \end{array} 
  \right].
\]
We can compute the area of $E$ and the area of its 
bounding box using $D$. Indeed, we have 
$\area(E)= \pi/\sqrt{\Det(D)}$ and 
$\area (\BBox(E))=4\sqrt{ad} /\Det(D)$. 
Substituting these in Definition~\ref{def-uprightness} 
yields the desired expression for uprightness:
\begin{equation}
\label{eqn:upright}
  \up(E) = \frac{\area(E)}
                {\area(\BBox(E))}
         = \frac{\pi}{4} 
             \sqrt{
               \frac{\Det(D)}{ad} 
             }.
    {\hspace{1.5cm}
    \begin{tikzpicture}[baseline=0, xscale=1.6, yscale=1]
      \draw[fill=yellow!20] (-1,-1) -- (1,-1) -- (1,1) -- (-1,1) -- cycle;
      \draw[fill=red!50, rotate=45, yscale=0.4] (0,0) circle (1.3131);
      \draw (0,0) node {\small $E$};
      \draw (-1,.75) node[right] {\small $\BBox(E)$};
    \end{tikzpicture}
    }
\end{equation}
It follows that the uprightness of $E$ is invariant under 
translation and scalar multiplication.

Recall that $\lambda=\sqrt{2} +1$. The matrix $D$ corresponding to 
an ellipse $E$ has determinant 1 if and only if it can be written 
in the form
\begin{equation}
  D =
  \left[ 
    \begin{array} {cc} 
      e{\lambda^{-z}} & b \\
      b & e{\lambda^{z}}
    \end{array} 
  \right]
\end{equation}
for some $b,e,z\in\R$ with $e>0$ and $e^2=b^2+1$.  In this case, the
definition of uprightness {\eqref{eqn:upright}} simplifies to
\begin{equation}\label{eqn-upe}
  \up(E) = \frac{\pi}{4e^2} = \frac{\pi}{4\sqrt{b^2+1}}.
\end{equation}
Equivalently, if $\up(E)=M$, then
\begin{equation}\label{eqn-b2}
  b^2 = \frac{\pi^2}{16M^2} - 1.
\end{equation}

Since Theorem~\ref{thm-ellipse} deals with pairs of ellipses, it is
convenient to introduce the following terminology for discussing pairs
of matrices.

\begin{definition}
A \emph{state} is a pair of real symmetric positive 
definite matrices of determinant 1. Given a state 
$\state$ with
\begin{equation}
\label{eqn:state}
  D =
  \left[ 
    \begin{array}{cc} 
      e{\lambda^{-z}} & b \\
      b & e{\lambda^{z}}
    \end{array} 
  \right]
  ~~~~
  \Delta =
  \left[ 
    \begin{array}{cc} 
      \varepsilon{\lambda^{-\zeta}} & \beta \\
      \beta & \varepsilon{\lambda^{\zeta}}
    \end{array} 
  \right]
\end{equation}
we define its \emph{skew} as 
$\sk\state=b^2+\beta^2$ and its
\emph{bias} as $\bias \state 
=\zeta-z$.
\end{definition}

Note that the skew of a state is small if and only if both $b^2$ and
$\beta^2$ are small, which happens, by {\eqref{eqn-upe}}, if and only
if the ellipses corresponding to $D$ and $\Delta$ both have large
uprightness. So our strategy for increasing the uprightness will be to
reduce the skew. In what follows, we use $\state$ to denote an
arbitrary state and always assume that the entries of $D$ and $\Delta$
are given as in \eqref{eqn:state}. For future reference, we record
here another useful property of states.

\begin{remark}\label{rem-be}
If $\state$ is a state with $b\geq 0$,
then $-be\leq -b^2$. Indeed:
\[
e^2=b^2+1 ~\imp~ e^2\geq b^2 ~\imp~ e\geq b ~\imp~ -be\leq -b^2.
\]
Similarly, if $b\leq 0$, then $be\leq -b^2$. Analogous 
inequalities also hold for $\beta$ and $\varepsilon$.
\end{remark}

The action of a grid operator on an ellipse can be adapted to states
in a natural way, provided that the operator is special.

\begin{definition}
\label{definition_action}
The action of special grid operators on states is defined as
follows. Here, $\G\da$ denotes the transpose of $\G$, and $\G\bul$ is
defined by applying $(-)\bul$ separately to each matrix entry, as in
Remark~\ref{gridoperatorcomposition}.
\[
\state\cdot \G = (\G^\dagger D\G,\G^{\bullet\dagger} \Delta \G^\bullet).
\]
\end{definition}

\begin{lemma}\label{lem-ellipse-action}
  Let $(D,\Delta)$ be a state, and let $A$ and $B$ be the ellipses
  centered at the origin that are defined by $D$ and $\Delta$,
  respectively.  Then the ellipses $\G(A)$ and $\G\bul(B)$ are defined
  by the matrices $D'$ and $\Delta'$, where
  \[ (D',\Delta') = (D,\Delta) \cdot \G\inv
  \]
\end{lemma}

\begin{proof}
  We have 
  \[ \begin{array}{rcl}
    \G(A) &=& \s{\G(u)\in\R^2 \mid u\da D u \leq 1} \\
    &=& \s{v\in\R^2 \mid (\G\inv v)\da D(\G\inv v) \leq 1} \\
    &=& \s{v\in\R^2 \mid v\da (\G\inv)\da D \G\inv v \leq 1},
  \end{array}
  \] 
  so the ellipse $\G(A)$ is defined by the positive operator
  $D'=(\G\inv)\da D \G\inv$. The proof for $\G\bul(B)$ is similar.
\end{proof}

The main ingredient in the proof of 
Theorem~\ref{thm-ellipse} is the following Step Lemma.

\begin{lemma}[Step Lemma]
\label{Step}
For any state $\state$, if $\sk\state \geq \formula{\P}$, 
then there exists a special grid operator $\G$ such that 
$\sk (\state \cdot \G)\leq \formula{\Qstep} ~\sk\state$. 
Moreover, $\G$ can be computed using a constant number of arithmetic
operations.
\end{lemma}

Before proving the Step Lemma, we show how it can be 
used to derive Theorem~\ref{thm-ellipse}, whose 
statement we reproduce here. 

\begin{un-theorem}
  Suppose $A,B\subseteq \R^2$ are ellipses. Then there exists a grid
  operator $\G$ such that $\G(A)$ and $\G^\bullet (B)$ are
  $1/\formula{\oneoverM}$-upright. Moreover, if $A$ and $B$ are
  $M$-upright, then $\G$ can be efficiently computed in $O(\log(1/M))$
  arithmetic operations.
\end{un-theorem}

\begin{proof}
  Let $D$ and $\Delta$ be the matrices defining $A$ and $B$
  respectively, in the sense of Definition~\ref{def-ellipses}. Since
  uprightness is invariant under translations and scaling, we may
  without loss of generality assume that both ellipses are centered at
  the origin, and that $\det D = \det \Delta = 1$.

  The pair $\state$ is a state. By applying Lemma~\ref{Step}
  repeatedly, we get grid operators $\G_1,\ldots,\G_n$ such that:
  \begin{equation}\label{eqn-skew-reduction}
    \sk(\state\cdot \G_1\ldots \G_n) \leq \formula{\P}.
  \end{equation}
  Now let $(D',\Delta')=\state\cdot \G_1\ldots \G_n$ and set
  $\G=(\G_1\cdots \G_n)^{-1}$. By Lemma~\ref{lem-ellipse-action}, the
  ellipses $\G(A)$ and $\G^{\bullet}(B)$ are defined by the matrices
  $D'$ and $\Delta'$, respectively.  Let $b$ and $\beta$ be the
  anti-diagonal entries of the matrices $D'$ and $\Delta'$,
  respectively.  We have:
  \[
  b^2+\beta^2  = \sk (D',\Delta')=
  \sk(\state\cdot \G\inv) =
  \sk (\state\cdot \G_1\ldots \G_n)\leq \formula{\P},
  \]
  hence $b^2\leq \formula{\P}$ and $\beta^2\leq\formula{\P}$.
  Using {\eqref{eqn-upe}}, we get
  \[
  \up(\G(A)) 
  = \frac{\pi}{4\sqrt{b^2+1}} 
  \geq \frac{\pi}{4\sqrt{\formula{\P+1}}}
  \geq 1/\formula{\oneoverM} 
  \quad\mbox{and}\quad
  \up(\G\bul(B)) 
  = \frac{\pi}{4\sqrt{\beta^2+1}} 
  \geq \frac{\pi}{4\sqrt{\formula{\P+1}}} 
  \geq 1/\formula{\oneoverM},
  \]
  as desired. 
  \assert{\P <= (pi^2 / (16 * \M^2))-1}%

  To bound the number of operations, note that each application of
  $\G_j$ reduces the skew by at least $\formula{\roundone{100-100*\Qstep}}$ percent.
  Therefore, the number $n$ in {\eqref{eqn-skew-reduction}} satisfies
  $n\leq \log_{\formula{\Qstep}}({\formula{\P}}/{\sk\state}) =
  O(\log(\sk\state))$.  Using {\eqref{eqn-b2}}, we have
  \[ \log (\sk\state) = \log (b^2+\beta^2)
  \leq \log ((\frac{\pi^2}{16M^2} - 1) + (\frac{\pi^2}{16M^2} - 1))
  = O(\log(1/M)).
  \]
  It follows that the computation of $\G$ requires
  $O(\log(1/M))$ applications of the Step Lemma, each of
  which requires a constant number of arithmetic operations, proving
  the final claim of the theorem. 
\end{proof}

The remainder of this appendix is devoted to proving the 
Step Lemma. To each state, we associate the pair $(z, 
\zeta)$. The proof of the Step Lemma is essentially 
a case distinction on the location of the pair $(z, 
\zeta)$ in the plane. We find 
coverings of the plane with the property that if the 
point $(z, \zeta)$ belongs to some region 
$\mathcal{O}$ of our covering, then it is easy to compute a 
special grid operator $\G$ such that $\sk(\state\cdot \G)\leq 
\formula{\Qstep}~\sk\state$. The relevant grid operators are 
given in Figure~\ref{list_operators}.
\begin{figure}
\[
R= \frac{1}{\sqrt{2}}\left[ \begin{array} {cc} 
1 & -1  \\
1 & 1
\end{array} \right]
~~
A = \left[ \begin{array} {cc} 
1 & -2  \\
0 & 1
\end{array} \right]
~~ 
B = \left[ \begin{array} {cc} 
1 & \sqrt{2}  \\
0 & 1
\end{array} \right]
\]
\[
K = \frac{1}{\sqrt{2}}\left[ \begin{array} {cc} 
-\lambda^{-1} & -1  \\
\lambda & 1
\end{array} \right]
~~
X= \left[ \begin{array} {cc} 
0 & 1  \\
1 & 0
\end{array} \right]
~~ 
Z= \left[ \begin{array} {cc} 
1 & 0  \\
0 & -1
\end{array} \right]
\]
\caption{List of useful grid operators.}
\label{list_operators}
\rule{\textwidth}{0.1mm}
\end{figure}
Each one of the next 5 subsections is dedicated to a
particular region of the plane. We prove the Step Lemma 
in Section~\ref{ssect-step}.

\subsection{The Shift Lemma}
\label{ssect-shift}

In this section, we consider states $\state$ such that 
$|\bias\state|>1$. Any such state can be ``shifted" to a 
state $(D', \Delta')$ of equal skew but with 
$|\bias(D', \Delta')|\leq 1$.

\begin{definition}
The \emph{shift operators} $\sigma$ and $\tau$ are defined by:
\[
\sigma =
\sqrt{\lambda^{-1}}\left[\begin{array}{cc}
\lambda & 0 \\
0 & 1
\end{array} \right],
\tau =
\sqrt{\lambda^{-1}}
\left[\begin{array}{cc}
1 & 0 \\
0 & -\lambda
\end{array} \right]
\]
\end{definition}

Even though $\sigma$ and $\tau$ are not grid operators, we can use 
them to define an operation on states called a \emph{shift by $k$}. 
By abuse of notation, we write this operation as an action.

\begin{definition}
Given a state $\state$ and $k\in\Z$, the \emph{k-shift of $\state$} is defined as:
\[
\state \cdot \shift^k = (\sigma^k D \sigma^k, \tau^k \Delta \tau^k).
\]
\end{definition}

The notation $\state \cdot \shift^k$ is justified by the following lemma.

\begin{lemma}
\label{propshift}
The shift of a state is a state and moreover:
\[
\sk(\state\cdot\shift^k)= \sk\state
~\mbox{ and }~
\bias(\state\cdot\shift^k)=\bias\state +2k
\]
\end{lemma}

\begin{proof}
Compute $\state\cdot\shift^k$:
\[
\begin{array}{rcl}
\state\cdot\shift^k & = & (\sigma^k D \sigma^k, \tau^k \Delta \tau^k) \\
& = & (\sigma^k\left[ \begin{array} {cc} 
e{\lambda^{-z}} & b \\
b & e{\lambda^{z}}
\end{array} \right] \sigma^k,
\tau^k\left[ \begin{array} {cc} 
\varepsilon{\lambda^{-\zeta}} & \beta \\
\beta & \varepsilon{\lambda^{\zeta}}
\end{array} \right]\tau^k) \\
& = & (\left[ \begin{array} {cc} 
e\lambda^{-z+k} & b \\
b & e\lambda^{z-k}
\end{array} \right],
\left[ \begin{array} {cc} 
\varepsilon\lambda^{-\zeta-k} & (-1)^k\beta \\
(-1)^k\beta & \varepsilon\lambda^{\zeta+k}
\end{array} \right])
\end{array}
\]
The resulting matrices are clearly symmetric and positive definite. Moreover, since $\sigma^k$ and $\tau^k$ have determinant $\pm 1$, both $\sigma^k D\sigma^k$ and $\tau^k \Delta\tau^k$ have determinant 1. Finally:
\begin{itemize}
  \item $\sk(\state\cdot\shift^k)= b^2+((-1)^k\beta)^2=b^2+\beta^2= \sk\state$ and
  \item $\bias(\state\cdot\shift^k)=(\zeta+k)-(z-k)=\bias\state +2k$,
\end{itemize}
which completes the proof.
\end{proof}

For every special grid operator $\G$, there is a special grid operator $\G'$ whose action on a state corresponds to shifting the state by $k$, applying $\G$ and then shifting the state by $-k$.

\begin{lemma}
\label{conjugationbysigma1}
If $\G$ is a special grid operator and $k\in\Z$, then $\G'=\sigma^k \G \sigma^k$ is a special grid operator and moreover $\G'^\bullet = (-\tau)^k \G^\bullet \tau^k$.
\end{lemma}

\begin{proof}
It suffices to show this for $k=1$. Suppose $\G=\left[ \begin{array} {cc} 
w & x \\
y & z
\end{array} \right]$ is a special grid operator and note that:
\[
\G'=
\sigma \G \sigma=
\left[ \begin{array} {cc} 
\lambda w & x \\
y & \lambda^{-1}z
\end{array} \right]
=
\left[ \begin{array} {cc} 
\lambda^{-1} & 0 \\
0 & \lambda^{-1}
\end{array} \right]
\left[ \begin{array} {cc} 
\lambda & 0 \\
0 & 1
\end{array} \right]
\G
\left[ \begin{array} {cc} 
\lambda  & 0 \\
0 & 1
\end{array} \right].
\]
Since all the factors in the above product are grid operators, the result is also a grid operator. Moreover $\Det(\sigma\G\sigma)=\Det(\G)=1$ so that $\sigma\G\sigma$ is special. Finally:
\[
\G'^\bullet=(\sigma \G \sigma)^\bullet=
\left[ \begin{array} {cc} 
\lambda^\bullet w^\bullet & x^\bullet \\
y^\bullet & (\lambda^{-1})^\bullet z^\bullet
\end{array} \right]
=
\left[ \begin{array} {cc} 
-\lambda^{-1} w^\bullet & x^\bullet \\
y^\bullet & -\lambda z^\bullet
\end{array} \right]
=
-\tau \G^\bullet \tau.
\]
\end{proof}

\begin{lemma}
\label{conjugationbysigma2}
If $\G$ is a grid operator, then:
\[
((\state\cdot \shift^k) \cdot \G )\cdot \shift^k = \state \cdot (\sigma^k \G\sigma^k).
\]
\end{lemma}

\begin{proof} Write $\G'=\sigma^k \G\sigma^k$. Simple computation then yields the result:
\[
\begin{array}{rcl}
((\state\cdot \shift^k) \cdot \G )\cdot \shift^k 
& = & ((\sigma^k D \sigma^k,\, \tau^k \Delta \tau^k)\cdot \G )\cdot \shift^k \\
& = & (\G^\dagger\sigma^k D \sigma^k \G,\, \G^{\bullet\dagger} \tau^k \Delta \tau^k \G^\bullet )\cdot \shift^k \\
& = & (\sigma^k\G^\dagger\sigma^k D \sigma^k \G\sigma^k,\, \tau^k \G^{\bullet\dagger} \tau^k \Delta \tau^k \G^\bullet \tau^k ) \\
& = & (\sigma^k\G^\dagger\sigma^k D \sigma^k \G\sigma^k,\, ((-\tau)^k\G^{\bullet\dagger} \tau^k )\Delta ((-\tau)^k \G^\bullet \tau^k) ) \\
& = & (\G'^\dagger D \G',\, \G'^{\bullet\dagger} \Delta {\G'}^\bullet ) \\
& = & \state \cdot \G' \\
& = & \state \cdot (\sigma^k \G\sigma^k).
\end{array}
\]
\end{proof}

Shifts allow us to consider only states $\state$ with 
$\bias\state\in[-1,1]$ in the proof of the Step Lemma.

\begin{lemma}
\label{shift_lemma}
If the Step Lemma holds for all states $\state$ with
$\bias\state\in [-1,1]$, then it holds for all states.
\end{lemma}

\begin{proof}
Let $\state$ be some state with $\sk\state\geq\formula{\P}$. 
Let $x=\bias\state$ and set $k=\floor{\frac{1-x}{2}}$. Then by 
Lemma~\ref{propshift}, we have $\sk(\state\cdot\shift^k)= \sk\state$ and
$\bias(\state\cdot\shift^k)\in[-1,1]$. Then by assumption, 
there exists a special grid operator $\G$ such that 
$\sk((\state\cdot\shift^k)\cdot \G)\leq \formula{\Qstep}~\sk(\state 
\cdot\shift^k)$. Now by Lemma~\ref{conjugationbysigma1} 
we know that $\G'=\sigma^k ~\G~ \sigma^k$ is a special 
grid operator. Moreover, by 
Lemma~\ref{conjugationbysigma2} and \ref{propshift}, 
we have:
\[
\begin{array}{rcl}
\sk(\state\cdot \G') & = & \sk(((\state\cdot\shift^k)\cdot \G)\cdot \shift^k) \\
& = & \sk((\state \cdot\shift^k)\cdot \G) \\
& \leq & \formula{\Qstep}~\sk(\state \cdot\shift^k) \\
& = & \formula{\Qstep}~\sk\state ,
\end{array}
\]
which completes the proof.
\end{proof}

\subsection[The R Lemma]{The $R$ Lemma}
\label{ssect-R}

\begin{definition}
  The \emph{hyperbolic sine in base $\lambda$} is defined as:
\[
  \sinl(x) = \frac{\lambda^x-\lambda^{-x}}{2}.
\]
\end{definition}

\begin{lemma}
\label{lemmaR}
Recall the operator $R$ from Figure~\ref{list_operators}. 
If $\state$ is such that $\sk\state\geq\formula{\P}$, and
such that $\formula{-\r}\leq z\leq\formula{\r}$ and $\formula{-\r}\leq
\zeta\leq\formula{\r}$, then:
\[
  \sk(\state\cdot R) \leq \formula{\Qstep}~\sk\state. 
\]
\end{lemma}

\begin{proof}
Compute the action of $R$ on $\state$:
\[
  \begin{array}{lcl}
    R^\dagger DR 
    & = & \displaystyle\frac{1}{2}
    \left[ \begin{array} {cc} 
      1 & 1  \\
      -1 & 1
    \end{array} \right]
    \left[ \begin{array} {cc} 
      e{\lambda^{-z}} & b \\
      b & e{\lambda^{z}}
    \end{array} \right]
    \left[ \begin{array} {cc} 
      1 & -1  \\
      1 & 1
    \end{array} \right] \\\\[-1.5ex]
    & = &
    \left[ \begin{array} {cc} 
      \ldots & \frac{e({\lambda^{z}}-{\lambda^{-z}})}{2}\\
      \frac{e({\lambda^{z}}-{\lambda^{-z}})}{2} & \ldots
    \end{array} \right] 
    ~ = ~ 
    \left[ \begin{array} {cc} 
      \ldots & e\sinl(z)\\
      e\sinl(z) & \ldots
    \end{array} \right], 
  \\\\[-1.5ex]
    R^{\bullet\dagger} \Delta R^\bullet 
    & = & \displaystyle\frac{1}{2}
    \left[ \begin{array} {cc} 
      -1 & -1  \\
      1 & -1
    \end{array} \right]
    \left[ \begin{array} {cc} 
      \varepsilon{\lambda^{-\zeta}} & \beta \\
      \beta & \varepsilon{\lambda^{\zeta}}
    \end{array} \right]
    \left[ \begin{array} {cc} 
      -1 & 1  \\
      -1 & -1
    \end{array} \right] \\\\[-1.5ex]
    & = & \left[ \begin{array} {cc} 
      \ldots & \frac{\varepsilon({\lambda^{\zeta}}-{\lambda^{-\zeta}})}{2}\\
      \frac{\varepsilon({\lambda^{\zeta}}-{\lambda^{-\zeta}})}{2} & \ldots
    \end{array} \right] 
    ~ = ~ \left[ \begin{array} {cc} 
      \ldots & \varepsilon\sinl(\zeta)\\
      \varepsilon\sinl(\zeta) & \ldots
    \end{array} \right].
  \end{array}
\]
Therefore $\sk(\state\cdot R)= e^2\sinl^2(z)+ 
\varepsilon^2\sinl^2(\zeta)$. But recall that 
$e^2=b^2+1$ and $\varepsilon^2=\beta^2+1$, so that in 
fact:
\[
  \sk(\state\cdot R) = (b^2+1)\sinl^2(z)+
  (\beta^2+1)\sinl^2(\zeta). 
\]
We assumed 
$\formula{-\r}\leq z,\zeta\leq\formula{\r}$
and this implies that 
$\sinl^2(\zeta), \sinl^2(z) \leq 
\sinl^2(\formula{\r})$. Writing 
$y=\sinl^2(\formula{\r})$ for brevity, and using 
the assumption that $\sk\state\geq\formula{\P}$, 
we get:
\[
  \begin{array}{rcl}
  \sk(\state\cdot R) & = &
  (b^2+1)\sinl^2(z)+
  (\beta^2+1)\sinl^2(\zeta)\\
  & \leq & 
  (b^2+1)y+(\beta^2+1)y\\
  & = & (b^2+\beta^2+2)y \\
  & \leq & \sk\state (1+\frac{2}{\formula{\P}})y.
  \end{array}
\]
This completes the proof, since $(1+\frac{2}
{\formula{\P}})y = (1+\frac{2}
{\formula{\P}})\sinl^2(\formula{\r}) 
\approx \formula{(1+(2/\P))*(\sinhl \r)^2} 
\leq \formula{\Qstep}$.
\assert{(1+(2/\P))*(\sinhl \r)^2 <= \Qstep}%
\end{proof}

\subsection[The K Lemma]{The $K$ Lemma}
\label{ssect-K}

\begin{definition}
The \emph{hyperbolic cosine in base $\lambda$} is 
defined as:
\[
  \cosl (x) = \frac{\lambda^x+\lambda^{-x}}{2}.
\]
\end{definition}

\begin{lemma}
\label{lemmaK}
Recall the operator $K$ from Figure~\ref{list_operators}. 
If $\state$ is such that $\bias\state\in 
[-1,1]$, $\sk\state\geq\formula{\P}$, and such that $b,\beta\geq 0$,
$z\leq\formula{-\a}$, and $\formula{\r}\leq\zeta$, then:
\[
  \sk(\state\cdot K)\leq \formula{\Qstep}~\sk\state.
\]
\end{lemma}

\begin{proof}
Compute the action of $K$ on $\state$:
\begin{eqnarray}
\lefteqn{K^\dagger DK}\nonumber\\
& = & \frac{1}{2}
  \left[ \begin{array}{cc} 
  -\lambda^{-1} & \lambda  \\
  -1 & 1
  \end{array} \right]
  \left[ \begin{array}{cc} 
  e{\lambda^{-z}} & b \\
  b & e{\lambda^{z}}
  \end{array} \right]
  \left[ \begin{array}{cc} 
  -\lambda^{-1} & -1  \\
  \lambda & 1
  \end{array} \right] \nonumber \\[9pt]
& = & \frac{1}{2} 
  \left[ \begin{array}{cc} 
  \ldots & e(\lambda^{z+1}+\lambda^{-z-1})
    -2\sqrt{2}b \\
  e(\lambda^{z+1}+\lambda^{-z-1})-2\sqrt{2}b & \ldots
  \end{array} \right]  \nonumber \\[9pt]
& = & 
  \left[ \begin{array}{cc} 
  \ldots & e\cosl(z+1) -\sqrt{2}b \\
  e\cosl(z+1)  -\sqrt{2}b & \ldots
  \end{array} \right],\nonumber
\end{eqnarray}
\begin{eqnarray} 
\lefteqn{K^{\bullet\dagger} \Delta K^\bullet} 
\nonumber\\
& = & \frac{1}{2}
  \left[ \begin{array}{cc} 
  \lambda & -\lambda^{-1}  \\
  -1 & 1
  \end{array} \right]
  \left[ \begin{array}{cc} 
  \varepsilon{\lambda^{-\zeta}} & \beta \\
  \beta & \varepsilon{\lambda^{\zeta}}
  \end{array} \right]
  \left[ \begin{array}{cc} 
  \lambda & -1  \\
  -\lambda^{-1} & 1
  \end{array} \right] \nonumber \\[9pt]
& = & \frac{1}{2} 
  \left[ \begin{array}{cc} 
  \ldots & -\varepsilon(
    \lambda^{\zeta-1}+\lambda^{-\zeta+1})+2\sqrt{2}\beta \\
  -\varepsilon(\lambda^{\zeta-1}+\lambda^{-\zeta+1})
    +2\sqrt{2}\beta & \ldots
  \end{array} \right]  \nonumber \\[9pt]
& = &
  \left[ \begin{array}{cc} 
  \ldots & \sqrt{2}\beta - \varepsilon \cosl(\zeta-1) \\
  \sqrt{2}\beta - \varepsilon\cosl(\zeta-1)
    & \ldots
  \end{array} \right].\nonumber
\end{eqnarray}

\noindent
Therefore:
\begin{equation}\label{k0}
\sk(\state \cdot K)=(\sqrt{2}b -e\cosl(z+1))^2 
+(\sqrt{2}\beta-\varepsilon\cosl(\zeta-1))^2.
\end{equation}

\noindent
But recall that $e^2=b^2+1$, and from Remark~\ref{rem-be} that $b\geq 0$ 
implies $-be\leq -b^2$, so:
\begin{eqnarray}
\lefteqn{(\sqrt{2}b -e\cosl(z+1))^2} 
\nonumber\\
& = & 2b^2-2\sqrt{2}\,be\cosl(z+1)+ 
e^2\cosl^2(z+1) \nonumber\\
& \leq & 2b^2-2\sqrt{2}\,b^2\cosl(z+1)+ 
(b^2+1)\cosl^2(z+1) \nonumber\\
& = & b^2(2-2\sqrt{2}\cosl(z+1)+ 
\cosl^2(z+1))+\cosl^2(z+1) 
\nonumber\\
& = & b^2(\sqrt{2}-\cosl(z+1))^2+
\cosl^2(z+1).\label{k1}
\end{eqnarray}
Reasoning analogously, we also have
\begin{eqnarray}
  (\sqrt{2}\beta-\varepsilon\cosl(\zeta-1))^2
  & \leq &
  \beta^2(\sqrt{2}-\cosl(\zeta-1))^2+\cosl^2(\zeta-1).
\end{eqnarray}

\noindent
By assumption, $\bias\state\in [-1,1]$, thus 
$\zeta\leq z+1$.
This, together 
with the assumptions $\formula{\r}\leq\zeta$ 
and $z\leq\formula{-\a}$,
implies that both $z+1$ and $\zeta-1$ are in the interval
$[\formula{-\rr},\formula{-\aa}]$. On this interval, the function 
$\cosl^2(x)$ assumes its maximum at $x=\pgfmathparse{
  ifthenelse((\coshl {-\rr}) <= (\coshl {-\aa}), -\aa,
  -\rr)}\pgfmathresult$, and the function 
$f(x) = (\sqrt 2 - \cosl(x))^2$ assumes its maximum at 
$x=\assert{-\rr <= 0 && 0 <= -\aa}%
\pgfmathparse{
  ifthenelse((\kone{-\rr}) <= (\kone{0}) && (\kone{-\aa}) <=
  (\kone{0}), 0,  
  ifthenelse ((\kone {-\rr}) <= (\kone {-\aa}), -\aa, -\rr)}\pgfmathresult$.
Therefore, 
\begin{eqnarray}
b^2(\sqrt{2}-\cosl(z+1))^2  + \cosl^2(z+1)
& \leq &
\pgfmathparse{
  ifthenelse((\kone{-\rr}) <= (\kone{0}) && (\kone{-\aa}) <= (\kone{0}), 0,
  ifthenelse ((\kone {-\rr}) <= (\kone {-\aa}), -\aa, -\rr)
}
b^2(\sqrt{2} -\cosl(\pgfmathresult))^2 
+ 
\pgfmathparse{
  ifthenelse((\coshl {-\rr}) <= (\coshl {-\aa}), -\aa, -\rr)
}
\cosl^2(\pgfmathresult).
\end{eqnarray}
and
\begin{eqnarray}
\beta^2(\sqrt{2}-\cosl(\zeta-1))^2+\cosl^2(\zeta-1)
& \leq & 
\beta^2\pgfmathparse{
  ifthenelse((\kone{-\rr}) <= (\kone{0}) && (\kone{-\aa}) <= (\kone{0}), 0,
  ifthenelse ((\kone {-\rr}) <= (\kone {-\aa}), -\aa, -\rr)
}
(\sqrt{2} -\cosl(\pgfmathresult))^2
+
\pgfmathparse{
  ifthenelse((\coshl {-\rr}) <= (\coshl {-\aa}), -\aa, -\rr)
}
\cosl^2(\pgfmathresult).\label{k3}
\end{eqnarray}

\noindent
Combining (\ref{k0})--(\ref{k3}), together 
with the assumption that $\sk\state\geq\formula{\P}$, yields:
\begin{eqnarray}
\sk(\state \cdot K) & = & 
(\sqrt{2}b -e\cosl(z+1))^2 
+(\sqrt{2}\beta-\varepsilon\cosl(\zeta-1))^2 \nonumber\\
& \leq & 
(b^2+\beta^2)\pgfmathparse{
  ifthenelse((\kone{-\rr}) <= (\kone{0}) && (\kone{-\aa}) <= (\kone{0}), 0,
  ifthenelse ((\kone {-\rr}) <= (\kone {-\aa}), -\aa, -\rr)
}
(\sqrt{2} -\cosl(\pgfmathresult))^2
+
2\pgfmathparse{
  ifthenelse((\coshl {-\rr}) <= (\coshl {-\aa}), -\aa, -\rr)
}
\cosl^2(\pgfmathresult) \nonumber\\
& = & 
\sk\state\pgfmathparse{
  ifthenelse((\kone{-\rr}) <= (\kone{0}) && (\kone{-\aa}) <= (\kone{0}), 0,
  ifthenelse ((\kone {-\rr}) <= (\kone {-\aa}), -\aa, -\rr)
}
(\sqrt{2} -\cosl(\pgfmathresult))^2
+
2\pgfmathparse{
  ifthenelse((\coshl {-\rr}) <= (\coshl {-\aa}), -\aa, -\rr)
}
\cosl^2(\pgfmathresult) \nonumber\\
& \leq & 
\sk\state(\pgfmathparse{
  ifthenelse((\kone{-\rr}) <= (\kone{0}) && (\kone{-\aa}) <= (\kone{0}), 0,
  ifthenelse ((\kone {-\rr}) <= (\kone {-\aa}), -\aa, -\rr)
}
(\sqrt{2} -\cosl(\pgfmathresult))^2
+
\frac{2}{\formula{\P}}\pgfmathparse{
  ifthenelse((\coshl {-\rr}) <= (\coshl {-\aa}), -\aa, -\rr)
}
\cosl^2(\pgfmathresult)) \nonumber
\end{eqnarray}
This completes the proof since 
$\pgfmathparse{
  ifthenelse((\kone{-\rr}) <= (\kone{0}) && (\kone{-\aa}) <= (\kone{0}), 0,
  ifthenelse ((\kone {-\rr}) <= (\kone {-\aa}), -\aa, -\rr)
}
(\sqrt{2} -\cosl(\pgfmathresult))^2
+
\frac{2}{\formula{\P}}\pgfmathparse{
  ifthenelse((\coshl {-\rr}) <= (\coshl {-\aa}), -\aa, -\rr)
}
\cosl^2(\pgfmathresult)
\approx\formula{max((\kone {-\rr}),(\kone {-\aa}),(\kone{0})) +
(max((\coshl {-\rr})^2, (\coshl {-\aa})^2))*(2/\P)}
\leq \formula{\Qstep}$.
\assert{max((\kone {-\rr}),(\kone {-\aa}),(\kone{0})) + 
(max((\coshl {-\rr})^2, (\coshl {-\aa})^2))*(2/\P) 
<= \Qstep}%
\end{proof}

\subsection[The A Lemma]{The $A$ Lemma}
\label{ssect-A}

\begin{definition}
  Let $g(x) = (1-2x)^2$. 
\end{definition}

\begin{lemma}
\label{lemmaA}
Recall the operator $A$ from Figure~\ref{list_operators}. 
If $\state $ is such that $\bias\state
\in[-1,1]$, $\sk\state\geq\formula{\P}$, and such that
$b,\beta\geq 0$ and $\formula{-\a}\leq{z},\zeta$,
then there exists $n\in \Z$ 
such that:
\[
  \sk(\state \cdot A^n) 
  \leq \formula{\Qstep}~\sk\state . 
\]
\end{lemma}

\begin{proof}
Let $c=\min\s{{z},{\zeta}}$ and 
$n=\max\s{1, \floor{\frac{\lambda^{c}}{2}}}$.
Compute the action of $A^{n}$ on $\state$:
\[
  \begin{array}{lcl}
    {A^{n}}^\dagger DA^{n} & = & 
    \left[\begin{array} {cc} 
      1 & 0  \\
      -2n & 1
    \end{array}\right]
    \left[ \begin{array} {cc} 
      e{\lambda^{-z}} & b \\
      b & e{\lambda^{z}}
    \end{array}\right]
    \left[ \begin{array} {cc} 
      1 & -2n  \\
      0 & 1
    \end{array}\right] \\[9pt]
    & = & \left[\begin{array} {cc} 
      \ldots & b-2ne{\lambda^{-z}}  \\
      b-2ne{\lambda^{-z}} & \ldots
    \end{array} \right], \\[18pt]
    {A^{n}}^{\bullet\dagger}
      \Delta{A^{n}}^\bullet 
    & = & {A^{n}}^\dagger \Delta A^{n} \\
    & = & \left[ \begin{array} {cc} 
      \ldots & \beta-2n\varepsilon{\lambda^{-\zeta}}  \\
      \beta-2n\varepsilon{\lambda^{-\zeta}} & \ldots
    \end{array} \right].
  \end{array}
\]
Therefore:
\[
\sk(\state \cdot A^{n}) = 
(b-2ne{\lambda^{-z}})^2+(\beta-2n\varepsilon{\lambda^{-\zeta}})^2
\]
But recall that $e^2=b^2+1$ and
$\varepsilon^2=\beta^2+1$, and from Remark~\ref{rem-be} that 
$b,\beta\geq 0$ implies $-be\leq -b^2$ and 
$-\varepsilon\beta\leq -\beta^2$. Using these 
facts, we can expand the above formula as 
follows:
\begin{eqnarray}
\lefteqn{\sk(\state\cdot A^{n})}\nonumber \\
& = & (b-2ne{\lambda^{-z}})^2+(\beta-2n\varepsilon{\lambda^{-\zeta}})^2 
  \nonumber\\
& = & b^2
-4nbe{\lambda^{-z}}
+4n^2e^2{\lambda^{-2z}}
+\beta^2
-4n\beta\varepsilon{\lambda^{-\zeta}}
+4n^2\varepsilon^2{\lambda^{-2\zeta}}
  \nonumber\\
& \leq & b^2
-4nb^2{\lambda^{-z}}
+4n^2(b^2+1){\lambda^{-2z}}
+\beta^2
-4n\beta^2{\lambda^{-\zeta}}
+4n^2(\beta^2+1){\lambda^{-2\zeta}}
  \nonumber\\
& = & 
b^2(
1
-4n{\lambda^{-z}}
+4n^2{\lambda^{-2z}}
)+\beta^2(
1
-4n{\lambda^{-\zeta}}
+4n^2{\lambda^{-2\zeta}}
)+4n^2({\lambda^{-2z}}+{\lambda^{-2\zeta}})
  \nonumber \\
& = & 
b^2(1-2n\lambda^{-z})^2
+\beta^2(1-2n\lambda^{-\zeta})^2
+4n^2({\lambda^{-2z}}+{\lambda^{-2\zeta}})
  \nonumber \\
& = & b^2g(n{\lambda^{-z}})+\beta^2g(n{\lambda^{-\zeta}})+ 
  4n^2({\lambda^{-2z}}+{\lambda^{-2\zeta}}).\nonumber
\end{eqnarray}

\noindent
Writing $y=\max\s{g(n{\lambda^{-z}}), g(n{\lambda^{-\zeta}})}$ for 
brevity, and using the assumption that 
$\sk\state\geq\formula{\P}$ together with 
the fact that $c\leq z,\zeta$, we 
get:
\begin{eqnarray}
  \sk(\state \cdot A^{n}) & \leq & 
  b^2y+\beta^2y+ 8n^2\lambda^{-2c} \nonumber\\
  & = & \sk\state y+ 8n^2\lambda^{-2c} \nonumber\\
  & \leq & \sk\state (y+ 
  \frac{8}{\formula{\P}}n^2\lambda^{-2c}). \nonumber
\end{eqnarray}
To finish the proof, it remains to show that 
$y+ \frac{8}{\formula{\P}}n^2\lambda^{-2c}\leq 
\formula{\Qstep}$.  There are two cases:
\begin{itemize}
  \item If $\floor{\frac{\lambda^{c}}{2}}\geq 1$, 
  then $\frac{\lambda^{c}}{4}\leq n\leq 
  \frac{\lambda^{c}}{2}$. From $n\leq \frac{\lambda^{c}}{2}$, we have
  $2n\lambda^{-c}\leq 1$, and so  
  $\frac{8}{\formula{\P}}n^2\lambda^{-2c} \leq 
  \frac{2}{\formula{\P}}$. Moreover, because 
  $\bias\state\in[-1,1]$, we have 
  $c \leq  z,\zeta \leq c+1$. Hence
  $\frac{1}{4\lambda} = \frac{\lambda^{c}}{4} \lambda^{-c-1}
  \leq n\lambda^{-c-1}
  \leq n{\lambda^{-z}}, n{\lambda^{-\zeta}} 
  \leq n\lambda^{-c}
  \leq 
  \frac{1}{2}$. On the interval $[\frac{1}{4\lambda},
  \frac{1}{2}]$, the function $g(x)$ assumes its maximum at
  $x=\frac{1}{4\lambda}$. This implies that $y\leq 
  g(\frac{1}{4\lambda})$. This completes the 
  present case since we get:
  \[
  y+ \frac{8}{\formula{\P}}n^2\lambda^{-2c}\leq 
  g(\frac{1}{4\lambda})
  +\frac{2}{\formula{\P}} 
  \approx\formula{(1/(2*\l)-1)^2+2/\P}
  \leq \formula{\Qstep}.
  \]
  \assert{(1/(2*\l)-1)^2+2/\P <= \Qstep}%
  \item If $\floor{\frac{\lambda^{c}}{2}}< 1$, 
  then $n=1$ and $\lambda^c< 2$. From 
$\formula{-\a}\leq c$, we have
  \assert{(ln 0.5)/(ln ((sqrt 2)+1)) <= \a}%
  $\frac{8}{\formula{\P}}n^2\lambda^{-2c}\leq
  \frac{8}{\formula{\P}}\lambda^{\formula{2*\a}}$.  
  Moreover, because $\bias\state\in[-1,1]$, we have
  ${\formula{-\a}}\leq c \leq z,\zeta\leq c+1$. With
  $\lambda^c\leq 2$, this implies that $\frac{1}{2\lambda}\leq
  \lambda^{-c-1}\leq \lambda^{-z},\lambda^{-\zeta}
  \leq\lambda^{\formula{\a}}$. Therefore both $\lambda^{-z}$ and
  $\lambda^{-\zeta}$ are in the interval $[\frac{1}{2\lambda},
  \lambda^{\formula{\a}}]$. On this interval, the function $g(x)$
  assumes its maximum at
  \assert{(1-2*(\l^(\a)))^2 < (1-2*(1/(2*\l)))^2}%
  $x=\frac{1}{2\lambda}$, and therefore $y\leq 
  g(\frac{1}{2\lambda})$. This completes the 
  proof since:
  \[
  y+\frac{8}{\formula{\P}}n^2\lambda^{-2c}\leq 
  g(\frac{1}{2\lambda})
  + \frac{8}{\formula{\P}}\lambda^{\formula{2*\a}}
  \approx\formula{(max(((2*(pow(\l,\a))-1)^2),
  (((pow(\l,(-1)))-1)^2)) + ((8*(pow (\l, 2*\a))) / \P))}
  \leq \formula{\Qstep}.
  \]
  \assert{(max(((2*(pow(\l,\a))-1)^2),
  (((pow(\l,(-1)))-1)^2)) + ((8*(pow (\l, 2*\a))) / \P)) 
  <= \Qstep}%
  \end{itemize}
\end{proof}

\subsection[The B Lemma]{The $B$ Lemma}
\label{ssect-B}

\begin{definition}
  Let $h(x)=(1-\sqrt{2}x)^2$.
\end{definition}

\begin{lemma}
\label{lemmaB}
Recall the operator $B$ from Figure~\ref{list_operators}. 
If $\state $ is such that $\bias\state\in 
[-1,1]$, $\sk\state\geq\formula{\P}$, and such that $b\leq 0\leq\beta$ 
and $\formula{-\b}\leq z,\zeta$, then 
there exists $n\in\Z$ such that:
\[
  \sk(\state\cdot B^n) \leq \formula{\Qstep}~\sk\state.
\]
\end{lemma}

\begin{proof}
Let $c=\min\s{z, \zeta}$, 
$n=\max\s{1, \floor{\frac{\lambda^{c}}{\sqrt{2}}}}$ 
and compute the action of $B^n$ on $\state$:
\[
\begin{array}{lcl}
{B^n}^\dagger DB^n & = & \left[ \begin{array} {cc} 
1 & 0  \\
\sqrt{2}n & 1
\end{array} \right]
\left[ \begin{array} {cc} 
e{\lambda^{-z}} & b \\
b & e{\lambda^{z}}
\end{array} \right]
\left[ \begin{array} {cc} 
1 & \sqrt{2}n  \\
0 & 1
\end{array} \right] \\[9pt]
& = & \left[ \begin{array} {cc} 
\ldots & b+\sqrt{2}\,ne{\lambda^{-z}}  \\
b+\sqrt{2}\,ne{\lambda^{-z}} & \ldots
\end{array} \right], \\[18pt]
{B^n}^{\bullet\dagger} D{B^n}^\bullet & = & \left[ \begin{array} {cc} 
1 & 0  \\
-\sqrt{2}n & 1
\end{array} \right]
\left[ \begin{array} {cc} 
\varepsilon{\lambda^{-\zeta}} & \beta \\
\beta & \varepsilon{\lambda^{\zeta}}
\end{array} \right]
\left[ \begin{array} {cc} 
1 & -\sqrt{2}n  \\
0 & 1
\end{array} \right] \\[9pt]
& = & \left[ \begin{array} {cc} 
\ldots & \beta-\sqrt{2}\,n\varepsilon{\lambda^{-\zeta}}  \\
\beta-\sqrt{2}\,n\varepsilon{\lambda^{-\zeta}} & \ldots
\end{array} \right].
\end{array}
\]
Therefore:
\[
\sk(\state\cdot B^n) = (b+\sqrt{2}\,ne{\lambda^{-z}})^2+
(\beta-\sqrt{2}\,n\varepsilon{\lambda^{-\zeta}})^2.
\]
But recall that $e^2=b^2+1$, that 
$\varepsilon^2=\beta^2+1$, and from Remark~\ref{rem-be} that 
$b\leq 0\leq \beta$ implies $be\leq -b^2$ and 
$-\beta\varepsilon\leq -\beta^2$. Using these 
facts, we can expand the above formula as 
follows:
\begin{eqnarray}
\lefteqn{\sk(\state\cdot B^n)}\nonumber \\
& = & (b+\sqrt{2}\,ne{\lambda^{-z}})^2+
(\beta-\sqrt{2}\,n\varepsilon{\lambda^{-\zeta}})^2\nonumber\\
& = & 
b^2
+2\sqrt{2}\,nbe{\lambda^{-z}}
+2n^2e^2{\lambda^{-2z}}
+\beta^2
-2\sqrt{2}\,n\beta\varepsilon{\lambda^{-\zeta}}
+2n^2\varepsilon^2{\lambda^{-2\zeta}}
\nonumber\\
& \leq & 
b^2
-2\sqrt{2}\,nb^2{\lambda^{-z}}
+2n^2(b^2+1){\lambda^{-2z}}
+\beta^2
-2\sqrt{2}\,n\beta^2{\lambda^{-\zeta}}
+2n^2(\beta^2+1){\lambda^{-2\zeta}}
\nonumber\\
& = & b^2(
1
-2\sqrt{2}\,n{\lambda^{-z}}
+2n^2{\lambda^{-2z}}
)+\beta^2(
1
-2\sqrt{2}\,n{\lambda^{-\zeta}}
+2n^2{\lambda^{-2\zeta}}
)+2n^2({\lambda^{-2z}}+{\lambda^{-2\zeta}})\nonumber \\
& = & 
b^2(1-\sqrt{2}\,n\lambda^{-z})^2
+\beta^2(1-\sqrt{2}\,n{\lambda^{-\zeta}})^2
+ 2n^2({\lambda^{-2z}}+{\lambda^{-2\zeta}}).\nonumber \\
& = & b^2h(n{\lambda^{-z}})+\beta^2h(n{\lambda^{-\zeta}})+ 2n^2({\lambda^{-2z}}+{\lambda^{-2\zeta}}).\nonumber
\end{eqnarray}

\noindent
Writing $y=\max\s{h(n{\lambda^{-z}}), h(n{\lambda^{-\zeta}})}$ for
brevity, and using the assumption that $\sk\state\geq\formula{\P}$,
together with the fact that $c\leq z,\zeta$, we get:
\begin{eqnarray}
\sk(\state \cdot B^n) & \leq & 
b^2y+\beta^2y+ 4n^2\lambda^{-2c} \nonumber\\
& = & \sk\state y +4n^2\lambda^{-2c}\nonumber\\ 
&\leq & \sk\state (y +\frac{4}{\formula{\P}}n^2\lambda^{-2c}) 
\nonumber.
\end{eqnarray}

\noindent
To finish the proof, it remains to show that 
$y+\frac{4}{\formula{\P}}n^2\lambda^{-2c}\leq \formula{\Qstep}$.  There
are two cases:
\begin{itemize}
  \item If $\floor{\frac{\lambda^{c}}{\sqrt{2}}}\geq 1$, 
  then $\frac{\lambda^{c}}{2\sqrt{2}}\leq n\leq 
  \frac{\lambda^{c}}{\sqrt{2}}$. From $n\leq
  \frac{\lambda^{c}}{\sqrt{2}}$, we have $2n^2\lambda^{-2c}\leq 1$, and
  so
  $\frac{4n^2\lambda^{-2c}}{\formula{\P}}\leq \frac{2}
  {\formula{\P}}$. Moreover, because $\bias\state\in[-1,1]$, we have
  $c\leq z,\zeta\leq c+1$. Hence $\frac{1}{2\sqrt{2}\,\lambda}
  = \frac{\lambda^{c}}{2\sqrt{2}}\lambda^{-c-1}\leq n\lambda^{-c-1}\leq
  n{\lambda^{-z}}, n{\lambda^{-\zeta}} \leq n\lambda^{-c}\leq
  \frac{1}{\sqrt{2}}$. On the interval $[\frac{1}{2\sqrt{2}\,\lambda},
  \frac{1}{\sqrt{2}}]$, the function $h(x)$ assumes its maximum at
  $x=\frac{1}{2\sqrt{2}\,\lambda}$. This implies that $y\leq
  h(\frac{1}{2\sqrt{2}\,\lambda})$. This
  completes the present case since we get:
  \[
  y+ \frac{4}{\formula{\P}}n^2\lambda^{-2c}\leq 
  h(\frac{1}{2\sqrt{2}\,\lambda}) +\frac{2}{\P} 
  \approx \formula{(1 - sqrt(2)*(1/(2*sqrt(2)*\l)))^2 + (2/\P)}
  \leq \formula{\Qstep}.
  \]
  \assert{(1 - sqrt(2)*(1/(2*sqrt(2)*\l)))^2 + (2/\P) <= (\Qstep)}%
  \item If $\floor{\frac{\lambda^{c}}{\sqrt{2}}}< 1$, 
  then $n=1$ and $\lambda^c<\sqrt{2}$. 
  From $\formula{-\b}\leq c$, we have  $\frac{4}{\formula{\P}}n^2\lambda^{-2c}
  \leq \frac{4}{\formula{\P}}\lambda^{\formula{2*\b}}$. Moreover, because
  %
  $\bias\state\in[-1,1]$, we have
  $\formula{-\b}\leq c\leq z,\zeta\leq c+1$.
  With $\lambda^c\leq\sqrt{2}$, this implies that
  $\frac{1}{\sqrt{2}\,\lambda} \leq \lambda^{-c-1}\leq {\lambda^{-z}}, {\lambda^{-\zeta}} \leq 
  \lambda^{\formula{\b}}$. Therefore both $\lambda^{-z}$ and
  $\lambda^{-\zeta}$ are in the interval
  $[\frac{1}{\sqrt{2}\,\lambda},\lambda^{\formula{\b}}]$.
  On this interval, the function $h(x)$ assumes its maximum at
  \assert{(1- sqrt(2)*(1/(sqrt(2)*\l)))^2 <= (1-sqrt(2)*((\l)^(\b)))^2}%
  $x=\lambda^{\formula{\b}}$, and therefore $y\leq
  h(\lambda^{\formula{\b}})$. 
  This completes the proof since:
  \[
  y+\frac{4}{\formula{\P}}n^2\lambda^{-2c}\leq 
  h(\lambda^{\formula{\b}})
  + \frac{4}{\formula{\P}}\lambda^{\formula{2*\b}}
  \approx\formula{((max((1 - sqrt(2)*(1/(sqrt(2)*\l)))^2,
  (1-sqrt(2)*((\l)^(\b)))^2) + (4/\P * (\l^(2*\b))))}
  \leq \formula{\Qstep}.
  \]
  \assert{((max((1 - sqrt(2)*(1/(sqrt(2)*\l)))^2,
  (1-sqrt(2)*((\l)^(\b)))^2) + (4/\P * (\l^(2*\b)))) 
  <= (\Qstep)}%
\end{itemize}
\end{proof}

\subsection{Proof of the Step Lemma}
\label{ssect-step}

The proof of the Step Lemma is now basically a case distinction, using
the cases enumerated in Sections~\ref{ssect-shift}--\ref{ssect-B}, as
well as some additional symmetric cases. In particular, the following
remark will allow us to use the grid operators $X$ and $Z$ to reduce
the number of cases to consider.

\begin{remark}
\label{rem-XZ}
The grid operator $Z$ negates the anti-diagonal entries
of a state while the operator $X$ swaps the diagonal entries
of a state. This follows by simple computation:
  \[
    \state\cdot Z = 
    \left(
      \left[
        \begin{array}{cc}
          e{\lambda^{-z}} & -b \\
          -b & e{\lambda^{z}}
        \end{array} 
      \right]
    ,
      \left[
        \begin{array}{cc}
          \varepsilon{\lambda^{-\zeta}} & -\beta \\
          -\beta & \varepsilon{\lambda^{\zeta}}
        \end{array} 
      \right]
    \right),
  ~~~~
    \state\cdot X = 
    \left(
      \left[
        \begin{array}{cc}
          e{\lambda^{z}} & b \\
          b & e{\lambda^{-z}}
        \end{array} 
      \right]
      ,
      \left[
        \begin{array}{cc}
          \varepsilon{\lambda^{\zeta}} & \beta \\
          \beta & \varepsilon{\lambda^{-\zeta}}
        \end{array} 
      \right]
    \right).
  \]            
  Moreover, $\bias (\state\cdot Z) = \bias \state$ and 
  $\bias (\state\cdot X) = -\bias \state$.
\end{remark}

\begin{un-lemma}[Step Lemma]
For any state $\state$, if $\sk\state \geq \formula{\P}$, 
then there exists a special grid operator $\G$ such that 
$\sk (\state\cdot \G)\leq \formula{\Qstep} ~\sk\state$.
Moreover, $\G$ can be computed using a constant number of arithmetic
operations.
\end{un-lemma}

\begin{proof}
Let $\state$ be a state such that $\sk\state \geq 
\formula{\P}$. By Lemma~\ref{shift_lemma} we can assume 
w.l.o.g. that $\bias\state\in[-1,1]$. Moreover, by 
Remark~\ref{rem-XZ}, we can also assume that $\beta \geq 0$ 
and $z+\zeta\geq 0$. Note that the application of the grid operators
$X$ and/or $Z$ in Remark~\ref{rem-XZ} preserves the fact that 
$\bias\state\in[-1,1]$. We now treat in turn the cases 
$b\geq 0$ and $b\leq 0$.
\begin{description}
  \item[Case 1] \label{b_beta_plus} $b \geq 0$. A covering of 
  the strip defined by $z-\zeta\in[-1,1]$ and 
  $z+\zeta\geq 0$ is depicted in Figure~\ref{fig-covers}(a). 
  The $R$ region (in green) and the $A$ region (in red) are 
  defined as the intersection of this space with 
  $\s{(z,\zeta)~|~\formula{-\r}\leq z,\zeta \leq\formula{\r}}$ and 
  $\s{(z,\zeta)~|~z\leq\formula{-\a}\mbox{ and }\formula{\r}\leq \zeta}$
  respectively. The $K$ and $K^\bullet$ regions (both in blue) 
  fill the remaining space.
  \assert{-\r <= \a}%
  \begin{figure}
  \[ (a)~ \mp{0.9}{\begin{tikzpicture}[scale=1.5]
      \fill[fill=blue!25] (-0.3,-0.8) -- (-0.3,-2.4) .. controls (0.3,-2.4) and
      (0.5,-1.5) .. (0.5,-0.8) -- cycle;
      \draw (0.25,-1.2) node {$XK$};
      \begin{scope}[scale=-1]
        \fill[fill=blue!25] (-0.3,-0.8) -- (-0.3,-2.4) .. controls (0.3,-2.4) and
        (0.5,-1.5) .. (0.5,-0.8) -- cycle;
        \draw (0.15,-1.3) node {$K$};
      \end{scope}
      \begin{scope}[cm={0,1,1,0,(0,0)}]
        \fill[fill=blue!25] (-0.3,-0.8) -- (-0.3,-2.4) .. controls (0.3,-2.4) and
        (0.5,-1.5) .. (0.5,-0.8) -- cycle;
        \draw (0.15,-1.3) node {$XK\bul$};
      \end{scope}
      \begin{scope}[cm={0,-1,-1,0,(0,0)}]
        \fill[fill=blue!25] (-0.3,-0.8) -- (-0.3,-2.4) .. controls (0.3,-2.4) and
        (0.5,-1.5) .. (0.5,-0.8) -- cycle;
        \draw (0.2,-1.3) node {$K\bul$};
      \end{scope}
      %
      \fill[fill=red!25] (-2.4,-0.3) -- (-0.3,-0.3) -- 
      (-0.3,-2.4) .. controls (-2,-2.4) and (-2.4,-2) ..
      (-2.4,-0.3) -- cycle;
      \draw (-2.5,-0.3) -- (-0.3,-0.3) -- (-0.3,-2.5);
      \draw (-1.2, -1.2) node {\small $XA^n$};
      %
      \begin{scope}[scale=-1]
        \fill[fill=red!25] (-2.4,-0.3) -- (-0.3,-0.3) -- 
        (-0.3,-2.4) .. controls (-2,-2.4) and (-2.4,-2) ..
        (-2.4,-0.3) -- cycle;
        \draw (-2.5,-0.3) -- (-0.3,-0.3) -- (-0.3,-2.5);
        \draw (-1.2, -1.2) node {\small $A^n$};
      \end{scope}
      \filldraw[fill=green!25] (0.8, 0.8) -- (-0.8, 0.8) -- (-0.8, -0.8) -- (0.8, -0.8) -- cycle;
      \draw (0.4, 0.4) node {\small $R$};
      \begin{scope}
        \path[clip] (-1.5,-2.5) -- (2.5,1.5) -- (1.5,2.5) -- (-2.5,-1.5) -- cycle;
        \begin{scope}[yscale=-1]
          \fill[fill=blue!50] (-0.3,-2.4) -- (-0.3,-0.8) -- 
          (0.3,-0.8) -- (0.3,-2.4) -- cycle;
          \draw ((-0.3,-2.5) -- (-0.3,-0.8) -- (0.3,-0.8) -- (0.3,-2.5);
        \end{scope}
        \begin{scope}[cm={0,-1,-1,0,(0,0)}]
          \fill[fill=blue!50] (-0.3,-2.4) -- (-0.3,-0.8) -- 
          (0.3,-0.8) -- (0.3,-2.4) -- cycle;
          \draw ((-0.3,-2.5) -- (-0.3,-0.8) -- (0.3,-0.8) -- (0.3,-2.5);
        \end{scope}
        %
        \begin{scope}[scale=-1]
          \fill[fill=red!50] (-2.4,-0.3) -- (-0.3,-0.3) -- 
          (-0.3,-2.4) .. controls (-2,-2.4) and (-2.4,-2) ..
          (-2.4,-0.3) -- cycle;
          \draw (-2.5,-0.3) -- (-0.3,-0.3) -- (-0.3,-2.5);
          \draw (-1.2, -1.2) node {\small $A^n$};
        \end{scope}
        \filldraw[fill=green!50] (0.8, 0.8) -- (-0.8, 0.8) -- (0.8, -0.8) -- cycle;
        \draw (0.4, 0.4) node {\small $R$};
      \end{scope}  
      %
      \draw[->] (-2.5,0) -- (2.5, 0) node[right] {$z$};
      \draw[->] (0,-2.5) -- (0, 2.5) node[above] {$\zeta$};
      \foreach \x in {-2, -1, 1, 2} {
        \draw (\x,0) -- (\x,-0.1);
      }
      \foreach \y in {-2, -1, 1, 2} {
        \draw (0,\y) -- (-0.1,\y);
      }
      \draw (1,-0.1) node[below] {\small $1$};
      \draw (-0.1,1) node[left] {\small $1$};
      %
      \draw (-1.5,-2.5) -- (2.5,1.5);
      \draw (-2.5,-1.5) -- (1.5,2.5);
      \draw [dashed] (-2, 2) -- (-0.5,0.5);    
      \begin{scope}[scale=-1]
      \draw [dashed] (-2, 2) -- (-0.5,0.5);          
      \end{scope}
    \end{tikzpicture}
  }
  \quad
  (b)~
   \mp{0.9}{\begin{tikzpicture}[scale=1.5]
      %
      \fill[fill=red!25] (-2.4,0.2) -- (0.2,0.2) -- 
      (0.2,-2.4) .. controls (-2,-2.4) and (-2.4,-2) ..
      (-2.4,0.2) -- cycle;
      \draw (-2.5,0.2) -- (0.2,0.2) -- (0.2,-2.5);
      \draw (-1.2, -1.2) node {\small $XB^n$};
      %
      \begin{scope}[scale=-1]
        \fill[fill=red!25] (-2.4,0.2) -- (0.2,0.2) -- 
        (0.2,-2.4) .. controls (-2,-2.4) and (-2.4,-2) ..
        (-2.4,0.2) -- cycle;
        \draw (-2.5,0.2) -- (0.2,0.2) -- (0.2,-2.5);
        \draw (-1.2, -1.2) node {\small $B^n$};
      \end{scope}
      \filldraw[fill=green!25] (0.8, 0.8) -- (-0.8, 0.8) -- (-0.8, -0.8) -- (0.8, -0.8) -- cycle;
      \draw (0.4, 0.4) node {\small $R$};
      \begin{scope}
        \path[clip] (-1.5,-2.5) -- (2.5,1.5) -- (1.5,2.5) -- (-2.5,-1.5) -- cycle;
        %
        %
        \begin{scope}[scale=-1]
          \fill[fill=red!50] (-2.4,0.2) -- (-0.8,0.2) -- (0.2, -0.8) -- 
          (0.2,-2.4) .. controls (-2,-2.4) and (-2.4,-2) ..
          (-2.4,0.2) -- cycle;
          \draw (-2.5,0.2) -- (-0.8,0.2) -- (0.2, -0.8) -- (0.2,-2.5);
          \draw (-1.2, -1.2) node {\small $B^n$};          
        \end{scope}
        \filldraw[fill=green!50] (0.8, 0.8) -- (-0.8, 0.8) -- (0.8, -0.8) -- cycle;
        \draw (0.4, 0.4) node {\small $R$};
      \end{scope}  
      %
      \draw[->] (-2.5,0) -- (2.5, 0) node[right] {$z$};
      \draw[->] (0,-2.5) -- (0, 2.5) node[above] {$\zeta$};
      \foreach \x in {-2, -1, 1, 2} {
        \draw (\x,0) -- (\x,-0.1);
      }
      \foreach \y in {-2, -1, 1, 2} {
        \draw (0,\y) -- (-0.1,\y);
      }
      \draw (1,-0.1) node[below=1.2ex] {\small $1$};
      \draw (-0.1,1) node[left=1.2ex] {\small $1$};
      %
      \draw (-1.5,-2.5) -- (2.5,1.5);
      \draw (-2.5,-1.5) -- (1.5,2.5);
      \draw [dashed] (-2, 2) -- (-0.5,0.5);    
      \begin{scope}[scale=-1]
      \draw [dashed] (-2, 2) -- (-0.5,0.5);          
      \end{scope}
    \end{tikzpicture}
  }
  \]
  \caption{(a) A covering of the region $z-\zeta\in [-1,1]$ and 
  $z+\zeta\geq 0$ for the case $b\geq 0$. (b) A covering of the 
  region $z-\zeta\in [-1,1]$ and $z+\zeta\geq 0$ for the case 
  $b\leq 0$.}
  \label{fig-covers}
  \rule{\textwidth}{0.1mm}
  \end{figure}
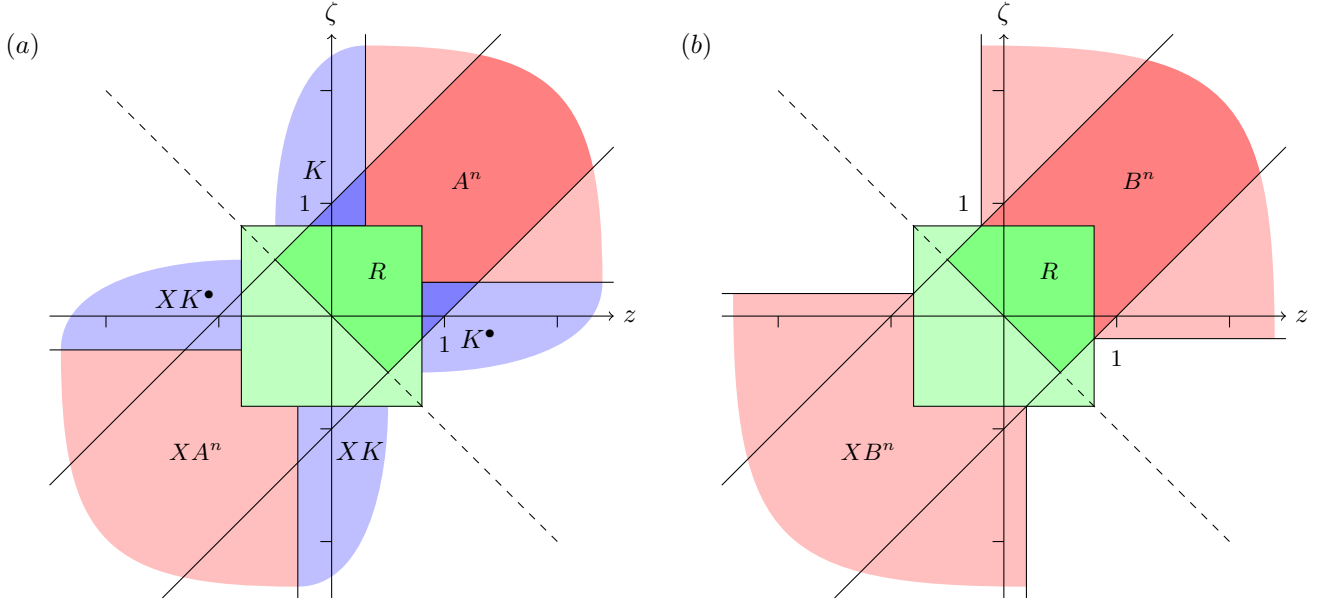
  We now consider in turn the possible locations of 
  the pair $(z,\zeta)$ in this covering. 
  \begin{enumerate}
    \item If $\formula{-\r}\leq z,\zeta \leq\formula{\r}$, 
    then $\sk(\state\cdot R)\leq \formula{\Qstep}~\sk\state$ by 
    Lemma~\ref{lemmaR}.
    \item \label{partk}If $z\leq\formula{-\a}$ and 
    $\formula{\r}\leq \zeta$, then 
    $\sk(\state\cdot K)\leq \formula{\Qstep}~\sk\state$ by 
    Lemma~\ref{lemmaK}.    
    \item If $\formula{-\a}\leq z,\zeta$, then there exists 
    $n\in\Z$ such that $\sk(\state\cdot A^n)\leq 
    \formula{\Qstep}~\sk\state$ by Lemma~\ref{lemmaA}.
    \item \label{partsk} If $\formula{\r}\leq z$ and 
    $\zeta\leq\formula{-\a}$, then note that 
    $\state\cdot K^\bullet = (\Delta, D)\cdot K$, and 
    therefore by Lemma~\ref{lemmaK}:
    \[
    \sk(\state\cdot K^\bullet) = \sk ((\Delta, D)\cdot K) 
    \leq \formula{\Qstep}~ \sk (\Delta, D)=
    \formula{\Qstep}~ \sk\state. 
    \]    
  \end{enumerate}
  \item[Case 2] $b\leq 0$. As above, we use a covering of 
  the strip defined by $z-\zeta\in[-1,1]$ and $z+\zeta\geq 0$ and 
  consider the possible locations of $(z, \zeta)$ in this 
  space. The relevant covering is depicted in 
  Figure~\ref{fig-covers}(b), where the $R$ region (in green) 
  is defined as above and the $B$ region (in red) is defined 
  as the intersection of the strip with 
  $\s{(z,\zeta)~|~z,\zeta\geq \formula{-\b}}$.
  \assert{-\b <= \r -1}%
  \begin{enumerate}
    \item If $\formula{-\r}\leq z,\zeta\leq\formula{\r}$, 
    then $\sk(\state\cdot R)\leq \formula{\Qstep}~\sk\state$ by 
    Lemma~\ref{lemmaR}.
    \item If $z,\zeta\geq \formula{-\b}$ then there exists 
    $n\in\Z$ such that $\sk(\state\cdot B^n)\leq 
    \formula{\Qstep}~\sk\state$ by Lemma~\ref{lemmaB}.
  \end{enumerate}
\end{description}
Finally, note that only a constant number of calculations are required
to decide which of the above cases applies. Moreover, each case only
requires a constant number of operations. Specifically, the
computation of $k$ and $\sigma^k$ in Lemma~\ref{shift_lemma}, of $n$
and $A^n$ in Lemma~\ref{lemmaA}, and of $n$ and $B^n$ in
Lemma~\ref{lemmaB} each require just a fixed number of operations, and
each of the remaining cases produces a fixed grid operator.
\end{proof}

\section[Proof of Proposition 5.17]{Proof of Proposition~\ref{prop-enclosing-ellipse}}
\label{app-enclosing-ellipse}

We prove Proposition~\ref{prop-enclosing-ellipse}, whose statement we 
reproduce here.

\begin{un-proposition}
  Let $A$ be a bounded convex subset of $\R^2$ with non-empty
  interior. Then there exists an ellipse $E$ such that $A\seq E$, and
  such that 
  \begin{equation}\label{eqn-area-EA}
    \area(E) \leq \frac{4\pi}{3\sqrt 3}\area(A). 
  \end{equation}
\end{un-proposition}

\begin{proof}
  We may assume without loss of generality that $A$ is compact, for it
  it is not, we can replace $A$ by its closure, which has the same
  area as $A$ because $A$ is convex. Let $\Disk$ be the closed unit
  disk, and consider the collection $\Aff(A,\Disk)$ of all affine
  transformations $f:\R^2\to\R^2$ satisfying $f(A)\seq\Disk$. Then
  $\Aff(A,\Disk)$, with the natural topology, is a compact
  set. Therefore, there exists some $f\in\Aff(A,\Disk)$ maximizing the
  area of $f(A)$.  We claim that $f$ is invertible. Indeed, since $A$
  is bounded, there exists some $\lambda>0$ such that $\lambda
  A\seq\Disk$. Since $\lambda A$ has non-zero area, and multiplication
  by $\lambda$ is an affine map, it follows that $f(A)$ has non-zero
  area as well, and so $f$ is invertible. Let $E=f^{-1}(\Disk)$. We
  claim that $E$ is the desired ellipse.
  \[ \m
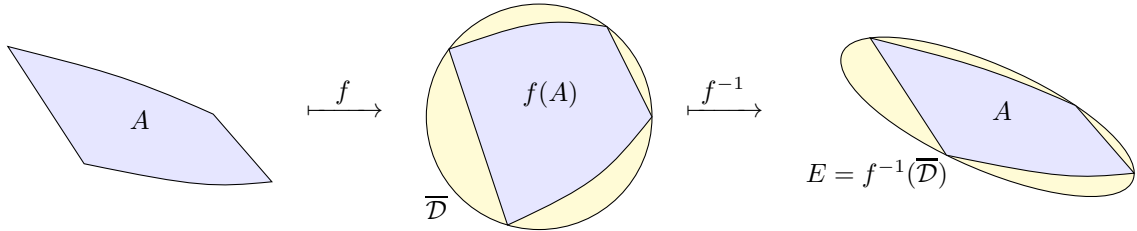
{\begin{tikzpicture}[scale=1.5,cm={1.3,-0.5,0,0.5,(0,0)}]
    \draw[fill=blue!10] (1,0) -- (0.6,0.8) .. controls (0,0.9) and (-0.2,0.8)
    .. (-0.8,0.6) 
    -- (-0.28,-0.96) .. controls (0.5,-0.6) and (0.6,-0.5) .. (1,0) --
    cycle;
    \draw (0.1,0.2) node {$A$};
    \def\dot#1{\draw (#1) node {$\bullet$};}
  \end{tikzpicture}}
  \quad\xmapsto{\makebox[8mm][c]{$f$}}\quad
  \m{\begin{tikzpicture}[scale=1.5]
    \draw[fill=yellow!20] (0,0) circle (1);
    \draw[fill=blue!10] (1,0) -- (0.6,0.8) .. controls (0,0.9) and (-0.2,0.8)
    .. (-0.8,0.6) 
    -- (-0.28,-0.96) .. controls (0.5,-0.6) and (0.6,-0.5) .. (1,0) --
    cycle;
    \draw (0.1,0.2) node {$f(A)$};
    \draw (-0.6,-0.8) node[left=2mm] {$\Disk$};
    \def\dot#1{\draw (#1) node {$\bullet$};}
  \end{tikzpicture}}
  \quad\xmapsto{\makebox[8mm][c]{$f^{-1}$}}\quad
  \m{\begin{tikzpicture}[scale=1.5,cm={1.3,-0.5,0,0.5,(0,0)}]
    \draw[fill=yellow!20] (0,0) circle (1);
    \draw[fill=blue!10] (1,0) -- (0.6,0.8) .. controls (0,0.9) and (-0.2,0.8)
    .. (-0.8,0.6) 
    -- (-0.28,-0.96) .. controls (0.5,-0.6) and (0.6,-0.5) .. (1,0) --
    cycle;
    \draw (0.1,0.2) node {$A$};
    \draw (0,-1) node[left=4mm] {$E=f^{-1}(\Disk)$};
    \def\dot#1{\draw (#1) node {$\bullet$};}
  \end{tikzpicture}}
  \]
  We have $A\seq E$ by construction. We must show
  {\eqref{eqn-area-EA}}. Because affine transformations preserve ratios
  of areas, we may equivalently show
  \[ \area(\Disk) \leq \frac{4\pi}{3\sqrt 3}\area(f(A)).
  \]
  Let $\del\Disk$ be the boundary of $\Disk$, and consider points $p$
  where $f(A)$ ``touches'' the boundary, i.e., points $p\in f(A)\cap
  \del\Disk$. We claim that any arc segment of $\del\Disk$ of length
  $2\pi/3$ radians ($120$ degrees) contains at least one such
  point. To prove this, assume, for the sake of contradiction, that
  there is such an arc segment $Q$ not containing any point of
  $f(A)$. By rotational symmetry, we may without loss of generality
  assume that $Q$ is the arc from $-\pi/3$ to $\pi/3$ radians on the
  unit circle. Let $z_1$ and $z_2$ be the endpoints of $Q$, as shown here:
  \[
  \m
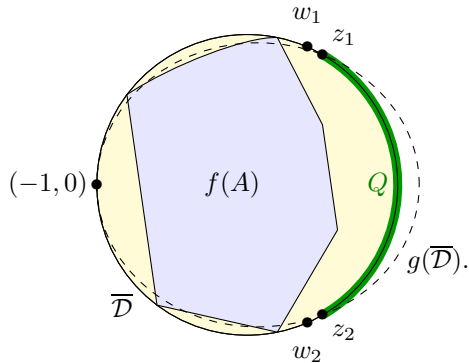
{\begin{tikzpicture}[scale=2]
    \draw[fill=yellow!20] (0,0) circle (1);
    \draw[line width=1.2mm, color=green!60!black] (0.5,-0.8660254037) arc (-60:60:1);
    \draw (0,0) circle (1);
    \draw[fill=blue!10] 
    (0.2,-.9797958971) -- (0.6,-0.3) -- (0.5,0.4) -- (0.2,.9797958971) .. controls (0,0.99) and (-0.5,0.8)
    .. (-0.8,0.6) 
    -- (-0.6,-0.8) -- (0.2,-.9797958971) --
    cycle;
    \draw[dashed, cm={1.0714285714,0,0,.9449111825,(.0714285714,0)}] (0,0) circle (1);
    \draw (1,-0.5) node [right] {$g(\Disk)$.};
    \draw (-0.1,0) node {$f(A)$};
    \draw (1,0) node[left, color=green!50!black] {$Q$};
    \draw (-0.6,-0.8) node[left=2mm] {$\Disk$};
    \def\dot#1{\draw (#1) node {$\bullet$};}
    \dot{0.4,.9165151389};  \draw (0.4,.9165151389) node[above=2mm] {$w_1$};
    \dot{0.4,-.9165151389}; \draw (0.4,-.9165151389) node[below=2mm] {$w_2$};
    \dot{0.5,.8660254037};  \draw (0.5,.8660254037) node[above right] {$z_1$};
    \dot{0.5,-.8660254037}; \draw (0.5,-.8660254037) node[below right] {$z_2$};
    \dot{-1,0}; \draw (-1,0) node[left] {$(-1,0)$};
  \end{tikzpicture}}
  \]
  Since both $Q$ and $f(A)$ are compact, there exists some $d>0$ such
  that the distance between any point of $A$ and any point of $Q$ is
  at least $d$. Let $w_1$ and $w_2$ be the two points on the unit
  circle whose distance from $Q$ is $d/2$. Now consider the affine
  transformation $g$ that fixes $(-1,0)$ and maps $w_1$ to $z_1$ and
  $w_2$ to $z_2$. Then $g(\Disk)$ is an ellipse whose boundary passes
  through the points $(-1,0)$, $z_1$, and $z_2$, shown as a dashed
  line in the above illustration. It is therefore bisected by
  $Q$. Since the map $g$ moves no point of the unit disk by more than
  distance $d$, the set $g(f(A))$ does not intersect $Q$. It follows
  that $g(f(A))$ is contained in $\Disk$. On the other hand, a
  calculation shows that the area of $g(f(A))$ is slightly greater
  than that of $f(A)$, contradicting the assumption that the area of
  $f(A)$ was maximal.

  We have proved that every arc segment of length $2\pi/3$ radians on
  the boundary of $\Disk$ contains a point of $f(A)$. It follows that there
  is some finite cyclic sequence of points $p_1,\ldots,p_n\in f(A)\cap
  \del\Disk$ such that consecutive points are no more than $2\pi/3$
  radians apart. By connecting each $p_i$ to the center, we partition
  each of the sets $f(A)$ and $\Disk$ into $n$ pieces $B_1,\ldots,B_n$ and
  $C_1,\ldots,C_n$, respectively.
  \[
  \m
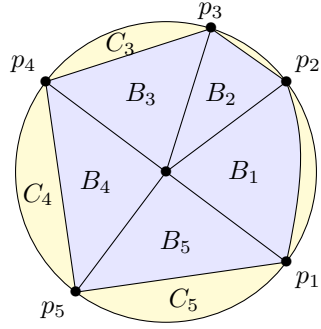
{\begin{tikzpicture}[scale=2]
    \draw[fill=yellow!20] (0,0) circle (1);
    \draw[fill=blue!10] (0.8,-0.6) .. controls (0.95,0.1) and
    (0.9,0.3) .. (0.8,0.6) -- (0.3,.9539392014) -- (-0.8,0.6) 
    -- (-0.6,-0.8) -- (0.8,-0.6) --
    cycle;
    \draw (0.52,0) node {$B_1$};
    \draw (.366,0.518) node {$B_2$};
    \draw (-0.166,.518) node {$B_3$};
    \draw (-.466,-.066) node {$B_4$};
    \draw (.066,-.466) node {$B_5$};
    \draw (-.3,.85) node {$C_3$};
    \draw (-.85,-0.15) node {$C_4$};
    \draw (.12,-.85) node {$C_5$};
    \def\dot#1{\draw (#1) node {$\bullet$};}
    \dot {0,0};
    \dot {0.8,-0.6};    
    \draw(0.8,-0.6) node[below right] {$p_1$} -- (0,0);
    \dot {0.8,0.6};     
    \draw(0.8,0.6) node[above right] {$p_2$} -- (0,0);
    \dot {0.3,.9539392014};   
    \draw(0.3,.9539392014) node[above] {$p_3$} -- (0,0);
    \dot {-0.8,0.6};    
    \draw(-0.8,0.6) node[above left] {$p_4$} -- (0,0);
    \dot {-0.6,-0.8};   
    \draw(-0.6,-0.8) node[below left] {$p_5$} -- (0,0);
  \end{tikzpicture}}
  \]
  The fact that the inner angles are less than $2\pi/3$ immediately
  implies that $\area(C_i)\leq \frac{4\pi}{3\sqrt 3}\area(B_i)$ for
  all $i$; hence also $ \area(\Disk) \leq \frac{4\pi}{3\sqrt
    3}\area(f(A))$. This finishes the proof of the proposition.
\end{proof}

\section[Proof of Theorem 6.2]{Proof of Theorem~\ref{thm-diophantine}}
\label{app-rings}

Consider the rings $\Z$ and $\D$, together with their respective
extensions $\Z[\sqrt2]$, $\Z[\omega]$ and $\D[\sqrt2]$, $\D[\omega]$,
as introduced in Section~\ref{sec-algebra}. We wish to give an
efficient method for solving equations of the form $t\da t=\xi$, for
given $\xi\in\D[\sqrt2]$ and unknown $t\in\D[\omega]$. To do this, we
first classify the primes in $\Z[\sqrt2]$ and $\Z[\omega]$, then show
how to solve the equation $t\da t=\xi$ in the case where
$\xi\in\Z[\sqrt2]$ and $t\in\Z[\omega]$, which we finally extend to
the general case. None of the results in this section are original;
they are well-known from the theory of cyclotomic fields, which goes
back to the 19th century work of Gauss and Kummer.  Nevertheless, we
hope that the following detailed treatment may be useful to readers
who are not experts in algebraic number theory. Moreover, we hope to
give sufficient details to enable an interested reader, in principle,
to implement the algorithm. This will also aid in our complexity
analysis.

A fundamental property of the rings $\Z$, $\Z[\sqrt2]$, and
$\Z[\omega]$ is that they are {\em Euclidean domains}; this implies
that the notions of divisibility, greatest common divisor, and unique
prime factorization all make sense in these rings.  Recall that in a
ring, a {\em unit} is an invertible element.  In a Euclidean domain,
we write $x\divides y$ if $x$ is a divisor of $y$, and $x\sim y$ if
$x\divides y$ and $y\divides x$; equivalently, $x\sim y$ iff there
exists a unit $u$ that $x u=y$. An element $x$ is {\em prime} if $x$
is not a unit, and $x=ab$ implies that either $a$ or $b$ is a
unit. Note that if $x$ is prime and $x\divides ab$, then $x\divides a$
or $x\divides b$; this follows from Euclid's algorithm.

\subsection[{Units in Z[sqrt 2]}]{Units in $\Z[\sqrt2]$}

\begin{definition}
  We say that $\xi\in\Z[\sqrt 2]$ is {\em doubly positive} if $\xi\geq
  0$ and $\xi\bul \geq 0$.
\end{definition}

\begin{lemma}\label{lem-units}
  The units of $\Z[\sqrt2]$ are of the form $u=(-1)^n\lambda^m$, where
  $\lambda=1+\sqrt 2$. Moreover, a unit $u$ is doubly positive if and
  only if $u$ is a square in $\Z[\sqrt2]$.
\end{lemma}

\begin{proof}
  Lemma~10 of {\cite{Selinger-newsynth}}.
\end{proof}

\begin{lemma}
  Let $\xi\in\Z[\sqrt2]$, and consider $n=\xi\bul\xi$. If $n$ is a unit
  in $\Z$, then $\xi$ is a unit in $\Z[\sqrt2]$.
\end{lemma}

\begin{proof}
  If $n$ is a unit, then there exists $m\in\Z$ such that $nm=1$, hence
  $\xi\xi\bul m=1$, hence $\xi$ is invertible in $\Z[\sqrt2]$.
\end{proof}

\subsection[{Primes in Z[sqrt 2]}]{Primes in $\Z[\sqrt2]$}

\begin{lemma}\label{lem-prime-criterion}
  Let $\xi\in\Z[\sqrt2]$, and consider $n=\xi\bul\xi$. If $n$ is prime
  in $\Z$, then $\xi$ is prime in $\Z[\sqrt2]$.
\end{lemma}

\begin{proof}
  Suppose $\xi=\alpha\beta$ in $\Z[\sqrt 2]$, and consider $n =
  \xi\bul\xi = \alpha\bul\alpha \beta\bul\beta$. Since
  $\alpha\bul\alpha$ and $\beta\bul\beta$ are integers and $n$ is
  prime, we must have that either $\alpha\bul\alpha$ or
  $\beta\bul\beta$ is a unit in $\Z$; hence $\alpha$ or $\beta$ is a
  unit in $\Z[\sqrt 2]$. So $\xi$ is prime.
\end{proof}

\begin{lemma}
  For every prime $\xi$ of $\Z[\sqrt2]$, there exists a unique (up to
  a unit) prime $p$ of $\Z$ such that $\xi \divides p$.
\end{lemma}

\begin{proof}
  To show existence, consider $n=\xi\bul\xi$. Note that $\xi\neq 0$,
  hence $n\neq 0$. Let $n=p_1p_2\cdots p_k$ be a prime factorization
  of $n$. Since $\xi \divides p_1p_2\cdots p_k$ and $\xi$ is prime in
  $\Z[\sqrt2]$, there exists some $i$ such that $\xi\divides p_i$.

  To show uniqueness, assume $\xi\divides p$ and $\xi\divides q$, where $p\not\sim q$.
  Then $\gcd(p,q)=1$, hence by Euclid's algorithm, we can write
  $1=np+mq$, for integers $n$ and $m$. Then $\xi\divides np+mq=1$, which
  is absurd since $\xi$ is prime. 
\end{proof}

\begin{lemma}\label{lem-z2-cases}
  Let $\xi$ be a prime of $\Z[\sqrt2]$, and let $p$ be the unique (up
  to a unit) prime of $\Z$ such that $\xi\divides p$. Then exactly one of
  the following holds: $\xi\sim p$ or $\xi\bul\xi\sim p$.
\end{lemma}

\begin{proof}
  First note that at most one of these properties can hold, because
  otherwise $\xi\sim\xi\bul\xi$, which implies that $\xi\bul$ is a
  unit, which is absurd since it is prime.

  We now show that at least one of the properties holds. Since
  $\xi\divides p$ and $p$ is an integer, we have $\xi\bul\divides p$, hence
  $\xi\bul\xi\divides p^2$. Also, $\xi\bul\xi$ is an integer, so either
  $\xi\bul\xi\sim 1$, $\xi\bul\xi\sim p$, or $\xi\bul\xi\sim p^2$. The
  first of these cases is absurd, since $\xi$ would then be a unit. In
  the second case, we have $\xi\bul\xi\sim p$, which is to be
  shown. In the third case, since $\xi\divides p$, there is
  $\alpha\in\Z[\sqrt2]$ such that $\xi\alpha=p$. Then we have
  $\xi\bul\xi\alpha\bul\alpha=p^2\sim\xi\bul\xi$, which implies that
  $\alpha$ is a unit, hence $\xi\sim p$.
\end{proof}

\begin{lemma}\label{lem-prime-factor-z2}
  Let $p$ be a prime of $\Z$. Then the prime factorization of $p$ in
  $\Z[\sqrt2]$ consists of one or two factors. 
\end{lemma}

\begin{proof}
  Let $\xi$ be some prime factor of $p$ in $\Z[\sqrt2]$. By
  Lemma~\ref{lem-z2-cases}, either $\xi\sim p$ or $\xi\bul\xi\sim
  p$. Either way, this gives a prime factorization of $p$ in $\Z[\sqrt2]$.
\end{proof}

We now determine the prime factorization in $\Z[\sqrt2]$ of every
prime $p$ of $\Z$. It turns out that there are 5 cases, depending
whether $p$ is even, or $p\equiv 1, 3, 5, 7\mmod{8}$.

\begin{lemma}\label{lem-2}
  The prime factorization of $p=2$ in $\Z[\sqrt2]$ is $p =
  \sqrt{2}\cdot \sqrt{2}$. The prime factorization of $p=-2$ in
  $\Z[\sqrt2]$ is $p=-\sqrt{2}\cdot\sqrt{2}$.
\end{lemma}

\begin{proof}
  We only need to show that $\sqrt2$ is prime in $\Z[\sqrt2]$, but
  this follows from Lemma~\ref{lem-prime-criterion}.
\end{proof}

\begin{lemma}\label{lem-3-5}
  Let $p$ be a prime of $\Z$ such that $p\equiv 3\mmod{8}$ or $p\equiv
  5\mmod{8}$. Then $p$ is prime in $\Z[\sqrt2]$.
\end{lemma}

\begin{proof}
  Let $\xi$ be a prime factor of $p$ in $\Z[\sqrt2]$. By
  Lemma~\ref{lem-z2-cases}, we know that $\xi\sim p$ or
  $\xi\bul\xi\sim p$. In the former case, $p$ is prime and we are
  done. In the latter case, writing $\xi=a+b\sqrt2$, we have
  $\xi\bul\xi = a^2-2b^2 =\pm p$, hence $a^2-2b^2\equiv\pm3\mmod{8}$.
  By easy case distinction, we see that $a^2$ can only be $0$, $1$, or
  $4$ modulo $8$, and $2b^2$ can only be $0$ or $2$ modulo $8$, so
  $a^2-2b^2\equiv\pm3\mmod{8}$ is plainly impossible.
\end{proof}

\begin{lemma}\label{lem-associate-primes}
  Let $\xi\in\Z[\sqrt2]$ such that $\xi\sim\xi\bul$. Then either
  $\xi\sim n$, or $\xi\sim n\sqrt{2}$, for some $n\in\Z$.
\end{lemma}

\begin{proof}
  By assumption, $\xi\sim\xi\bul$, so there exists a unit
  $u\in\Z[\sqrt2]$ such that $\xi\bul=u\xi$. By Lemma~\ref{lem-units}, the 
  units of $\Z[\sqrt2]$ are exactly of the form $u=(-1)^n\lambda^m$, where
  $\lambda=1+\sqrt2$. Applying the automorphism to $\xi\bul=u\xi$, we get 
  $\xi=u\bul\xi\bul$, hence $\xi\bul\xi = u\bul u\xi\bul\xi$, hence 
  $u\bul u=1$, hence $(\lambda\bul\lambda)^m = 1$. But $\lambda\bul\lambda=-1$, 
  so $m$ is even; say $m=2k$. Let $\xi' = \lambda^k\xi$. Note that
  $\xi'\sim\xi$. Then we have
  \begin{equation}\label{eqn-xi-prime}
    \xi'^\bullet = (\lambda\bul)^k\xi\bul = (-\lambda)^{-k} u\xi
    = (-1)^{-k}\lambda^{-k}(-1)^n\lambda^{m}\xi = \pm\lambda^k\xi = \pm \xi'.
  \end{equation}
  We can write $\xi' = a+b\sqrt2$ for some $a,b\in\Z$. From
  {\eqref{eqn-xi-prime}}, we have either $a-b\sqrt 2 = a+b\sqrt 2$, or
  $a-b\sqrt 2 = -(a+b\sqrt 2)$. In the first case, $b=0$, and therefore
  $\xi\sim\xi'= a$. In the second case, $a=0$, and therefore
  $\xi\sim\xi'= b\sqrt 2$.
\end{proof}

\begin{lemma}\label{lem-1-7}
  Let $p$ be a prime of $\Z$ such that $p\equiv 1\mmod{8}$ or $p\equiv
  7\mmod{8}$. Then $p$ has a prime factorization of the form $p\sim
  \xi\bul\xi$ in $\Z[\sqrt2]$; moreover, $\xi\not\sim\xi\bul$, so that
  the two prime factors are distinct.
\end{lemma}

\begin{proof}
  It is well-known (by quadratic reciprocity) that if $p$ is a prime
  with $p\equiv\pm1\mmod{8}$, then $2$ is a quadratic residue modulo
  $p$. Therefore, the equation $x^2\equiv 2\mmod{p}$ has an integer
  solution; let $x$ be such a solution. Let $\alpha=x+\sqrt2$. Then
  $\alpha\bul\alpha=x^2-2\equiv0\mmod{p}$, and hence
  $p\divides\alpha\bul\alpha$. On the other hand, clearly $p\nmid\alpha$
  (since $\alpha/p\not\in\Z[\sqrt2]$), and similarly
  $p\nmid\alpha\bul$, which shows that $p$ is not a prime in
  $\Z[\sqrt2]$. By Lemma~\ref{lem-z2-cases}, $p\sim\xi\bul\xi$ for
  some prime $\xi$ of $\Z[\sqrt 2]$. 

  The final claim follows from Lemma~\ref{lem-associate-primes},
  because if $\xi\sim\xi\bul$, then either $\xi\sim n$ or $\xi\sim
  n\sqrt2$. In the first case, $n^2\divides \xi\bul\xi \sim p$, but $p$ is
  prime, so that $n=\pm 1$, contradicting that $\xi$ is prime. In the
  second case, $2n^2\divides\xi\bul\xi\sim p$, contradicting the
  assumption that $p$ is odd.
\end{proof}

\begin{lemma}\label{lem-prime-factorization-z2}
  Let $p$ be a prime of $\Z$. A prime factorization of $p$ in
  $\Z[\sqrt2]$ can be computed in probabilistic polynomial time.
\end{lemma}

\begin{proof}
  Assume without loss of generality that $p>0$. If $p=2$, then
  $p=\rt{2}$ is the desired prime factorization by
  Lemma~\ref{lem-2}. If $p\equiv 3,5\mmod{8}$, then $p$ is already
  prime in $\Z[\sqrt2]$ by Lemma~\ref{lem-3-5}. The remaining case is
  when $p\equiv 1,7\mmod{8}$. In this case the prime factorization is
  of the form $p\sim\xi\bul\xi$ by Lemma~\ref{lem-1-7}. Moreover, the
  proof of Lemma~\ref{lem-1-7} indicates how such $\xi$ can be
  computed, namely as $\xi=\gcd(p,x+\sqrt2)$, where $x$ is a solution
  of $x^2\equiv 2\mmod{p}$. The equation $x^2\equiv 2\mmod{p}$ can be
  solved in probabilistic polynomial time by a well-known algorithm,
  see {\cite{Rabin1980}}.
\end{proof}

\subsection[{Primes in Z[omega]}]{Primes in $\Z[\omega]$}

\begin{lemma}
  Let $\xi$ be a prime in $\Z[\sqrt2]$. Then either $\xi$ is prime in
  $\Z[\omega]$, or else $\xi\sim t\da t$ where $t$ is some prime of
  $\Z[\omega]$. In particular, the prime factorization
  of $\xi$ in $\Z[\omega]$ consists of either one or two factors. 
\end{lemma}

\begin{proof}
  This is similar to the proof of Lemma~\ref{lem-prime-factor-z2}.
  Let $t$ be some prime factor of $\xi$ in $\Z[\omega]$. Then
  $t\divides \xi$, therefore $t\da\divides \xi\da=\xi$, therefore
  $t\da t\divides \xi^2$ in $\Z[\omega]$. But both $t\da t$ and $\xi$
  are elements of $\Z[\sqrt2]$, so $t\da t\divides \xi^2$ in
  $\Z[\sqrt2]$. Since $\xi$ is prime in $\Z[\sqrt2]$, there are only
  three possibilities: $t\da t\sim 1$, $t\da t\sim \xi$, and $t\da
  t\sim \xi^2$. In the first case, $t$ is a unit, contradicting the
  fact that it is prime in $\Z[\omega]$. In the second case, we are
  done. In the third case, since $t\divides \xi$, there is
  $a\in\Z[\omega]$ such that $at=\xi$. Then $a\da at\da
  t=\xi\da\xi=\xi^2 = t\da t$, which implies $a\da a=1$. Therefore $a$
  is a unit, and $t\sim \xi$. Therefore $\xi$ is prime in $\Z[\omega]$
  and we are done.
\end{proof}

\subsection[{The Diophantine equation t* t = xi}]
{The Diophantine equation $t\da t=\xi$}

We are interested in solving equations of the form 
\begin{equation}\label{eqn-tdat}
t\da t=\xi,
\end{equation}
where $\xi\in\D[\sqrt2]$ is given, and $t\in\D[\omega]$ is unknown.

\begin{definition}
  Recall that for elements $\xi,\xi'\in\Z[\sqrt2]$, the notation
  $\xi\sim\xi'$ means that $\xi,\xi'$ differ by a unit, i.e., there
  exists a unit $u\in\Z[\sqrt2]$ such that $\xi=u\xi'$. We extend this
  notation also to the ring $\D[\sqrt2]$: for $\xi,\xi'\in\D[\sqrt2]$,
  we write $\xi\sim\xi'$ iff there exists a unit $u\in\Z[\sqrt2]$ such
  that $\xi=u\xi'$. Note that we have taken $u$ to be a unit of the
  ring $\Z[\sqrt2]$, not of $\D[\sqrt2]$.
\end{definition}

It will often be convenient to replace {\eqref{eqn-tdat}} by the
following weaker condition. 

\begin{definition}
  We say that $\xi\in\D[\sqrt2]$ is \emph{$\dagger$-decomposable} if
  the equation
  \begin{equation}\label{eqn-tdat-associates}
    t\da t\sim\xi
  \end{equation} 
  has a solution $t\in\D[\omega]$.
\end{definition}

Solutions to {\eqref{eqn-tdat}} can be recovered from solutions to
{\eqref{eqn-tdat-associates}} by using the following observation:

\begin{lemma}\label{lem-sim}
  Let $\xi\in\D[\sqrt2]$. Then equation {\eqref{eqn-tdat}} has a
  solution if and only if $\xi$ is doubly positive and
  $\dagger$-decomposable.
\end{lemma}

\begin{proof}
  If $\xi$ is a solution to {\eqref{eqn-tdat}}, then it is obviously
  $\dagger$-decomposable. It is also doubly positive by
  Lemma~\ref{lem-nec-cond}. Conversely, assume $t\da t\sim
  \xi$. Then there exists a unit $u$ of $\Z[\sqrt2]$ such that
  $\xi=ut\da t$. Since both $\xi$ and $t\da t$ are doubly positive, it
  follows that $u$ is doubly positive, and hence $u$ is a square of
  the ring $\Z[\sqrt2]$ by Lemma~\ref{lem-units}; say $u=v^2$. Since
  $v\in\Z[\sqrt2]$, we have $v=v\da$. Setting $t'=vt$, we have
  $\xi=v\da vt\da t = t'^{\dagger}t'$, which finishes the proof.
\end{proof}

\subsection[{The case xi in Z[root 2]}]
{The case $\xi\in\Z[\sqrt 2]$}

\begin{lemma}\label{lem-zomega-domega}
  Suppose $t\da t\sim \xi$ for $t\in\D[\omega]$ and $\xi\in\Z[\sqrt
  2]$. Then $t\in\Z[\omega]$. 
\end{lemma}

\begin{proof}
  Note that, in $\Z[\omega]$, we have $\rt{k}\divides t$ if and only
  if $2^k\divides t\da t$. Choose $t'\in\Z[\omega]$ and $k\geq 0$ such
  that $t=t'/\rt{k}$. Then $t'^{\dagger}t' = 2^k\xi$, hence
  $2^k\divides t'^{\dagger}t'$, hence $\rt{k}\divides t'$, hence
  $t\in\Z[\omega]$. 
\end{proof}

\begin{lemma}\label{lem-gcd}
  If $x,y,z$ are three elements of a Euclidean domain, then
  \[\gcd(xy,z)\divides \gcd(x,z)\cdot\gcd(y,z).
  \]
\end{lemma}

\begin{proof}
  By considering the prime factorization of $z$. 
\end{proof}

\begin{lemma}\label{lem-alpha-beta}
  Suppose $\xi=\alpha\beta$, where $\alpha,\beta\in\Z[\sqrt2]$ and
  $\gcd(\alpha,\beta)=1$. Then $\xi$ is $\dagger$-decomposable iff
  $\alpha$ and $\beta$ are $\dagger$-decomposable.
\end{lemma}

\begin{proof}
  For the right-to-left implication, assume $\alpha\sim s\da s$ and
  $\beta\sim r\da r$. Clearly $\xi=\alpha\beta\sim (sr)\da sr$. For
  the left-to-right implication, assume $\xi\sim t\da t$. Note that
  $t\in\Z[\omega]$ by Lemma~\ref{lem-zomega-domega}. To show that
  $\alpha$ is $\dagger$-decomposable, let $s=\gcd(t, \alpha)$.  We
  claim that $s\da s\sim\alpha$. Clearly, since $s\divides\alpha$ and
  $s\da\divides\alpha$, we have $s\da s\divides \alpha^2$. One the
  other hand, we know that $s\da s\divides t\da t \sim
  \xi=\alpha\beta$. Since $\alpha$ and $\beta$ are relatively prime,
  it follows that $s\da s\divides \alpha$. Conversely, note that by
  Lemma~\ref{lem-gcd}, $\alpha = \gcd(t\da t,\alpha) \divides
  \gcd(t\da,\alpha)\cdot \gcd(t,\alpha) = s\da s$. Therefore $s\da
  s\sim\alpha$ and $\alpha$ is $\dagger$-decomposable. The argument
  for $\beta$ is similar.
\end{proof}

\begin{lemma}\label{lem-1235}
  Suppose that $\xi\in\Z[\sqrt2]$ is prime. Let $p>0$ be the unique
  positive prime in $\Z$ such that $\xi\divides p$. Then $\xi$ is
  $\dagger$-decomposable if and only if $p=2$ or $p\equiv
  1,3,5\mmod{8}$.
\end{lemma}

\begin{proof}
  We consider each case in turn. If $p=2$, then $\xi\sim\sqrt2$. Let
  $\delta=1+\omega$; a simple calculation shows that $\delta\da
  \delta=\lambda\sqrt2\sim\xi$. So $\xi$ is $\dagger$-decomposable.

  If $p\equiv 1\mmod{4}$, then by quadratic reciprocity, there exists
  some integer $u$ such that $u^2\equiv -1\mmod{p}$. Therefore
  $\xi\divides p\divides u^2+1 = (u+i)(u-i)$. Let $t=\gcd(\xi,u+i)$.
  We claim that $\xi\sim t\da t$. Note that $t\divides \xi$, hence
  $t\da\divides\xi\da=\xi$, hence $t\da t\divides \xi^2$. Since $\xi$
  is prime in $\Z[\sqrt2]$, there are 3 possibilities: $t\da t\sim 1$,
  $t\da t\sim \xi$, or $t\da t\sim \xi^2$. The first case is not
  possible, because in this case, $t$ would be a unit, so that $\xi$
  is relatively prime to $u+i$, hence to $u-i$, hence to $u^2+1$,
  contradicting $\xi\divides u^2+1$. In the second case, we have $t\da
  t\sim \xi$, which was to be shown. In the third case, we have $t\da
  t\sim\xi^2$. Since $t\divides \xi$, there exists some
  $s\in\Z[\omega]$ such that $ts=\xi$. But then $t\da ts\da s =
  \xi\da\xi=\xi^2\sim t\da t$, so that $s$ is a unit. In this case, we
  have $t\sim \xi$, therefore $\xi\divides u+i$, therefore also
  $\xi=\xi\da\divides u-i$, hence $\xi\divides (u+i)-(u-i)\sim 2$,
  contradicting the fact that $p$ is the only prime integer divisible
  by $\xi$.

  The case $p\equiv 3\mmod{8}$ is very similar. In this case, by
  quadratic reciprocity, $-2$ is a square modulo $p$, so that there
  exists some integer $u$ such that $u^2\equiv-2\mmod{p}$. Therefore
  $\xi\divides p\divides u^2+2 = (u+i\sqrt2)(u-i\sqrt2)$. Let
  $t=\gcd(\xi,u+i\sqrt2)$.  We claim that $\xi\sim t\da t$. Note that
  $t\divides\xi$, hence $t\da\divides\xi\da=\xi$, so $t\da t\divides
  \xi^2$. So again we have the three possibilities $t\da t\sim 1$,
  $t\da t\sim \xi$, or $t\da t\sim \xi^2$. The first case is not
  possible, because in this case, $t$ would be a unit, so that $\xi$
  is relatively prime to $u+i\sqrt2$, hence to $u-i\sqrt2$, hence to
  $u^2+2$, contradicting $\xi\divides u^2+2$. In the second case, we
  have $t\da t\sim \xi$, which was to be shown. In the third case, we
  have $t\da t\sim\xi^2$. Since $t\divides \xi$, there exists some
  $s\in\Z[\omega]$ such that $ts=\xi$. But then $t\da ts\da s =
  \xi\da\xi=\xi^2\sim t\da t$, so that $s$ is a unit. In this case, we
  have $t\sim \xi$, therefore $\xi\divides u+i\sqrt2$, therefore also
  $\xi=\xi\da\divides u-i\sqrt2$, hence $\xi\divides
  i(u-i\sqrt2)-i(u+i\sqrt2)=2\sqrt2$. On the other hand, $\xi\divides
  p$, therefore $\xi\divides\gcd(p,2\sqrt2)=1$, contradicting the fact
  that $\xi$ is not a unit.

  Finally, if $p\equiv 7\mmod{8}$, then $p\sim\xi\bul\xi$ by
  Lemma~\ref{lem-1-7}. Assume that $\xi$ is $\dagger$-decomposable as
  $\xi\sim t\da t$. Then
  \[ (t\da t)\bul(t\da t) \sim \xi\bul\xi\sim p.
  \]
  Note that $t\bul t\in\Z[i]$, so we can write $t\bul t=a+bi$ for some
  $a,b\in\Z$. But then we have
  \[ p \sim (t\bul t)(t\bul t)\da = (a+bi)(a-bi) = a^2+b^2.
  \]
  But $a^2$ and $b^2$ can only be congruent to $0$, $1$, or $4$ modulo
  8, contradicting $p\equiv 7\mmod{8}$. Therefore $\xi$ is not
  $\dagger$-decomposable, which is what was claimed.
\end{proof}

\begin{lemma}\label{lem-1235-powers}
  Suppose $\xi\in\Z[\sqrt2]$ is prime. Let $p>0$ be the unique
  positive prime in $\Z$ such that $\xi\divides p$. Let $m$ be a
  positive integer. Then $\xi^m$ is $\dagger$-decomposable if and only
  if $m$ is even or $p\equiv 1,2,3,5\mmod{8}$.
\end{lemma}

\begin{proof}
  If $m$ is even, then note that $\xi\in\Z[\sqrt2]$, so $\xi=\xi\da$;
  therefore, $t=\xi^{m/2}$ will be a solution. If $p\equiv
  1,2,3,5\mmod{8}$, then by Lemma~\ref{lem-1235}, there exists
  $s\in\Z[\omega]$ with $s\da s\sim\xi$; therefore $t=s^m$ is a
  solution of $t\da t\sim\xi^m$. The only remaining case is then $m$
  is odd and $p\equiv 7\mmod{8}$. In this case, there can be no
  solution. For assume on the contrary that $t\da t=\xi^m$. Then as in
  the proof of Lemma~\ref{lem-1235}, we can write $t\bul t=a+bi$ for
  some $a,b\in\Z$, and we get
  \[ p^m \sim (t\bul t)(t\bul t)\da = (a+bi)(a-bi) = a^2+b^2,
  \]
  contradicting $p^m\equiv 7\mmod{8}$.
\end{proof}

\begin{remark}\label{rem-1235-eff}
  The proofs of Lemmas~\ref{lem-1235} and {\ref{lem-1235-powers}} are
  constructive, and immediately yield efficient algorithms for
  determining $t$, provided that we have an efficient method of
  solving $u^2\equiv-1\mmod{p}$ when $p$ is a prime congruent to
  $1\mmod{4}$ and of solving $u^2\equiv-2\mmod{p}$ when $p$ is a prime
  congruent to $3\mmod{8}$. These last problems can be solved in
  probabilistic polynomial time by a well-known algorithm, see
  {\cite{Rabin1980}}.
\end{remark}

\begin{lemma}\label{lem-solve}
  Given $\xi\in\Z[\sqrt2]$, together with its prime factorization in
  $\Z[\sqrt2]$, there exists an algorithm that determines, in
  probabilistic polynomial time, whether the equation $t\da t\sim\xi$ has
  a solution or not, and finds a solution if there is one.
\end{lemma}

\begin{proof}
  Let $\xi\sim \xi_1^{m_1}\xi_2^{m_2}\cdots\xi_k^{m_k}$ be a prime
  factorization of $\xi$ in $\Z[\sqrt2]$, where $\xi_1,\ldots,\xi_k$
  are distinct (i.e., pairwise non-associate) primes. By
  Lemma~\ref{lem-alpha-beta}, $t\da t\sim\xi$ has a solution if and
  only if $t\da t\sim\xi_i^{m_i}$ has a solution for all $i$.  By
  Lemma~\ref{lem-1235-powers}, $t\da t\sim\xi_i^{m_i}$ has a solution
  if and only if $m_i$ is even or $p\equiv 1,2,3,5\mmod{8}$. Since
  these conditions are easy to check, we can therefore determine the
  existence of a solution efficiently, i.e., in polynomial time.

  Moreover, once a solution has been shown to exist, the actual
  solution can be efficiently computed. Namely, for each $i$, find
  $t_i$ such that $t_i\da t_i\sim\xi_i^{m_i}$ as in
  Lemma~\ref{lem-1235-powers}. Then $t\sim t_1t_2\cdots t_k$ satisfies
  $t\da t\sim \xi_1^{m_1}\xi_2^{m_2}\cdots\xi_k^{m_k}\sim\xi$. By
  Remark~\ref{rem-1235-eff}, all of this can be computed in
  probabilistic polynomial time.
\end{proof}

\begin{proposition}\label{prop-diophantine}
  Let $\xi\in\Z[\sqrt2]$, and let $n=\xi\bul\xi$. Note that $n$ is an
  integer. Given the prime factorization of $n$, there exists an
  algorithm that determines, in probabilistic polynomial time, whether
  the equation $t\da t\sim\xi$ has a solution or not, and finds a
  solution if there is one.
\end{proposition}

\begin{proof}
  Suppose that $n=\pm p_1^{m_1}\cdots p_k^{m_k}$ is the prime
  factorization of $n$, where $p_1,\ldots,p_k$ are distinct positive
  primes.  Each $p_i$ can be efficiently factored into primes of
  $\Z[\sqrt2]$ by Lemma~\ref{lem-prime-factorization-z2}; this yields
  the prime factorization of $n=\xi\bul\xi$ in $\Z[\sqrt2]$. From
  this, it is easy to obtain a prime factorization of $\xi$. The rest
  of the claim then follows from Lemma~\ref{lem-solve}.
\end{proof}

\subsection[{The case xi in D[root 2]}]
{The case $\xi\in\D[\sqrt 2]$}

The following lemma can be used to reduce the problem of
$\dagger$-decomposability in $\D[\sqrt2]$ to $\dagger$-decomposability
in $\Z[\sqrt2]$.

\begin{lemma}\label{lem-Z-to-D-new}
  Consider $\xi\in\D[\sqrt 2]$. Then $\xi$ is $\dagger$-decomposable
  if and only if $\sqrt2\,\xi$ is $\dagger$-decomposable.
\end{lemma}

\begin{proof}
  Recall that $\delta=1+\omega$ satisfies $\delta\da\delta=\lambda\sqrt2 \sim
  \sqrt2$. Also note that $\delta$ is invertible in $\D[\omega]$;
  specifically, $\delta\inv = \delta\lambda\inv\omega\inv/\sqrt2 \in\D[\omega]$.
  Assume that $\xi$ is $\dagger$-decomposable, so that $t\da
  t\sim\xi$. Letting $t'=\delta t$, we have
  $t'^{\dagger}t'=\delta\da\delta t\da t \sim \sqrt2 t\da t\sim
  \sqrt2\,\xi$, so that $\sqrt2\,\xi$ is $\dagger$-decomposable.
  The converse is proved similarly, using $\delta\inv$ instead of
  $\delta$. 
\end{proof}

We can now prove Theorem~\ref{thm-diophantine}, whose statement we 
reproduce here.

\begin{un-theorem}
  Let $\xi\in\D[\sqrt 2]$. Note that $\xi\bul\xi\in\D$, so we can
  write $\xi\bul\xi=\frac{n}{2^\ell}$ for some $n\in\Z$ and $\ell\in\N$. 
  There exists a probabilistic algorithm which, given $\xi$ and, in 
  case $n\neq 0$, a prime factorization of $n$, determines whether or 
  not the equation $t\da t =\xi$ has a solution, and finds a 
  solution if there is one. Moreover, the expected runtime of this 
  algorithm is polynomial in the size of $n$.
\end{un-theorem}

\begin{proof}
  Let $\xi$ be an element of $\D[\sqrt 2]$ with $\xi\bul\xi=
  \frac{n}{2^\ell}$ for some $n\in\Z$. First, the algorithm can easily
  check whether $\xi$ is doubly positive, and if it is not, there is
  no solution. Next, if $n=0$, then $\xi=0$, and $t=0$ is obviously a
  solution, so there is nothing to do. Otherwise, let $\xi'= {\rt
    \ell}\xi$. Since $\xi'\in \Z[\sqrt 2]$ and $\xi'^{\bullet} \xi'=n$, and
  a factorization of $n$ is given, we can use
  Proposition~\ref{prop-diophantine} to efficiently determine whether
  $s\da s\sim \xi'$ has a solution. By
  Lemma~\ref{lem-Z-to-D-new} this is the case if and only if the equation
  $t\da t\sim\xi$ also has a solution, in which case it is given by
  $t=\delta^{-\ell} s$. Finally, since $\xi$ is doubly positive, we can
  solve $t'^{\dagger}t'=\xi$ by Lemma~\ref{lem-sim}.
\end{proof}

We conclude this appendix with a useful fact: the algorithm of
Theorem~\ref{thm-diophantine} always succeeds in case $n$ is a prime
that is congruent to $1$ modulo $8$.

\begin{proposition}\label{prop-prime-1mod8}
  Let $\xi=\D[\sqrt 2]$ be doubly positive, with $\xi\bul\xi=
  \frac{n}{2^\ell}$ for some $n\in\Z$. If $n$ is prime and $n\equiv
  1\mmod{8}$, then the equation $t\da t=\xi$ has a solution.
\end{proposition}

\begin{proof}
  Let $\xi'={\rt \ell}\xi$. Then $\xi'\in\Z[\sqrt 2]$ and
  $\xi'^{\bullet}\xi'=n$. Then $\xi'$ is prime by
  Lemma~\ref{lem-prime-criterion} and $\dagger$-decomposable by
  Lemma~\ref{lem-1235}. By Lemma~\ref{lem-Z-to-D-new}, $\xi$ is
  $\dagger$-decomposable, so that the equation $t\da t=\xi$ can be
  solved by Lemma~\ref{lem-sim}.
\end{proof}

\section[Proof of Lemma 8.4]{Proof of Lemma~\ref{lem-n}}
\label{appendix-n}

\begin{definition}
  Recall that $\delta=1+\omega$.
  Every element $u\in\D[\omega]$ can be written in the form
  \begin{equation}\label{eqn-delta-exponent}
    u = \frac{1}{\delta^k} (a\omega^3+b\omega^2+c\omega+d),
  \end{equation}
  where $a,b,c,d\in\Z$ and $k\geq 0$.
  The smallest $k$ such that $u$ can be written in this form is called
  the {\em least $\delta$-exponent} of $u$.
\end{definition}

\begin{remark}\label{rem-delta-exp}
  A calculation shows that
  $\frac{1}{\delta}(a\omega^3+b\omega^2+c\omega+d)$ is equal to
  \[ \frac12\left[\,(a-b+c-d)\omega^3 
    + (a+b-c+d)\omega^2 
    + (-a+b+c-d)\omega 
    + (a-b+c+d)\,\right].
  \]
  It follows that an element
  $a\omega^3+b\omega^2+c\omega+d\in\Z[\omega]$ is divisible by
  $\delta$ if and only if $a+b+c+d$ is even.
\end{remark}

We can now prove Lemma~\ref{lem-n}, whose statement we 
reproduce here.

\begin{un-lemma}
  Each of the numbers $n$ produced in step 2(a) of
  Algorithm~\ref{alg-main} satisfies $n\geq 0$, and either $n=0$ or
  $n\equiv1\mmod{8}$.
\end{un-lemma}

\begin{proof}
  Recall that in step 2(a) of Algorithm~\ref{alg-main}, we are given
  $u\in\D[\omega]$ such that $u\in\Disk$ and $u\bul\in\Disk$. We let
  $\xi=1-u\da u$ and write $\xi\bul\xi=\frac{n}{2^\ell}$, where
  $\ell\geq 0$ is minimal. We must show that $n\geq 0$, and that either
  $n=0$ or $n\equiv 1\mmod{8}$.

  The first claim is trivial, since by assumption $u\da u\leq 1$ and
  $(u\da u)\bul\leq 1$, and therefore $\xi,\xi\bul\geq 0$, which
  implies $n\geq 0$. For the second claim, write $u$ in the form
  {\eqref{eqn-delta-exponent}}, with least $\delta$-exponent $k$.
  
  \begin{itemize}
  \item Case 1: $k\leq 1$. In this case, one can show by a direct
    calculation (for example with the help of the algorithm from
    Proposition~\ref{prop-scaled1}) that the scaled two-dimensional
    grid problem only has the 9 solutions $u=0$ and $u=\omega^j$,
    where $j=0,\ldots,7$. In the first case, $n=1$, and in the
    remaining 8 cases, $n=0$, so the claim follows.
  \item Case 2: $k\geq 2$. We calculate
    \[ u\da u = \frac{1}{(\delta\da\delta)^k}(A+B\sqrt 2),
    \] 
    where $A=a^2+b^2+c^2+d^2$ and $B=cd+bc+ab-da$. Note that, since
    $k$ is the least $\delta$-exponent of $u$,
    Remark~\ref{rem-delta-exp} implies that $a+b+c+d$ is odd. It
    easily follows that $A$ is odd; moreover, since $B\equiv
    (a+c)(b+d)\mmod{2}$, it also follows that $B$ is even.  Also note
    that $(\delta\da\delta)^k = (\lambda\sqrt2)^k$ is an element
    of $\Z[\sqrt2]$, and is divisible by $2$ since $k\geq
    2$. Therefore, we have $(\delta\da\delta)^k = C+D\sqrt2$ for
    some $C,D\in\Z$ where $C,D$ are both even. We further calculate
    \[ \xi = 1-u\da u =
    \frac{1}{(\delta\da\delta)^k}(C+D\sqrt2 - A -
    B\sqrt 2)
    = \frac{1}{(\delta\da\delta)^k}(x+y\sqrt2),
    \]
    where $x=C-A$ is odd and $y=D-B$ is even. Noting that
    $(\delta\da\delta)\bul(\delta\da\delta)=2$ and that
    $x^2-2y^2\equiv 1\mmod{8}$, we therefore have
    \[ \xi\bul\xi = 
    \frac{1}{((\delta\da\delta)\bul(\delta\da\delta))^k}(x^2-2y^2)
    = \frac{1}{2^k}(x^2-2y^2).
    \]
    It follows that $\ell=k$ and $n=x^2-2y^2$, and therefore $n\equiv 1\mmod{8}$,
    which was to be shown.\qedhere
  \end{itemize}
\end{proof}

\section[Proof of Lemma 8.6]{Proof of Lemma~\ref{lem-summation}}\label{app-d}

\begin{lemma}\label{lem-calculus}
  \begin{enumerate}\alphalabels
  \item\label{lem-calculus-b} $\sqrt x\geq\ln x$ for all $x>0$.
  \item\label{lem-calculus-c} $3\sqrt x\geq(\ln x)^2$ for all $x\geq 1$.
  \item\label{lem-calculus-d} $(1-\frac{a}{x})^x \leq e^{-a}$, for all
    $x\geq a>0$.
  \end{enumerate}
\end{lemma}

\begin{proof}
  By elementary calculus. For (a) and (b), note that the functions
  $f(x)={\ln x}/{\sqrt x}$ and $g(x)={(\ln x)^2}/{3\sqrt x}$, on the
  given domains, take their maxima at $x=e^2$ and $x=e^4$,
  respectively, and in both cases, the maximum is less than 1. 
  For (c), first note that for all $z\in\R$, $z+1\leq e^z$; the claim
  follows by letting $z=-\frac{a}{x}$.
\end{proof}

We can now prove Lemma~\ref{lem-summation}, whose statement we reproduce 
here.

\begin{un-lemma}
  Let $b>0$ be an arbitrary fixed constant. Then for $a\geq 1$,
  \[ \sum_{x=1}^{\infty} \bigparen{1-\frac{1}{a+b\ln x}}^x = O(a).
  \]
\end{un-lemma}

\begin{proof}
  Let
  \begin{equation}
    x_0 = 16a^2 + 144b^2.\label{eqn-x0}
  \end{equation}
  We claim that for all $x \geq x_0$, 
  \begin{equation}
    \bigparen{1-\frac{1}{a+b\ln x}}^x \leq \frac{1}{x^2}. \label{eqn-1x2}
  \end{equation}
  Indeed, from $x\geq 16a^2$, we get
  \begin{equation}\label{eqn-rx22a}
    \frac{\sqrt{x}}{2} \geq 2a.
  \end{equation}
  From $x\geq 144b^2$, we get 
  \begin{equation}\label{eqn-rx62b}
    \frac{\sqrt{x}}{6} \geq 2b.
  \end{equation}
  Combining {\eqref{eqn-rx22a}} and {\eqref{eqn-rx62b}} with
  Lemma~\ref{lem-calculus}(\ref{lem-calculus-b}) and
  (\ref{lem-calculus-c}), we have
  \[
    x = \frac{\sqrt x}{2}\cdot \sqrt x + \frac{\sqrt
      x}{6}\cdot 3\sqrt x
    \geq 2a\ln x + 2b(\ln x)^2 = 2\ln x(a+b\ln x),
  \]
  hence
  \[
  \frac{1}{a + b \ln x} \geq \frac{2\ln x}{x},
  \]
  hence
  \[ 
  \bigparen{1-\frac{1}{a + b \ln x}}^x \leq \bigparen{1-\frac{2\ln x}{x}}^x
  \leq e^{-2\ln x} = \frac{1}{x^2}
  \]
  where the final inequality uses
  Lemma~\ref{lem-calculus}(\ref{lem-calculus-d}). This finishes the
  proof of {\eqref{eqn-1x2}}. The lemma now immediately follows,
  because we have 
  \begin{eqnarray}
    \sum_{x=1}^{\infty} \bigparen{1-\frac{1}{a+b\ln x}}^x
    &=& \sum_{x=1}^{\floor{x_0}} \bigparen{1-\frac{1}{a+b\ln x}}^x +
    \sum_{x=\floor{x_0}+1}^{\infty} \bigparen{1-\frac{1}{a+b\ln x}}^x\nonumber\\
    &\leq& \sum_{x=1}^{\floor{x_0}} \bigparen{1-\frac{1}{a+b\ln x_0}}^x +
    \sum_{x=\floor{x_0}+1}^{\infty} \frac{1}{x^2}\nonumber\\
    &\leq& \sum_{x=0}^{\infty} \bigparen{1-\frac{1}{a+b\ln x_0}}^x +
    \sum_{x=1}^{\infty} \frac{1}{x^2}\nonumber \\
    &=& a+b\ln x_0 + \frac{\pi^2}{6}
    ~=~  a+b\ln (16a^2 + 144b^2) + \frac{\pi^2}{6} ~=~ O(a).\nonumber
  \end{eqnarray}
\end{proof}

\bibliographystyle{abbrv}
\bibliography{gridsynth}

\end{document}